\documentclass[11pt]{article}
\usepackage[utf8]{inputenc}
\usepackage{enumerate}
\usepackage{url} 
\usepackage{subcaption}
\usepackage{float}
\usepackage{setspace}
\usepackage[title]{appendix}
\usepackage[table,xcdraw,dvipsnames]{xcolor}
\usepackage[english]{babel}
\usepackage{soul}
\usepackage[round]{natbib}
\usepackage{booktabs}
\usepackage{setspace}
\usepackage[affil-it]{authblk}

\usepackage{titlesec}
\usepackage[ruled,vlined, linesnumbered]{algorithm2e} 
\usepackage{algpseudocode}
\usepackage{parskip}
\usepackage{dsfont}
\usepackage{enumitem}
\usepackage[margin=1.6in]{geometry}
\usepackage{multibib}
\newcites{app}{Appendix References}

\usepackage[colorlinks,citecolor=blue,urlcolor=magenta,bookmarks=false,hypertexnames=true]{hyperref} 
\usepackage{xspace}
\usepackage{multirow}
\usepackage{amsthm}
\usepackage{amssymb}
\usepackage{amsmath}
\usepackage{amsfonts}
\usepackage{thmtools}
\usepackage{mathrsfs}
\usepackage{mathtools}
\usepackage{bm}
\usepackage{bbm}
\usepackage{hyphenat}
\usepackage{caption,subcaption}
\usepackage{tikz}
\usepackage[T1]{fontenc}
\usepackage{lmodern}
\usepackage{graphicx}
\usepackage{blindtext}
\usepackage{verbatim}
\usepackage{setspace}
\usepackage[title]{appendix}
\usepackage[linesnumbered,ruled,vlined]{algorithm2e}
\SetKwInOut{Input}{input}
\usepackage{algpseudocode}
\usepackage{placeins}
\numberwithin{equation}{section}
\makeatother

\SetAlgoVlined
\SetAlgoInsideSkip{12mm}

\newtheorem{thm}{Theorem}[section]
\newtheorem{cor}{Corollary}[thm]
\newtheorem{lemma}[thm]{Lemma}
\newtheorem{prop}{Proposition}[section]
\theoremstyle{definition} 
\newtheorem{rmk}{Remark}[section]
\newtheorem{defn}{Definition}[section]
\newtheorem{example}{Example}[section]
\newtheorem{assumption}{Assumption}[section]

\newcommand{\eps}{\varepsilon}

\newcommand{\abc}[1]{%
  \ifmmode
  \mathsf{#1}%
  \else
  $\mathsf{#1}$%
  \fi
}

\newcommand{\p}{\mathbb{P}}

\newcommand{\potnet}{\textsf{POTNet}\xspace}

\newcommand{\X}{\mathcal{X}}
\newcommand{\R}{\mathbb{R}}
\newcommand{\N}{\mathcal{N}}
\newcommand{\W}{W}

\DeclareMathOperator*{\argmin}{arg\,min}

\newcommand{\xobs}{x^*}
\newcommand{\msw}{\textsf{MSW}\xspace}
\newcommand{\abi}{\textsf{ABI}\xspace}

\newcommand{\cp}{\mathcal{P}}
\newcommand{\B}{\mathcal{B}}

\newcommand{\ind}{\mathds{1}}
\newcommand{\D}{\mathcal{D}}
\newcommand{\SW}{SW}
\newcommand{\qf}{Q_\tau}

\newcommand{\Lip}{\mathrm{Lip}}
\newcommand{\F}{\mathcal{F}}
\newcommand{\smb}{\mathbb{S}}

\newcommand{\E}{\mathbb{E}}

\newcommand{\hqf}{\widehat{Q}_\tau}
\newcommand{\Q}{\mathcal{Q}}

\newcommand{\qnetest}{\widehat{Q}_{K', H}}

\SetAlgoNlRelativeSize{-1}
\SetAlCapSkip{0.5em}
\SetKw{Accept}{Accept}
\newcommand{\tpos}{\pi^*}

\DeclareSymbolFont{sfletters}{OML}{cmbrm}{m}{it}
\DeclareMathSymbol{\slambda}{\mathord}{sfletters}{"15}

\DeclareMathOperator{\Var}{Var}

\newcommand{\ntrain}{\mathsf{N}_{\mathsf{train}}}
\newcommand{\supp}{\mathrm{supp}}
\newcommand{\bX}{X}
\newcommand{\bx}{x}
\newcommand{\bxobs}{\xobs}

\let\emptyset\varnothing

\usepackage{tabularx,ragged2e}
\usepackage{multirow}
\usepackage{makecell}
\usepackage{array}

\title{Likelihood-Free Adaptive Bayesian Inference via \\Nonparametric Distribution Matching}

\author[1]{Wenhui Sophia Lu\thanks{Corresponding author; Electronic address: \texttt{sophialu@stanford.edu}}}

\author[1,2]{Wing Hung Wong}

\affil[1]{Department of Statistics, Stanford University}
\affil[2]{Department of Biomedical Data Science, Stanford University}
\date{May 8, 2025}

\begin{document}
\maketitle

\begin{abstract}
When the likelihood is analytically unavailable and computationally intractable, approximate Bayesian computation (ABC) has emerged as a widely used methodology for approximate posterior inference; however, it suffers from severe computational inefficiency in high-dimensional settings or under diffuse priors.
To overcome these limitations, we propose Adaptive Bayesian Inference (\abi), a framework that bypasses traditional data-space discrepancies and instead compares distributions directly in posterior space through nonparametric distribution matching. 
By leveraging a novel Marginally-augmented Sliced Wasserstein (\msw) distance on posterior measures and exploiting its quantile representation, \abi transforms the challenging problem of measuring divergence between posterior distributions into a tractable sequence of one-dimensional conditional quantile regression tasks. 
Moreover, we introduce a new adaptive rejection sampling scheme that iteratively refines the posterior approximation by updating the proposal distribution via generative density estimation.
Theoretically, we establish parametric convergence rates for the trimmed \msw distance and prove that the \abi posterior converges to the true posterior as the tolerance threshold vanishes.
Through extensive empirical evaluation, we demonstrate that \abi significantly outperforms data-based Wasserstein ABC, summary-based ABC, as well as state-of-the-art likelihood-free simulators, especially in high-dimensional or dependent observation regimes.
\end{abstract}

\noindent
{\bf Keywords:} approximate Bayesian computation; likelihood-free inference; simulator-based inference; conditional quantile regression; nonparametric distribution matching;  adaptive rejection sampling; generative modeling; Wasserstein distance


\section{Introduction}
Bayesian modeling is widely used across natural science and engineering disciplines. It enables researchers to easily construct arbitrarily complex probabilistic models through forward sampling techniques (implicit models) while stabilizing ill-posed problems by incorporating prior knowledge.
Yet, the likelihood function $x \mapsto f_\theta(x)$ may be intractable to evaluate or entirely inaccessible in many scenarios \citep{zeng2019novel, chiachio2021solving}, thus rendering Markov chain-based algorithms---such as Metropolis-Hastings and broader Markov Chain Monte Carlo methods---unsuitable for posterior inference. Approximate Bayesian Computation (ABC) emerges as a compelling approach for scenarios where exact posterior inference for model parameters is infeasible \citep{tavare2018history}. Owing to its minimal modeling assumptions and ease of implementation, ABC has garnered popularity across various Bayesian domains, including likelihood-free inference \citep{markram2015reconstruction, alsing2018massive}, Bayesian inverse problems \citep{chatterjee2021approximate}, and posterior estimation for simulator-based stochastic systems \citep{wood2010statistical}.
ABC generates a set of parameters $\theta \in \Omega \subseteq \R^d$ with high posterior density through a rejection-based process: it simulates fake datasets for different parameter draws and retains only those parameters that yield data sufficiently similar to the observed values.

However, when the data dimensionality is high or the prior distribution is uninformative about the observed data, ABC becomes extremely inefficient and often requires excessive rejections to retain a single sample.
Indeed, Lemmas \ref{lemma:curse-of-dim-gaussian} and \ref{lemma:curse-of-dim-bounded} show that the expected number of simulations needed to retain a single draw grows exponentially in the data dimension.
To enhance computational efficiency, researchers frequently employ low-dimensional summary statistics and conduct rejection sampling instead in the summary statistic space \citep{fearnhead2012constructing}. 
Nevertheless, the Pitman-Koopman-Darmois theorem stipulates that low-dimensional sufficient statistics exist only for the exponential family. 
Consequently, practical problems often require considerable judgment in choosing appropriate summary statistics, typically in a problem-specific manner \citep{wood2010statistical, marin2012approximate}. 
Moreover, the use of potentially non-sufficient summary statistics to evaluate discrepancies can result in ABC approximations that, while useful, may lead to a systematic loss of information relative to the original posterior distribution. For instance, \cite{fearnhead2011constructing} and \cite{jiang2017learning} propose a semi-automatic approach that employs an approximation of the posterior mean as a summary statistic; however, this method ensures only first-order accuracy.

Another critical consideration is selecting an appropriate measure of discrepancy between datasets. A large proportion of the ABC literature is devoted to investigating ABC strategies adopting variants of the $\ell_p$-distance between summaries \citep{prangle2017adapting}, which are susceptible to significant variability in discrepancies across repeated samples from $f_\theta$ \citep{bernton2019approximate}.
Such drawbacks have spurred a shift towards summary-free ABC methods that directly compare the empirical distributions of observed and simulated data via an integral probability metric (IPM), thereby obviating the need to predefine summary statistics \citep{legramanti2022concentration}.
Popular examples include ABC versions that utilize the Kullback-Leibler divergence \citep{jiang2018approximate}, 2-Wasserstein distance \citep{bernton2019approximate}, and Hellinger and Cramer–von Mises distances \citep{frazier2020robust}. 
The accuracy of the resulting approximate posteriors relies crucially on the fixed sample size $n$ of the observed data, 
as the quality of IPM estimation between data-generating processes from a finite, often small, number of samples is affected by the convergence rate of empirically estimated IPMs to their population counterparts.
In particular, a significant drawback of Wasserstein-based ABC methods stems from the slow convergence rate of the Wasserstein distance, which scales as $O(n^{-1/s})$ when the data dimension $s\geq 3$
\citep{talagrand1994transportation}.
As a result, achieving accurate posterior estimates is challenging with limited samples, particularly for high-dimensional datasets. 
A further limitation of sample-based IPM evaluation is the need for additional considerations in the case of dependent data, since ignoring such dependencies might render certain parameters unidentifiable \citep{bernton2019approximate}.

Thus, two fundamental questions ensue from this discourse: What constitutes an informative set of summary statistics, and what serves as an appropriate measure of divergence between datasets? 
To address the aforementioned endeavors, we introduce the \textit{Adaptive Bayesian Inference} (\abi) framework, which directly compares posterior distributions through distribution matching and adaptively refines the estimated posterior via rejection sampling.
At its core, \abi bypasses observation-based comparisons by selecting parameters whose synthetic-data-induced posteriors align closely with the target posterior, a process we term \textit{nonparametric distribution matching}.
To achieve this, \abi learns a discrepancy measure in the posterior space, rather than the observation space, by leveraging the connection between the Wasserstein distance and conditional quantile regression, thereby transforming the task into a tractable supervised learning problem.
Then, \abi simultaneously refines both the posterior estimate and the approximated posterior discrepancy over successive iterations.

Viewed within the summary statistics framework, our proposed method provides a principled approach for computing a model-agnostic, one-dimensional \textit{kernel statistic}. 
Viewed within the discrepancy framework, our method approximates an integral probability metric on the space of \textit{posteriors}, thus circumventing the limitations of data-based IPM evaluations such as small sample sizes and dependencies among observations.

\paragraph{Contributions} 
Our work makes three main contributions.
First, we introduce a novel integral probability metric---the Marginally-augmented Sliced Wasserstein (\msw) distance---defined on the space of \textit{posterior probability measures}. 
We then characterize the \abi approximate posterior as the distribution of parameters obtained by conditioning on those datasets whose induced posteriors fall within the prescribed \msw tolerance of the target posterior.
Whereas conventional approaches rely on integral probability metrics on empirical data distributions, our posterior–based discrepancy remains robust even under small observed sample sizes $n$, intricate sample dependency structures, and parameter non-identifiability.
We further argue that considering the axis-aligned marginals can help improve the projection efficiency of uniform slice-based Wasserstein distances.
Second, we show that the posterior \msw distance can be accurately estimated through conditional quantile regression by exploiting the equivalence between the univariate Wasserstein distance and differences in quantiles.
This novel insight reduces the traditionally challenging task of operating in the posterior space into a supervised distributional regression task, which we solve efficiently using deep neural networks.
The same formulation naturally accommodates multi-dimensional parameters and convenient sequential refinement via rejection sampling.
Third, we propose a sequential version of the rejection--ABC that, to the best of our knowledge, is the first non-Monte-Carlo-based sequential ABC.
Existing sequential refinement methods in the literature frequently rely on adaptive importance sampling techniques, such as sequential Monte Carlo \citep{del2012adaptive, bonassi2015sequential} and population Monte Carlo \citep{beaumont2009adaptive}. 
These approaches, particularly in their basic implementations, are often constrained to the support of the empirical distribution derived from prior samples. 
While advanced variants can theoretically explore beyond this initial support through rejuvenation steps and MCMC moves, they nevertheless require careful selection of transition kernels and auxiliary backward transition kernels \citep{del2012adaptive}.
In contrast, \abi iteratively refines the posterior distribution via rejection sampling by updating the proposal distribution using the \textit{generative} posterior approximation from the previous step---learned through a generative model (not to be confused with the original simulator in the likelihood-free setup).
Generative-model-based approaches for posterior inference harness the expressive power of neural networks to capture intricate probabilistic structures without requiring an explicit distributional specification.
This generative learning stage enables \abi to transcend the constrained support of the empirical parameter distribution and eliminates the need for explicit prior-density evaluation (unlike \cite{papamakarios2016fast}), thereby accommodating cases where the prior distribution itself may be intractable.

We characterize the topological and statistical behavior of the \msw distance, establishing both its parametric convergence rate and its continuity on the space of posterior measures. 
Our proof employs a novel martingale-based argument appealing to Doob's theorem, which offers an alternative technique to existing proofs based on the Lebesgue differentiation theorem \citep{barber2015rate}.
This new technique may be of independent theoretical interest for studying the convergence of other sequential algorithms.
We then prove that, as the tolerance threshold vanishes (with observations held fixed), the \abi posterior converges in distribution to the true posterior. Finally, we derive a finite-sample bound on the bias induced by the approximate rejection-sampling procedure.
Through comprehensive empirical experiments, we demonstrate that \abi achieves highly competitive performance compared to data-based Wasserstein ABC,  and several recent, state-of-the-art likelihood-free posterior simulators.

\paragraph{Notation}
Let the parameter and data $(\theta, X)$ be jointly defined on some probability space.  
The prior probability measure $\pi$ on the parameter space $\Omega \subseteq \R^d$ is assumed absolutely continuous with respect to Lebesgue measure, with density $\pi(\theta)$ for $\theta\in\Omega$.  For simplicity, we use $\pi(\cdot)$ to denote both the density and its corresponding distribution.  
Let the observation space be $\X\subseteq\R^{d_X}$ for some $d_X\in\mathbb{N}_+$, where $\mathbb{N}_+ := \{1,2,\dots\}$.  We observe a data vector 
$\xobs=(\xobs_1,\dots,\xobs_n)^\top\in\X^n\subset\R^{n d_X}$,
whose joint distribution on $\X^n$ is given by the likelihood $P_\theta^{(n)}$.
If the samples are not exchangeable, we simply set $n=1$ with a slight abuse of notation and write $\xobs$ for that single observation. 
We assume $\xobs$ is generated from $P_{\theta^*}^{(n)}$ for some true but unknown $\theta^*\in\Omega$.  
Both the prior density $\pi(\theta)$ and the likelihood $P_\theta^{(n)}(x)$ may be analytically intractable; however, we assume access to  
\begin{itemize}
  \item a \emph{prior simulator} that draws $\theta\sim\pi$, and  
  \item a \emph{data generator} that simulates $X\sim P_\theta^{(n)}$ given any $\theta$.  
\end{itemize}
We do not assume parameter identifiability; that is, we allow for the possibility that distinct parameter values $\theta \neq \theta'$ to yield identical probability distributions, $P_{\theta}^{(n)} = P_{\theta'}^{(n)}$.
Our inferential goal is to generate samples from the posterior $\pi(\theta \mid \xobs) \propto \pi(\theta) P_\theta^{(n)}(\xobs)$, where $\theta \in \Omega$.
For notational convenience, we use $\D(\cdot,\cdot)$ for a generic distance metric, which may act on the data space or probability measures, depending on the context.

For any function class $\mathcal{G}$ and probability measures $\mu$ and $\nu$, we define the Integral Probability Metric (IPM) between $\mu$ and $\nu$ with respect to $\mathcal{G}$ as: $\D_\mathcal{G}(\mu, \nu) = \sup_{g \in \mathcal{G}} \left|\int g d\mu - \int g d\nu\right|$.

Let $\|{\cdot}\|$ denote the $\ell_2$ (Euclidean) distance and let $(\Omega, \|{\cdot}\|)$ be a Polish space.
For $p \in [1, \infty)$, we denote by $\cp_p(\Omega)$ the set of Borel probability measures defined on $\Omega$ with finite $p$-th moment. 
For $\mu, \nu \in \cp_p(\Omega)$, the $p$-Wasserstein distance between $\mu$ and $\nu$ is defined as the solution of the optimal mass transportation problem 
\begin{align}\label{Wpminimization}
    \W_p(\mu, \nu) = \left(\inf_{\gamma \in \Gamma(\mu, \nu)}\int_{\Omega \times \Omega} \|x - y\|^p d\gamma(x, y) \right)^{1/p},
\end{align}
where $\Gamma(\mu, \nu)$ is the set of all couplings $\gamma \in \Gamma(\mu, \nu)$ such that 
\begin{align*}
    \gamma(B_1 \times \Omega) = \mu(B_1), ~\gamma(\Omega \times B_2) = \nu(B_2),  ~\text{for all Borel sets $B_1, B_2\subseteq \Omega$}.
\end{align*}
The $p$-Wasserstein space is defined as $(\cp_p(\Omega), \W_p)$.
For a comprehensive treatment of the Wasserstein distance and its connections to optimal transport, we refer the reader to \citet{villani2009optimal}.

\subsection{Approximate Bayesian Computation}
We begin with a brief review of classic Approximate Bayesian Computation (ABC). Given a threshold $\epsilon > 0$, a distance $\D(\cdot,\cdot)$ on summary statistics $s(\cdot)$, classic ABC produces samples from the approximate posterior,
\begin{align*}
    \pi^\epsilon_{\mathrm{ABC}}(\theta \mid \xobs) \propto \pi(\theta) \int_{\X^n} \ind\big[\D\big(s(x), s(\xobs)\big) \leq \epsilon \big] \, dP^{(n)}_\theta(x), 
\end{align*}
via the following procedure:

\begin{algorithm}[!htp]
\caption{Rejection-ABC Algorithm}
\For{$i = 1, 2, \dots, \mathsf{N}$}{
    Simulate $\theta^{(i)} \sim \pi(\theta)$,  $X^{(i)} = (X_1^{(i)}, X_2^{(i)}, \dots, X_n^{(i)}) \sim   P_{\theta^{(i)}}^{(n)}$\,
    \Accept $\theta^{(i)}$ if $\D\big(s(X^{(i)}), s(\xobs)\big) \leq \epsilon$\,
}
\end{algorithm}
For results on convergence rates and the bias–cost trade-off when using sufficient statistics in ABC, see \citet{barber2015rate}, who establish consistency of ABC posterior expectations via the Lebesgue differentiation theorem.

\subsection{Sliced Wasserstein Distance}
The Sliced Wasserstein (SW) distance, introduced by \cite{rabin2012wasserstein}, provides a computationally efficient approximation to the Wasserstein distance in high dimensions. 
For measures $\mu$ and $\nu$ on $\R^d$, the $p$-Sliced Wasserstein distance integrates the $p$-th power of the one-dimensional Wasserstein distance over all directions on the unit sphere:
\begin{align}\label{def:sw}
    \SW^p_p(\mu,\nu) = \int_{\mathbb{S}^{d-1}} \W_p^p(\varphi_\#\mu, \varphi_\#\nu) \,d\sigma(\varphi),
\end{align}
where $\mathbb{S}^{d-1}$ is the unit sphere in $\R^d$, $\sigma$ is the uniform measure on $\mathbb{S}^{d-1}$, $\varphi_\#$ denotes the projection onto the one-dimensional subspace spanned by $\varphi$. 
By reducing the problem to univariate cases, each of which admits an analytic solution, this approach circumvents the high computational cost of directly evaluating the $d$-dimensional Wasserstein distance while preserving key topological properties of the classical Wasserstein metric, including its ability to metrize weak convergence \citep{bonnotte2013unidimensional}. 

To approximate the integral in \eqref{def:sw}, in practice, one draws $K$ directions i.i.d.\ from the sphere and forms the unbiased Monte Carlo estimator  
\begin{align}
    \widehat{\SW}_p(\mu,\nu) = \left(\frac{1}{K} \sum_{k=1}^K \W_p^p(\varphi^{(k)}_\#\mu, \varphi^{(k)}_\#\nu)\right)^{1/p},
\end{align}
where each $\varphi^{(k)}\sim \mathrm{Unif}(\smb^{d-1})$.

\subsection{Conditional Quantile Regression}
Drawing on flexible, distribution-free estimation methods, we briefly review conditional quantile regression. 
Introduced by \citet{koenker1978regression}, quantile regression offers a robust alternative to mean response modeling by estimating conditional quantiles of the response variable. 
For response $Y \in \R$ and covariates $X \in \R^d$, the $\tau$-th conditional quantile of $Y$ given $X$ is
\begin{align}
    \qf(x) = \inf\{y \in \R : F_{Y\mid X}(y \mid x) \geq \tau\}, \quad x \in \R^d,
\end{align}
where $F_{Y\mid X}(\cdot \mid x)$ is the conditional CDF of $Y$ given $X = x$. 
We estimate $\qf$ by minimizing the empirical quantile loss over a model class $\Q$:
\begin{align}
    \hqf = \argmin_{Q \in \Q} \sum_{i=1}^{\ntrain} \rho_\tau(y_i - Q(x_i)),
\end{align}
where $\rho_\tau(u) = \max\{\tau u, (\tau - 1)u\}$ is the quantile loss function, and $\ntrain$ denotes the number of training samples. 
We use $\ntrain$ deliberately to avoid confusion with $n$, which represents the (fixed) number of observations in Bayesian inference.

\subsection{Generative Density Estimation}

Consider a dataset $Y = \{Y_1,\dots,Y_n\}$ of i.i.d.\ draws from an unknown distribution $P_Y$.  
Generative models seek to construct a distribution $\hat{P}_Y$ that closely approximates $P_Y$ and is amenable to efficient generation of new samples.
Typically, one assumes an underlying latent variable structure whereby samples are generated as $\widetilde{Y} = G_\beta(Z)$,
with $Z \sim P_Z$ being a low-dimensional random variable with a simple, known distribution. 
Here $G_\beta\colon\mathcal Z\to\mathcal Y$ is the generative (push-forward) map that transforms latent vectors into observations.
The parameters $\beta$ are optimized so that the generated samples $\tilde{y} \sim \hat{P}_Y$ are statistically similar to the real data. 
Learning objectives may be formulated using a variety of generative frameworks, including optimal transport networks \citep{lu2025generative}, generative adversarial networks, and auto-encoder models.

\subsection{Article Structure and Related Literature}

\paragraph{Related Works}
Several prior works have proposed the usage of generative adversarial networks (GANs) as conditional density estimators \citep{zhou2023deep, wang2022adversarial}. These approaches aim to learn a generator mapping $(X, \eta) \mapsto \widetilde{\theta}$ with $\eta$ independent of $X$ and $\theta$. 
Adversarial training then aims to align the model’s induced joint distribution $P_{\mathrm{data}}(X)\, P_G(\widetilde{\theta} \mid X)$ with the true joint distribution $P_{\mathrm{data}}(X)\, P_{\mathrm{data}}(\theta\mid X)$.
The distinction between \abi and these methods is crucial: \abi operates in the space of posterior distributions on $\Omega$ and measures distances between full posterior distributions, while the latter works in the joint space $\X^n \times \Omega$, using a discriminator loss to match generated pairs $(X, \widetilde{\theta})$ to real data $(X, \theta)$.
Due to this fundamental difference, these generative approaches focus on learning conditional mappings uniformly over the entire domain $\X^n$ in a single round.
While GAN-based methods can be adapted for sequential refinement using importance re-weighting, as demonstrated by \cite{wang2022adversarial}, such adaptations typically require significant additional computational resources, including training auxiliary networks (such as the classifier network needed to approximate ratio weights in their two-step process). In contrast, \abi is inherently designed for sequential refinement without necessitating such auxiliary models.
In particular, results on simulated data in Figure \ref{fig:toy-model-wgan-abc} show that our sequential method significantly outperforms Wasserstein GANs when the prior is uninformative.

Recent work by \cite{polson2023generative} and its multivariate extension by \cite{kim2024deep} propose simulating posterior samples via inverse transform sampling.
Although both methods and \abi employ quantile regression, they differ fundamentally in scope and mechanism.
The former approaches apply a one-step procedure that pushes noise through an inverse-CDF or multivariate quantile map to produce posterior draws, inherently precluding direct sequential refinement.
In contrast, \abi employs conditional quantile regression to estimate a posterior metric---the posterior \msw distance---which extends naturally to any dimension and any $p \in[1,\infty)$. 
In particular, the case $p=1$ is noteworthy, since it renders both $\msw_1$ and $W_1$ distances integral probability metrics (see Theorem \ref{thm:msw-ipm}) and allows a dual formulation \citep{villani2009optimal}.
On the other hand, \cite{kim2024deep} focus exclusively on the 2-Wasserstein case. Moreover, they rely on a combination of Long Short-Term Memory and Deep Sets architectures to construct a multivariate summary and approximate the 2-Wasserstein transport map, a step that the authors acknowledge is sensitive to random initialization for obtaining a meaningful quantile mapping.
By comparison, \abi uses a simple feed-forward network to learn a one-dimensional kernel statistic corresponding to the estimated posterior \msw distance. Our experiments (Section \ref{sec:sim}) demonstrate that \abi is robust across different scenarios, insensitive to initialization, and requires minimal tuning.

\paragraph{Organization of the Paper}

The remainder of this manuscript is organized as follows. 
Section \ref{sec:method} introduces the \abi framework and its algorithmic components.
Section \ref{sec:theory} establishes the empirical convergence rates of the proposed \msw distance, characterizes its topological properties, and proves that the \abi posterior converges to the target posterior as the tolerance threshold vanishes.
Section \ref{sec:sim} demonstrates the effectiveness of \abi through extensive empirical evaluations.
Finally, Section \ref{sec:discussion} summarizes the paper and outlines future research directions.
Proofs of technical results and additional simulation details are deferred to the Appendix.


\section{Adaptive Bayesian Inference}\label{sec:method}
In this section, we introduce the proposed Adaptive Bayesian Inference (\abi) methodology. 
The fundamental idea of \abi is to transcend observation-based comparisons by operating directly in posterior space. Specifically, we approximate the target posterior as
\begin{align*}
\pi_{\abi}(\theta) = \pi \Bigl(\theta \Bigm
 | \msw(\pi(\theta\mid x),\,\pi(\theta\mid\xobs))\leq \epsilon\Bigr),
\end{align*}
that is, the distribution of $\theta$ conditional on the event that the posterior induced by dataset $x$ lies within an $\epsilon$-neighborhood of the observed posterior under the \msw metric.
This formulation enables direct comparison of candidate posteriors via the posterior \msw distance and supports efficient inference by exploiting its quantile representation.
As such, \abi approximates the target posterior through nonparametric posterior matching. 
In practice, for each proposed $\theta$, we simulate an associated dataset $X$ and evaluate the \msw distance between the conditional posterior $\pi(\theta \mid X)$ and the observed-data posterior $\pi(\theta \mid \xobs)$. 
Simulated samples for which this estimated deviation is small are retained, thereby steering our approximation progressively closer to the true posterior.
For brevity, we denote the target posterior by $\tpos = \pi(\theta \mid \xobs)$. Our approach proceeds in four steps.

\begin{enumerate}[label=\textbf{Step \arabic*.}, itemsep=-4pt, leftmargin=*, labelindent=0pt]
\item Estimate the trimmed \msw distance between posteriors $\pi^*$ and $\pi(\theta \mid x)$, $x \in \mathcal{X}$ using conditional quantile regression with multilayer feedforward neural networks; see Section \ref{method:msw-estimation}.
\item Sample from the current proposal distribution by decomposing it into a marginal component over $\theta$ and an acceptance constraint on $X$, then employ rejection sampling; see Section \ref{method:sample-proposal}.
\item Refine the posterior approximation via acceptance--rejection sampling: retain only those synthetic parameter draws whose simulated data yield an estimated \msw distance to $\tpos$ below the specified threshold, and discard the rest; see Section \ref{method:pruning-stage}.
\item Update the proposal for the next iteration by fitting a generative model to the accepted parameter draws; see Section \ref{method:generative-density-estimation}.
\end{enumerate}

These objectives are integrated into a unified algorithm with a nested sequential structure.
At each iteration $t = 1, 2, \dots, \mathsf{T}$, the proposal distribution is updated using the previous step's posterior approximation, thus constructing a sequence of \textit{partial posteriors} $\pi^{(1)}_*, \pi^{(2)}_*, \dots, \pi^{(\mathsf{T})}_*$ that gradually shift toward the target posterior $\tpos$. 
This iterative approach improves the accuracy of posterior approximation through adaptive concentration on regions of high posterior alignment, which in turn avoids the unstable variance that can arise from single-round inference. 
By contrast, direct Monte Carlo estimation would require an infeasible number of simulations to observe even a single instance where the generated sample exhibits sufficient similarity to the observed data. 
For a concrete illustration of sequential refinement, see the simple Gaussian–Gaussian conjugate example in the Appendix \ref{example:gaussian-gaussian}.
The complete procedure is presented in \abc{Algorithm} \ref{alg:abi}.

\begin{algorithm}[!htp]
\caption{Adaptive Bayesian Inference (\abi)}
\label{alg:abi}
\KwIn{
    Tolerance thresholds: $\infty > \epsilon_1 > \epsilon_2 > \cdots > \epsilon_{\mathsf{T}}$\;
    \phantom{\textbf{Inputt:} }Generative model: $G_\beta$\;
    \msw \textbf{distance parameters:}\\
    \phantom{\textbf{Inputt:} }Trimming parameter: $\delta \in (0, 1/2)$\;
    \phantom{\textbf{Inputt:} }Mixing parameter: $\lambda \in (0, 1)$\;
    \phantom{\textbf{Inputt:} }Number of slices: $K > 0$\;
    \phantom{\textbf{Inputt:} }Number of discretization points: $H > 0$\;
}
\vspace{8pt}
Initialize $\pi^{(1)}(\theta) \gets \pi(\theta)$\;
\vspace{5pt}
\For{Iteration $t = 1, 2, \dots, \mathsf{T}$}{
    \vspace{.5em}
    \tcp{Sample the current proposal; see Section \ref{method:sample-proposal}}
    Generate $\mathsf{N}$ parameter-data pairs $\{(\theta^{(i)}, X^{(i)})\}_{i=1}^{\mathsf{N}}$ by rejection sampling from the unnormalized joint proposal using Algorithm \ref{alg:approx-rej-samp}:
    \begin{align*}
        \pi^{(t)}(\theta, X) = \pi\bigl(\theta, X \mid \,\ind\{\widehat{\msw}_{p,\delta,K,H}(\pi_X,\pi_{\xobs})\le\epsilon_{t-1}\}\bigr).
    \end{align*}\\
    Sample $K$ random projections $\varphi^{(k)} \sim \sigma(\smb^{d-1})$ and form the projection set $\Phi = \{\varphi^{(k)}\}_{k=1}^K \cup \{e_j\}_{j=1}^d$; set $K' \gets K + d$\;
    \vspace{.5em}
    \tcp{Train quantile estimator network; see Section \ref{method:msw-estimation}}
    Generate $\mathsf{N}_{\mathsf{train}}$ parameter-data pairs $\{(X^{(m)}, \theta^{(m)})\}_{m=1}^{\mathsf{N}_{\mathsf{train}}}$ using Algorithm \ref{alg:approx-rej-samp} to form the training set for the quantile ReLU network\;
    Fine-tune the conditional quantile ReLU network $\qnetest$ using Algorithm \ref{alg:msw-estimation}\;
    For each $i = 1, \ldots, \mathsf{N}$, evaluate $\widehat{\msw}_{p, \delta, K, H}\bigl(\pi(\theta \mid X^{(i)}), \pi(\theta \mid\xobs)\bigr)$ with $\qnetest$ using Eq.\ \eqref{msw:emp-eval}\;
    \vspace{.5em}
    \tcp{Pruning stage, refine the partial posterior distribution; see Section \ref{method:pruning-stage}}
    Retain parameter draws satisfying the tolerance threshold condition:
    \vspace{-.5em}
    \begin{align*}
        \mathsf{S}^{(t)}_{\theta, *} \gets \{\theta^{(i)}: \widehat{\msw}_{p, \delta, K, H}\bigl(\pi(\theta \mid X^{(i)}), \pi(\theta \mid \xobs)\bigr) \leq \epsilon_t\}
    \end{align*}\\
    \vspace{-.5em}
    \vspace{.5em}
    \tcp{Sequential generative density estimation stage; see Section \ref{method:generative-density-estimation}}
    Train generative density estimator $G_\beta$ on $\mathsf{S}^{(t)}_{\theta, *}$ to obtain the approximate partial posterior $\widehat{\pi}_*^{(t)}$\;
    Update the proposal distribution for iteration $(t+1)$:
    \vspace{-.5em}
    \begin{align*}
        \pi^{(t+1)}(\theta) \gets \widehat{\pi}_*^{(t)}(\theta)
    \end{align*}\\
}
Set $\widehat{\pi}^{(\epsilon_{\mathsf{T}})}_{\abi}(\theta \mid \xobs) \gets \widehat{\pi}_*^{(\mathsf{T})}(\theta)$\;
\vspace{.5em}
\KwOut{Generative model that samples from the \abi posterior: $\widehat{\pi}^{(\epsilon_{\mathsf{T}})}_{\abi}(\theta \mid \xobs)$\;}
\end{algorithm}

\subsection{Nonparametric Distribution Matching}
\subsubsection{Motivation: Minimal Posterior Sufficiency}

We first discuss the conceptual underpinnings underlying our proposed posterior-based distance. 
Consider the target posterior $\pi^*$, which conditions on the observed data $\xobs$.
If an alternative posterior $\pi(\theta \mid X)$ is close to $\pi^*$ under an appropriate measure of posterior discrepancy, then $\pi(\theta \mid X)$ naturally constitutes a viable approximation to the target posterior. Thus, we can select, from among all candidate posteriors, the ones whose divergence from the true posterior falls within the prescribed tolerance. 
We formalize this intuition through the notion of \textit{minimal posterior sufficiency}.

Under the classical frequentist paradigm, Fisher showed that the likelihood function $L(\theta; X) = P^{(n)}_\theta(X)$, viewed as a random function of the data, is a minimal sufficient statistic for $\theta$ as it encapsulates all available information about the parameter $\theta$ \citep[Chapter 3]{berger1988likelihood}.
This result is known as the likelihood principle. 
The likelihood principle naturally extends to the Bayesian regime since the posterior distribution is proportional to $\pi(\theta)P^{(n)}_\theta(X)$. In particular, 
Theorem \ref{thm:minimal-posterior-sufficiency} shows that, given a prior distribution $\pi(\cdot)$, the posterior map $X \mapsto \pi(\cdot \mid X)$ is minimally Bayes sufficient with respect to the prior $\pi$.
We refer to this concept as \textit{minimal posterior sufficiency}.

\begin{defn}[Bayes Sufficient]
    A statistic $T(X)$ is \textit{Bayes sufficient} with respect to a prior distribution $\pi(\theta)$ if $\pi(\theta \mid X) = \pi(\theta \mid T(X))$. 
\end{defn}

\begin{thm}[Minimal Bayes Sufficiency of Posterior Distribution]\label{thm:minimal-posterior-sufficiency}
    The posterior map $X \mapsto \pi(\cdot \mid X)$ is minimally Bayes sufficient. 
\end{thm}

Ideally, inference should be based on minimally sufficient statistics, which suggests that our posterior inference should utilize such statistics when low-dimensional versions exist. However, low-dimensional sufficient statistics---let alone minimally sufficient ones---are available for only a very limited class of distributions; consequently, we need to consider other alternatives to classical summaries.
Observe that the acceptance event formed by matching the infinite-dimensional statistics $\pi(\theta\mid X)$ and $\pi(\theta\mid \xobs)$ coincides with the event,
\begin{align*}
  \ind\bigl\{ \mathcal{D}\bigl(\pi(\theta\mid X),\, \pi(\theta\mid \xobs)\bigr)\bigr\}\; \leq \;\epsilon.
\end{align*}
Leveraging this equivalence, our key insight is to collapse the infinite-dimensional posterior maps into a one-dimensional \textit{kernel statistic} formed by applying a distributional metric on posterior measures---thus preserving essential geometric structures, conceptually analogous to the ``kernel trick.''
We realize this idea concretely via the novel Marginally-augmented Sliced Wasserstein (\msw) distance.
The \msw distance preserves marginal structure and mitigates the curse of dimensionality, achieving the parametric convergence rate when $p=1$ (see Section \ref{theory:msw-conv-rate}). Moreover, \msw is topologically equivalent to the classical Wasserstein distance, retaining its geometric properties such as metrizing weak convergence.

\subsubsection{Estimating Trimmed \msw Distance via Deep Conditional Quantile Regression}\label{method:msw-estimation}
To mitigate the well-known sensitivity of the Wasserstein and Sliced Wasserstein distances to heavy tails, we adopt a robust, trimmed variant of the \msw distance, expanding upon the works of \cite{alvarez2008trimmed} and \cite{manole2022minimax}.
To set the stage for our multivariate extension, we first recall the definition of the trimmed Wasserstein distance in one dimension. 
For univariate probability measures $\mu$ and $\nu$, and a trimming parameter $\delta \in [0, 1/2)$, the $\delta$-trimmed $\W_p$ distance is defined as:
\begin{equation}
    \W_{p, \delta}(\mu, \nu) = \left(\frac{1}{1-2\delta}\int_\delta^{1-\delta} \left|F^{-1}_{\mu}(\tau) - F^{-1}_\nu(\tau)\right|^p \, d\tau\right)^{1/p},
\end{equation}
where $F^{-1}_{\mu}$ and $F^{-1}_{\nu}$ denote the quantile functions of $\mu$ and $\nu$, respectively.

We now extend this univariate trimming concept to the multivariate setting and provide a formal definition of the trimmed \msw distance.

\begin{defn}\label{def:msw-formal}
Let $\delta\in[0,1/2)$ be a trimming constant, and let $\mu$ and $\nu$ be probability measures on $\mathbb{R}^d$ (with $d\geq 2$) that possess finite $p$-th moments for $p\geq1$. The $\delta$-trimmed Marginally-augmented Sliced Wasserstein (\msw) distance between $\mu$ and $\nu$ is defined as
\begin{align}
\msw_{p,\delta}(\mu,\nu) &= \lambda\, \underbrace{\frac{1}{d}\sum_{j=1}^{d}\W_{p,\delta}\Bigl(\mu_j,\, \nu_j\Bigr)}_{\text{marginal augmentation}} + (1-\lambda)\, \underbrace{\left(\E_{\varphi\sim\sigma}\Bigl[\W_{p,\delta}^p\Bigl(\varphi_\#\mu,\varphi_\#\nu\Bigr)\Bigr]\right)^{\frac{1}{p}}}_{\text{Sliced Wasserstein distance}},
\end{align}
$\lambda \in (0, 1)$ is a mixing parameter; $\sigma$ is the uniform probability measure on the unit sphere $\smb^{d-1}$; $\mu_j$ denotes the marginal distribution of the $j$-th coordinate under the joint measure $\mu$; and
$\varphi_\#$ denotes the pushforward by the projection $\varphi$. 
This robustification of the \msw distance compares distributions after trimming up to a $2\delta$ fraction of their mass along each projection.
\end{defn}

\begin{rmk}
When $d=1$, Definition \ref{def:msw-formal} reduces to the marginal term alone. In that case, the trimmed \msw distance coincides exactly with the standard trimmed Wasserstein distance between the two one-dimensional distributions.
\end{rmk}

\begin{rmk}
    When $\delta=0$, $\msw_{p, \delta}$ reduces to the untrimmed $\msw_{p}$ distance. For completeness, we give the formal definition of $\msw_p(\cdot, \cdot)$ in Appendix \ref{def:msw-untrimmed}.
\end{rmk}

The trimmed \msw distance comprises two components: the Sliced Wasserstein term, which captures joint interactions through random projections on the unit sphere, and a marginal augmentation term, which gauges distributional disparities along coordinate axes.
The inclusion of the marginal term enhances the \msw distance’s sensitivity to discrepancies along each coordinate axis, remedying the inefficiency of standard SW projections that arises from uninformative directions sampled uniformly at random.
Furthermore, because the SW distance is approximated via Monte Carlo, explicitly accounting for coordinate-wise marginals is particularly pivotal as these marginal distributions directly determine the corresponding posterior credible intervals. 
The value of incorporating axis-aligned marginals has also been highlighted in recent works \citep{moala2010elicitation,drovandi2024improving,chatterjee2025univariate,lu2025generative}.
For brevity, unless stated otherwise, we refer to the trimmed \msw distance simply as the \msw distance throughout for the remainder of this section.

In continuation of our earlier discussion on the need for a posterior space metric, the posterior \msw distance quantifies the extent to which posterior distributions shift in response to perturbations in the observations.
In contrast, most existing ABC methods rely on distances computed directly between datasets, either as $\mathcal{D}(x, \xobs)$ or as $\mathcal{D}(\widehat{\mu}_x, \widehat{\mu}_{\xobs})$ where $\widehat{\mu}_{\boldsymbol{\cdot}}$ denotes the empirical distribution---serving as indirect proxies for posterior discrepancy due to the fundamental challenges in estimating posterior-based metrics.
Importantly, our approach overcomes this limitation by leveraging the quantile representation of the posterior \msw distance, as formally established in Definition \ref{thm:quantile-msw}.

\begin{defn}[Quantile Representation of \msw Distance]\label{thm:quantile-msw}
    The trimmed \msw distance defined in Definition \ref{def:msw-formal} can be equivalently expressed using the quantile representation as
    \begin{align}
        \msw_{p, \delta}(\mu, \nu) &=   \frac{\lambda}{d} \sum_{j=1}^d \left(\frac{1}{1-2\delta}\int_\delta^{1-\delta} \left|F^{-1}_{\mu_j}(\tau) - F^{-1}_{\nu_j}(\tau)\right|^p \, d\tau\right)^{1/p} \nonumber\\
        &\quad + (1-\lambda) \left( \frac{1}{1-2\delta}\int_{\smb^{d-1}} \int_\delta^{1-\delta} \left|F^{-1}_{\varphi_\# \mu}(\tau) - F^{-1}_{\varphi_\# \nu}(\tau)\right|^p \, d\tau \, d\sigma(\varphi)\right)^{1/p}.
    \end{align}
\end{defn}

Building on Definition \ref{thm:quantile-msw}, we reformulate posterior comparison as a conditional quantile regression problem for $\theta$ given $X=x$. 
Specifically, the \msw distance
is constructed in terms of one-dimensional projections of the distributions to leverage the closed-form expression available for univariate Wasserstein evaluation.
By approximating the spherical integral with $K$ Monte Carlo–sampled directions, computing \msw therefore reduces to fitting a series of conditional quantile regressions, each corresponding to a distinct single-dimensional projection. 

To evaluate the \msw distance, we first draw $K > 0$ projections $\{\varphi^{(1)}, \dots, \varphi^{(K)}\}$ uniformly at random from the unit sphere $\smb^{d-1}$. 
Subsequently, we discretize the interval $[\delta, 1-\delta]$ into $H > 0$ equidistant subintervals, each of width $\Delta = (1-2\delta) / H$.
For $x, x' \in \mathcal{X}$, we write $\pi_x = \pi(d\theta \mid x)$ and $\pi_{x'} = \pi(d\theta \mid x')$ as shorthand for the respective posterior distributions.
The posterior $\msw_{p,\delta}$ distance between $\pi_x, \pi_{x'}$, denoted $\widehat{\msw}_{p, \delta, K, H}\left(\pi_x, \pi_{x'} \right)$, can be approximated as follows:
\begin{align}
    \widehat{\msw}_{p, \delta, K, H}\left(\pi_x, \pi_{x'}\right) &= 
     \frac{\lambda}{d} \sum_{j=1}^d \left[ I_H \left(F^{-1}_{\pi_{x, j}}, F^{-1}_{\pi_{x', j}} \right)\right]^{1/p} \nonumber \\
    &\qquad + \, \, (1-\lambda) \left ( \frac{1}{K} \sum_{k=1}^K I_H\left (F^{-1}_{\varphi^{(k)}_\# \pi_x}, F^{-1}_{\varphi^{(k)}_\# \pi_{x'}} \right) \right)^{1/p}, \label{msw:emp-eval} \\[1em]
    \text{where } I_H(q_1, q_2) &= \frac{\Delta}{2 (1-2\delta)} \Big( \big|q_1(\delta ) - q_2(\delta) \big|^p + 2\sum_{h=1}^{H-1} \big|q_1(\delta + h\Delta) - q_2(\delta + h\Delta) \big|^p \nonumber \\
    &\qquad +  \, \, \big|q_1(1 - \delta) - q_2(1-\delta) \big|^p\Big).
\end{align}
In the equations above, $I_H$ represents the trapezoidal discretization function, while $F^{-1}_{\pi_{x, j}}(\tau)$ denotes the quantile function associated with the $j$-th coordinate of $\pi_x$ evaluated at the $\tau$th quantile. 

\begin{rmk}
    The estimated \msw distance can be viewed both as a measure of discrepancy between the posterior distributions and as an informative low-dimensional kernel statistic. Depending on the context, we will use these interpretations interchangeably to best suit the task at hand.
\end{rmk}

\begin{rmk}
    When only the individual posterior marginals $\pi(\theta_j \mid\xobs)$ for $j=1,\dots,d$ are of interest, one can elide the Sliced Wasserstein component entirely and compute only the univariate marginal terms. 
    Since the marginals often suffice for decision-making without the full joint posterior \citep{moala2010elicitation}, this approach yields substantial computational savings.  
\end{rmk}

To estimate the quantile functions, we perform nonparametric conditional quantile regression (CQR) via deep ReLU neural networks. These networks have demonstrated remarkable abilities to approximate complex nonlinear functions and adapt to unknown low-dimensional structures while possessing attractive theoretical properties. 
In particular, \cite{padilla2022quantile} establish that, under mild smoothness conditions, the ReLU-network quantile regression estimator attains minimax-optimal convergence rates.

\begin{defn}[Deep Neural Networks]\label{def:dnn}
Let $\phi(\mathbf{x}) = \max\{\mathbf{x}, 0\}$ be the ReLU activation function. For a network with $L$ hidden layers, let $\mathbf{d} = (d_0, d_1, \dots, d_{L+1})^\top \in \mathbb{R}^{L+2}$ specify the number of neurons in each layer, where $d_0$ represents the input dimension and $d_{L+1}$ the output dimension.
The class of multilayer feedforward ReLU neural networks specified by architecture $(L, \mathbf{d})$ comprises all functions from $\R^{d_0}$ to $\R^{d_{L+1}}$ formed by composing affine maps with elementwise ReLU activations:
\begin{align*}
    f(\mathbf{x}) = g_{L+1} \circ \phi \circ g_L \circ \cdots \circ \phi \circ g_1(\mathbf{x}),
\end{align*}
where each layer $\ell$ is represented by an affine transformation $g_\ell(\mathbf{d}^{(\ell-1)}) = \mathbf{W}^{(\ell)}d_{\ell-1} + \mathbf{b}^{(\ell)}$, with $\mathbf{W}^{(\ell)} \in \mathbb{R}^{d_\ell \times d_{\ell-1}}$ as the weight matrix and $\mathbf{b}^{(\ell)} \in \mathbb{R}^{d_\ell}$ as the bias vector.
\end{defn}

Building on Definition \ref{def:dnn}, we approximate \eqref{msw:emp-eval} by training a single deep ReLU network to jointly predict all slice-quantiles.
Let $K' := K + d$ denote the total number of directions in the augmented projection set $\Phi = \{\varphi^{(k)}\}_{k=1}^K \cup \{e_j\}_{j=1}^d$.
For each projection $\varphi^{(k)} \in \Phi$ and quantile level $\tau_h = \delta + h\Delta$ with $h = 0, 1, \dots, H$ and $\Delta = (1-2\delta)/H$, let 
\begin{align*}
    Q^*_{k, h}(x) = F^{-1}_{\varphi^{(k)}_\# \pi_x} (\tau_h)
\end{align*}
be the true conditional $\tau_h$-quantile along $\varphi^{(k)}$. 
Given $\ntrain$ training pairs $\{(x^{(m)}, \theta^{(m)})\}_{m=1}^{\ntrain}$, we learn $\qnetest: \mathcal X \to \R^{K'\,(H+1)}$ by solving
\begin{align}\label{msw:quantile-optim}
    \qnetest &= \argmin_{Q \in \Q} ~~\sum_{k=1}^{K'} \sum_{h=0}^H\sum_{m=1}^{\ntrain} \rho_{\tau_h, \kappa}(\langle \varphi^{(k)}, \theta^{(m)}\rangle - Q_{[k, h]}(x^{(m)})),
\end{align}
where $\Q$ is the class of ReLU neural network models with architecture $(L, \mathbf{d})$ and output dimension $d_{L+1}=K'(H+1)$, and $Q_{[k, h]}$ is the $((H+1)(k-1)+h+1)$-th entry of the flattened output. 
This single network thus shares parameters across all $K'$ slices and $H+1$ quantile levels. 

Contrary to the conventional pinball quantile loss \citep{padilla2022quantile}, we employ the Huber quantile regression loss \citep{huber1964robust}, which is less sensitive to extreme outliers. This loss function, parameterized by threshold $\kappa$ \citep{dabney2018implicit}, is defined as:
\begin{align}
    \rho_{\tau, \kappa}(u) &= 
    \begin{cases}
        \frac{1}{2 \kappa} |\tau - \ind(u < 0)| u^2, & |u| \leq \kappa\\
        |\tau - \ind(u < 0)|(|u| - \frac{1}{2}\kappa),  & |u| > \kappa.
    \end{cases}
\end{align}
Upon convergence of the training process, we obtain a single quantile network that outputs the predicted quantile of $\theta$ for any given projection $\varphi^{(k)}$, quantile level $\tau_h$, and conditioning variable $x \in \mathcal{X}$.

\begin{rmk}
    Unlike \cite{padilla2022quantile}, we impose no explicit monotonicity constraints in \eqref{msw:quantile-optim}; that is, we do not enforce the following ordering restrictions during the joint estimation stage:
    \begin{align*}
        \widehat{Q}_{[k,h]}(X^{(m)}) \; \leq \; \widehat{Q}_{[k,h']}(X^{(m)}) \quad \text{for all } 0 \leq h<h'\le H,\;k\in[K'],\; m\in[\ntrain].
    \end{align*}
    Instead, after obtaining the predictions $\{\widehat{Q}_{[k,h]}(x)\}_{h=0}^H$ for each slice $\varphi^{(k)}$, we simply sort these $H+1$ values in ascending order.  
    This post‐processing step automatically guarantees the non‐crossing restriction without adding any constraints to the optimization.
\end{rmk}

Collectively, the elements in this section constitute the core of the \textit{nonparametric distribution matching} component of our proposed methodology. The corresponding algorithmic procedure for distribution matching is summarized in \abc{Algorithm} \ref{alg:msw-estimation}.

\begin{algorithm}[h]
\caption{Trimmed \msw Distance Estimation via CQR}
\label{alg:msw-estimation}
\KwIn{Proposal distribution $\pi^{(t)}(\theta, X)$\;
    \phantom{\textbf{Input:} }Number of random projections $K$\;
    \phantom{\textbf{Input:} }Number of quantile levels $H$\;
    \phantom{\textbf{Input:} }Smoothing parameter $\kappa$\;
    \phantom{\textbf{Input:} }Network architecture $(L,\mathbf{d})$\;
    \phantom{\textbf{Input:} }Training sample size $\ntrain$\;
}
\vspace{0.5em}
\begin{spacing}{1.2}
Sample $K$ directions $\{\varphi^{(k)}\}_{k=1}^K\sim\mathrm{Unif}(S^{d-1})$\; 
Set $\Phi \gets \{\varphi^{(k)}\}_{k=1}^K \cup \{e_j\}_{j=1}^d$\;
Compute quantile grid $\tau_h=\delta+h\frac{1-2\delta}H$ for $h=0,\dots,H$\;
\For{$m=1,\dots,\ntrain$}{
  Generate sample $(\theta^{(m)}, X^{(m)})\sim \pi^{(t)}(\theta, X)$\;
}
Train quantile network $\qnetest$ with architecture $(L,\mathbf{d})$ by minimizing the empirical risk in \eqref{msw:quantile-optim} using loss $\rho_{\tau_h,\kappa}$\;
\For{$k=1, \dots, K'$}{
  For each sample $X^{(m)}$, sort the $H+1$ outputs $\{\hat{Q}_{[k,h]}(X^{(m)})\}_{h=0}^H$ to enforce monotonicity\;
}
\Return trained network $\qnetest$ for evaluating $\widehat{\msw}_{p,\delta,K,H}(\pi_x,\pi_{\xobs})$ via \eqref{msw:emp-eval}.
\end{spacing}
\end{algorithm}

\subsection{Adaptive Rejection Sampling}
In contrast to the single-round rejection framework of classical rejection-ABC, \abi implements an adaptive rejection sampling approach. For clarity of exposition, we decompose this approach into three distinct stages: proposal sampling, conditional refinement, and sequential updating.

At a high level, this sequential scheme decomposes the target event into a chain of more tractable conditional events.
Define a sequence of nested subsets $\X^n \supset A_1 \supset A_2 \supset \cdots \supset A_{\mathsf T}$ associated with decreasing thresholds $\epsilon_1 > \epsilon_2 > \cdots > \epsilon_\mathsf{T}$, following a structure similar to adaptive multilevel splitting.
We proceed by induction. 
At initialization, draw $\theta \sim \pi(\theta)$ and $X \mid \theta \sim P_\theta^{(n)}(\cdot)$. 
At iteration one, we condition on the event $X\in A_1$ by selecting among the initial samples
$(\theta, X)$ those for which $X \in A_1$, so that the joint law becomes 
\begin{equation*}
\qquad \p\bigl(\theta\mid X\in A_1\bigr)\;P_\theta\bigl(X\mid X\in A_1\bigr)\;=\;\p\bigl(\theta,\,X\mid X\in A_1\bigr)\quad(\text{conditional refinement}).
\end{equation*}
At iteration $t$, we first obtain samples from $\p(\theta,X\mid X\in A_{t-1})$ generated by
\begin{equation*}
    \qquad \theta \sim \pi \bigl(\theta \mid X\in A_{t-1} \bigr), 
    \quad X \mid \theta, ~A_{t-1} \sim P_\theta^{(n)} \bigl (\cdot \mid X\in A_{t-1}\bigr) \qquad(\text{proposal sampling}).
\end{equation*}
To refine to $A_t\subseteq A_{t-1}$, note that
\begin{align*}
    \p\bigl(\theta \mid X\in A_{t-1}\bigr)\, P_\theta\bigl(X\in A_t\mid X\in A_{t-1}\bigr) &= \frac{\p(\theta,X\in A_t)}{\p(X\in A_{t-1})} \\
&\propto \p\bigl(\theta\mid X\in A_t\bigr) \qquad (\text{sequential update}).
\end{align*}
Thus, by conditioning on $A_t$, we obtain samples from the intermediate partial posterior $\p(\theta\mid X\in A_t)$. 
Iterating this procedure until termination yields the final approximation $\p(\theta\mid X\in A_\mathsf{T})$, which converges to the true posterior as $\epsilon_{\mathsf{T}}$ approaches 0 and the acceptance regions become increasingly precise.

In the following subsections, we present a detailed implementation for each of these three steps.

\subsubsection{Sampling from the Refined Proposal Distribution}\label{method:sample-proposal}
In this section, we describe how to generate samples from the refined proposal distribution using rejection sampling.

\paragraph{Decoupling the Joint Proposal}
Let $\epsilon_1 > \epsilon_2 > \cdots > \epsilon_\mathsf{T}$ be a user-specified, decreasing sequence of tolerances.
Define the data‐space acceptance region and its corresponding event by
\begin{gather*}
    A_{t} = \{x \in \X^n : \widehat{\msw}_{p, \delta, K, H}(\pi_x, \pi_{\xobs}) \leq \epsilon_{t}\} \subseteq \X^n, \\
    E_{t} = \left\{\omega: X(\omega) \in A_{t}\right\}.
\end{gather*}
At iteration $t \geq 1$, we adopt the joint proposal distribution over $(\theta,X)$ given by:
\begin{align*}
  \pi^{(1)}(\theta,X) &= \pi(\theta)\,P_\theta^{(n)}(X),\\
  \pi^{(t)}(\theta,X) &= \pi\bigl(\theta,X\mid E_{t-1}\bigr) \quad \text{for } t=2,\dots, \mathsf{T},
\end{align*}
where we condition on the event $E_{t-1}=\{X\in A_{t-1}\}$ with $A_0=\mathcal{X}$. 
Since direct sampling from this conditional distribution is generally infeasible, we recover it via rejection sampling after decoupling the joint proposal.  
Let 
\begin{align*}
    \pi^{(t)}(\theta) =\int_{\X^n} \pi^{(t)}(\theta,x) \,dx
\end{align*}
be the marginal law of $\pi^{(t)}$ over $\theta$. 
This auxiliary distribution matches the correct conditional marginal while remaining independent of $X$. 
Observe that the joint proposal for the $t$-th iteration $\pi^{(t)}$ admits the factorization:
\begin{align}\label{method:rej-sampling-factorization}
    \pi^{(t)}(\theta, X) 
    &= \pi(\theta \mid E_{t-1}) P_\theta^{(n)}(X \mid E_{t-1}) \nonumber\\
    &= \underbrace{\pi^{(t)}(\theta)}_{\text{marginal in $\theta$}} \underbrace{P_\theta^{(n)}\left(X \mid X \in A_{t-1} \right)}_{\text{constraint on $X$}}.
\end{align}
Note that for a given $\theta \in \Omega$, the data‐conditional term in the equation above satisfies
\begin{align}
    P_\theta^{(n)}\left(X\mid X \in A_{t-1}\right)\propto P_\theta^{(n)}(X)\underbrace{\ind \left\{ \widehat{\msw}_{p, \delta, K, H}(\pi_{X}, \pi_{\xobs}) \leq \epsilon_{t-1}\right\}}_{\text{constraint on $X$}},
\end{align}
where the proportional symbol hides the normalizing constant $P_\theta^{(n)}(A_{t-1})$.
The decomposition in \eqref{method:rej-sampling-factorization} cleanly decouples the proposal distribution into a marginal draw over $\theta$ and a constraint on $X$. 
In other words, the first component eliminates the coupling while retaining the correct conditional marginal $\pi^{(t)}(\theta)$, and the second term imposes a data‐dependent coupling constraint to be enforced via a simple rejection step.

\paragraph{Sampling the Proposal via Rejection}
In order to draw 
\begin{align*}
    (\theta,X) \sim \pi^{(t)}(\theta,X) = \pi(\theta, X \mid E_{t-1})
\end{align*}
without computing its normalizing constant $P_\theta^{(n)}(A_{t-1})$, we apply rejection sampling to the unnormalized joint factorization in \eqref{method:rej-sampling-factorization}:
\begin{enumerate}
    \item Sample $\theta \sim \pi^{(t)}(\theta)$;
    \item Generate $X \sim P_{\theta}^{(n)}$ \textbf{repeatedly} until $\ind\{\widehat{\msw}_{p, \delta, K, H}(\pi_{X}, \pi_{\xobs}) \leq \epsilon_{t-1}\} = 1$.\label{step:rej-samp-while} 
\end{enumerate}
By construction, the marginal distribution of $\theta$ remains $\pi(\theta\mid E_{t-1})$ since all $\theta$ values are unconditionally accepted, while the acceptance criterion precisely enforces the constraint $X \in A_{t-1}$.
Consequently, the retained pairs $(\theta,X)$ follow the desired joint distribution $\pi(\theta,X\mid E_{t-1})$.
However, it is practically infeasible to perform exact rejection sampling as the expected number of simulations for Step \ref{step:rej-samp-while} may be unbounded.
To address this limitation, we introduce a budget-constrained rejection procedure termed Approximate Rejection Sampling (\abc{ARS}), as outlined in Algorithm \ref{alg:approx-rej-samp}.
The core idea is as follows: given a fixed computational budget $\mathsf{R} \in \mathbb{N}_+$, we repeat Step \ref{step:rej-samp-while} at most $\mathsf{R}$ times. If no simulated data set satisfies the tolerance criterion within this budget, the current parameter proposal is discarded, and the algorithm proceeds to the next parameter draw.

\begin{algorithm}[h] 
\caption{Approximate Rejection Sampling}
\label{alg:approx-rej-samp}
\For{$i = 1, \dots, \mathsf{N}$}{
    Sample $\theta^{(i)} \sim \pi^{(t)}(\theta)$\;
    \For{$r = 1, \dots, \mathsf{R}$}{
        Sample $X^{(i, r)} \sim P_{{\theta^{(i)}}}^{(n)}$\;
        \If{$\widehat{\msw}_{p, \delta, K, H}(\pi_{X^{(i, r)}}, \pi_{\xobs}) \leq \epsilon_{t-1}$} {
            Retain $(\theta^{(i)}, X^{(i, r)})$\;
            \textbf{break}\;        
        }
    }
    \If{$\widehat{\msw}_{p, \delta, K, H}(\pi_{X^{(i, r)}}, \pi_{\xobs}) > \epsilon_{t-1}$ for all $r = 1, \dots, \mathsf{R}$} {
        Discard $\theta^{(i)}$\;
    }
}
\end{algorithm}

While approximate rejection sampling introduces a small bias, the approximation error becomes negligible under appropriate conditions. Theorem \ref{thm:ars-sample-complexity} establishes that, under mild regularity conditions, the resulting error decays exponentially fast in $\mathsf{R}$.

\begin{assumption}[Local Positivity]\label{assump:tv-dist-local-positivity}
    There exists constants $c > 0$ and $\gamma>0$ such that, for every $\theta$ with $\pi^{(t)}(\theta)>0$,
    the per-draw acceptance probability $P_\theta^{(n)}(A_t)$ is uniformly bounded away from $0$ satisfying $P_\theta^{(n)}(A_t) \geq c \epsilon_t^\gamma$.
\end{assumption}

\noindent Assumption \ref{assump:tv-dist-local-positivity} is satisfied, for instance, if the kernel statistic $D_\theta(X) = \widehat{\msw}_{p, \delta, K, H} (\pi_X, \pi_{\xobs})$ under $X\sim P_\theta^{(n)}$ admits a continuous density $q_\theta(u)$ that is strictly positive in a neighborhood of $u=0$.  In that case, for small $\epsilon_t$,  
\begin{align*}
    P_{\theta}^{(n)}(A_{t}) = P_\theta^{(n)}(D_\theta \leq \epsilon_t) = \int_0^{\epsilon_t} q_\theta(u) du \geq q_\theta(0)\epsilon_t/2. 
\end{align*}

\begin{thm}[Sample Complexity for \abc{ARS}]\label{thm:ars-sample-complexity}
Suppose Assumption \ref{assump:tv-dist-local-positivity} holds.
For any $\bar{\delta} \in (0,1)$ and $\epsilon_t > 0$, if the number of proposal draws $\mathsf{R}$ satisfies $\mathsf{R} = O\left(\frac{\log(1/\bar{\delta})}{\epsilon_t^{\gamma}}\right)$, 
then the total-variation distance between the exact and approximate proposal distributions obeys
$\mathcal{D}_{\mathsf{TV}} \left(\pi^{(t)}, \pi_{\mathsf{ARS}}^{(t)} \right) \leq \bar{\delta}$.
\end{thm}

\FloatBarrier

\subsubsection{Adaptive Refinement of the Partial Posterior}\label{method:pruning-stage}

The strength of \abi lies in its sequential refinement of partial posteriors through a process guided by a descending sequence of tolerance thresholds $\epsilon_1 > \epsilon_2 > \cdots > \epsilon_T$. 
This sequence progressively tightens the admissible deviation from the target posterior $\tpos$, yielding increasingly improved posterior approximations. 
By iteratively decreasing the tolerances rather than prefixing a single small threshold, \abi directs partial posteriors to dynamically focus on regions of the parameter space most compatible with the observed data.
This adaptive concentration is particularly advantageous when the prior is diffuse (i.e., uninformative) or the likelihood is concentrated in low-prior-mass regions, a setting in which one-pass ABC is notoriously inefficient.

The refinement procedure unfolds as follows.
First, we acquire $\mathsf{N}$ samples from the proposal distribution via Algorithm \ref{alg:approx-rej-samp}, namely
\begin{align*}
  \theta^{(i)} &\sim \pi^{(t)}(\theta),\\
  X^{(i)}\mid\theta^{(i)} &\sim P_{\theta^{(i)}}^{(n)}\bigl(\cdot\mid E_{t-1}\bigr),
  \quad i=1,\dots,\mathsf{N},
\end{align*}
which form the initial proposal set $\mathsf{S}^{(t)}_{0} = \{(\theta^{(i)}, X^{(i)})\}_{i=1}^\mathsf{N}$ for proceeding refinement.
Next, we retain only those parameter draws $\theta^{(i)}$ that exhibit a sufficiently small estimated posterior discrepancy; that is, parameters whose associated simulated data satisfy 
\begin{align*}
    \widehat{\msw}_{p, \delta, K, H}(\pi_X, \pi_{\xobs}) \leq \epsilon_{t}
\end{align*}
and discard the remainder.
This selection yields the training set for the generative density estimation step,
\begin{align*}
\mathsf S_{\theta,*}^{(t)}
=\bigl\{\theta^{(i)} : (\theta^{(i)},X^{(i)})\in\mathsf S_0^{(t)}
\ \text{and}\ 
\widehat{\msw}_{p,\delta,K,H}(\pi_{X^{(i)}},\pi_{\xobs})\leq\epsilon_t \bigr\},
\end{align*}
consisting of parameter samples drawn from the refined conditional distribution at the current iteration.
Let $\pi_*^{(t)}$ denote the true (unobserved) marginal distribution of $\theta$ underlying the empirical parameter set $\mathsf{S}^{(t)}_{\theta, *}$\footnote{The true distribution $\pi_*^{(t)}$ is unobserved, as only its empirical counterpart is available.}. 
Importantly, $\pi_*^{(t)}$ depends solely on $\theta$ since the $X$ component has been discarded---thus removing the coupling between $\theta$ and $X$.

By design, our pruning procedure progressively refines the parameter proposals by incorporating accumulating partial information garnered from previous iterations.
As the tolerance decays, the retained parameters are incrementally confined to regions that closely align with the target posterior $\tpos$, thereby sculpting each partial posterior $\pi_*^{(t)}$ toward $\tpos$.

\paragraph{Determining the Sequence of Tolerance Levels}
Thus far, our discussion has implicitly assumed that the choice of $\epsilon_t$ yields a non-empty set $\mathsf{S}^{(t)}_{\theta, *}$. 
To ensure that $\mathsf{S}^{(t)}_{\theta, *}$ is non-empty, the tolerance level $\epsilon_t$ must be chosen judiciously relative to $\epsilon_{t-1}$. in particular, $\epsilon_t$ should neither be substantially smaller than $\epsilon_{t-1}$ (which might result in an empty set) nor excessively large (which would lead to inefficient refinement). 
By construction, the initial proposal samples $\mathsf{S}^{(t)}_{0} = \{\theta^{(i)}, X^{(i)}\}$ satisfy
\begin{align*}
    \max_{i=1,\dots,\mathsf{N}} \widehat{\msw}_{p, \delta, K, H}(\pi_{X^{(i)}}, \pi_{\xobs}) \leq \epsilon_{t-1}.
\end{align*}

Consequently, we determine the sequence of tolerance thresholds empirically (analogous to adaptive multilevel splitting) by setting $\epsilon_t(\alpha)$ as the $\alpha$th quantile of the set:
\begin{align*}
    \left\{\widehat{\msw}_{p, \delta, K, H}(\pi_{X^{(i)}}, \pi_{\xobs})\}_{i=1}^{\mathsf{N}} : \widehat{\msw}_{p, \delta, K, H}(\pi_{X^{(i)}}, \pi_{\xobs}) \leq \epsilon_{t-1}\right\},
\end{align*}
where $\alpha \in (0, 1)$ is a quantile threshold hyperparameter \citep{biau2015new}.
In this manner, our selection procedure yields a monotone decreasing sequence of thresholds, $\epsilon_0(\alpha) > \epsilon_1(\alpha) > \cdots > \epsilon_T(\alpha)$, while ensuring that the refined parameter sets $\mathsf{S}_{\theta, *}^{(t)}$ remain non-empty.

\subsubsection{Sequential Density Estimation}\label{method:generative-density-estimation}
To incorporate the accumulated information into subsequent iterations, we update the proposal distribution using the current partial posterior. Recall that the partial posterior factorizes into a marginal component over $\theta$ and a constraint on $X$, as described in Section \ref{method:sample-proposal}. 
In this section, we focus on updating the marginal partial posterior by applying generative density estimation to the retained parameter draws from the preceding pruning step. 
This process ensures that the proposal distribution at each iteration reflects all refined information acquired in earlier iterations.

\paragraph{Marginal Proposal Update with Generative Modeling}
Our approach utilizes a generative model $G_\beta(Z): Z \mapsto \hat{\theta}$, parameterized by $\beta$, which transforms low-dimensional latent noise $Z$ into synthetic samples $\hat{\theta}$. 
When properly trained to convergence on the refined set $\mathsf{S}_{\theta, *}^{(t)}$, $G_{\beta}$ produces a \textit{generative distribution} denoted by $\widehat{\pi}_*^{(t)}$ that closely approximates the target partial posterior $\pi^{(t)}_*$.
Since \abi is compatible with any generative model---including 
generative adversarial networks, variational auto-encoders, and Gaussian mixture models---practitioners enjoy considerable flexibility in their implementation choices.
In this work, we employ \potnet \citep{lu2025generative} because of its robust performance and resistance to mode collapse, which are crucial attributes for preserving diversity of the target distribution and minimizing potential biases arising from approximation error that could propagate to subsequent iterations.

At the end of iteration $t$, we update the proposal distribution for the $(t+1)$-th iteration with the $t$-th iteration's approximate marginal partial posterior:
\begin{align*}
    \pi^{(t+1)}(\theta) \leftarrow \widehat{\pi}^{(t)}_*(\theta).
\end{align*}
Thereafter, we can simply apply the \abc{ARS} algorithm described in Section \ref{method:sample-proposal} to generate samples from the joint proposal $\pi^{(t+1)}(\theta, X)$ conditional on the event $E_t$.
At the final iteration $\mathsf{T}$, we take the \abi\ posterior to be  
\begin{align*}
    \widehat{\pi}_{\abi}^{(\epsilon_\mathsf{T})} (\theta \mid \xobs) \gets \widehat{\pi}_*^{(\mathsf{T})}(\theta),
\end{align*}
which approximates the coarsened target distribution 
$\pi(\theta\mid \msw_p(\pi_X,\pi_{\xobs})\leq \epsilon_{\mathsf T})$.
We emphasize that the core of \abi’s sequential refinement mechanism hinges on the key novelty of utilizing generative models, whose inherent generative capability enables approximation and efficient sampling from the revised proposal distributions.

\paragraph{Iterative Fine-tuning of the Quantile Network}
We retrain the quantile network on the newly acquired samples at each iteration to fully leverage the accumulated information, thereby adapting the kernel statistics to become more informative about the posterior distribution.
This continual fine-tuning improves our estimation of the posterior \msw discrepancy and thus yields progressively more accurate refinements of the parameter subset in subsequent iterations.


\section{Theoretical Analysis}\label{sec:theory}
In this section, we investigate theoretical properties of the proposed \msw distance, its trimmed version $\msw_{p, \delta}(\cdot, \cdot)$, and the \abi algorithm. 
In Section \ref{sec:theory:msw}, we first establish some important topological and statistical properties of the \msw distance between distributions $\mu, \nu$ under mild moment assumptions. 
In particular, we show that the error between the empirical \msw distance and the true \msw distance decays at the parametric rate (see Remark \ref{rmk:msw-conv-rate}).
Then, in Section \ref{theory:abi-property}, we derive asymptotic guarantees on the convergence of the resulting \abi posterior in the limit of $\epsilon\downarrow 0$.

We briefly review the notation as follows. Throughout this section, we assume that $p\geq 1$ and $d\in\mathbb{N}_+$. We denote by $\sigma(\cdot)$ the uniform probability measure on $\smb^{d-1}$, and by $\mathcal{P}_p(\R^d)$ the space of probability measures on $\R^d$ with finite $p$-th moments. Given a probability measure $\mu\in\mathcal{P}_p(\R^d)$ and the projection mapping $f_\varphi: \theta \mapsto \varphi^\top \theta$ (where $\varphi\in\smb^{d-1}$), we write $\varphi_\# \mu$ for its pushforward under the projection $f_\varphi(\cdot)$. Additionally, for any $\alpha\in [0,1]$ and any one-dimensional probability measure $\gamma$, we denote by $F^{-1}_{\gamma}(\alpha)$ the $\alpha$th quantile of $\gamma$, 
\begin{align*}
    F^{-1}_{\gamma}(\alpha) = \inf\bigl\{\,x\in\mathbb R\;\bigm|\;F_{\gamma}(x)\ge\alpha\bigr\},
\end{align*}
where $F_{\gamma}(x)=\gamma((-\infty,x])$ is the cumulative distribution function of $\gamma$.

\subsection{Topological and Statistical Properties of the \msw Distance}\label{sec:theory:msw}
We first establish important topological and statistical properties of the \msw distance. Specifically, we show that the \msw distance is indeed a metric, functions as an integral probability metric when $p=1$, topologically equivalent to $W_p$ on $\cp_p(\R^d)$ and metrizes weak convergence on $\cp_p(\R^d)$.

\subsubsection{Topological Properties}
Subsequently, we denote by $\msw_p(\cdot,\cdot)$ the untrimmed version of the \msw distance with $\delta=0$ (as defined in \ref{def:msw-untrimmed}) and omit the subscript $\delta$.

\begin{prop}[Metricity]\label{prop:msw-metric}
    The untrimmed Marginally-augmented Sliced Wasserstein $\msw_{p}(\cdot, \cdot)$ distance is a valid metric on $\mathcal{P}_p(\R^d)$. 
\end{prop}

\medskip 
\noindent Our first theorem shows that the $1$-\msw distance is an IPM and allows a dual formulation. 

\begin{thm}\label{thm:msw-ipm}
    The $1$-\msw distance is an Integral Probability Metric defined by the class,
    \begin{align}
    \mathcal{F}_{\msw} = \Bigg\{f: \mathbb{R}^d \to \mathbb{R} \,\Bigg|\,& f(x) = 
    \tfrac{\lambda}{d} \sum_{j=1}^d g_j(e_j^\top x) + (1-\lambda) \int_{\mathbb{S}^{d-1}} g_\varphi(\varphi^\top x)\, d\sigma(\varphi) :\nonumber\\
    &\quad  g_j, g_\varphi \in \Lip_1(\mathbb{R}),\sup_{\varphi \in \smb^{d-1}} |g_\varphi(0)| < \infty\Bigg\}, \quad 0 < \lambda < 1,
\end{align}
where for each $\varphi \in \mathbb{S}^{d-1}$, $g_\varphi: \mathbb{R} \to \mathbb{R}$ is a 1-Lipschitz function, such that the mapping $(\varphi, t) \mapsto g_\varphi(t)$ is jointly measurable with respect to the product of the Borel $\sigma$-algebras on $\mathbb{S}^{d-1}$ and $\mathbb{R}$. 
\end{thm}

\begin{thm}[Topological Equivalence of $\msw_p$ and $\W_p$]\label{thm:msw-topo-equiv}
There exists a constant $C^\dagger_{d, p, \lambda} > 0$ depending on $d, p, \lambda$, such that for $\mu, \nu \in \cp_p(B_R(0))$, 
    \begin{align*}
        \msw_p(\mu, \nu) \leq C_{d, p, \lambda} \W_p(\mu, \nu) \leq C^\dagger_{d, p, \lambda} R^{1-\frac{1}{p(d+1)}} \msw_p^{\frac{1}{p(d+1)}}(\mu,\nu),
    \end{align*}
where $C_{d, p, \lambda} =\lambda+ (1-\lambda)\left(d^{-1}\int_{\smb^{d-1}} \|\varphi\|_p^p \, d\sigma(\varphi)\right)^{1/p}$. Consequently, the $p$-\msw distance induces the same topology as the $p$-Wasserstein distance.
\end{thm}

\begin{thm}[\msw Metrizes Weak Convergence]\label{thm:msw-metrize-conv}
The $\msw_p$ distance metrizes weak convergence on $\cp_p(\R^d)$, in the sense of metricity as defined in Definition 6.8 of \citet{villani2009optimal}.
\end{thm}

\begin{rmk}
This result holds without the requirement of compact domains.
\end{rmk}

\subsubsection{Statistical Properties}\label{theory:msw-stats-prop}

In this section, we establish statistical guarantees for the trimmed $\msw_{p, \delta}(\cdot,\cdot)$ distance as formalized in Definition \ref{def:msw-formal}. We focus particularly on how closely the empirical version of this distance approximates its population counterpart when estimated from finite samples.
For any $\mu,\nu\in\mathcal{P}_p(\mathbb{R}^d)$ and $m,m'\in \mathbb{N}_+$, we denote by $\widehat{\mu}_m$ and $\widehat{\nu}_{m'}$ the empirical measures constructed from $m$ and $m'$ i.i.d.\ samples drawn from $\mu$ and $\nu$, respectively. Our main result, presented in Theorem \ref{theory:msw-conv-rate}, derives a non-asymptotic bound on $\left|\msw_{p, \delta}(\widehat{\mu}_m, \widehat{\nu}_{m'})-\msw_{p, \delta}(\mu, \nu)\right|$ that achieves the parametric convergence rate of $m^{-1/2}$ when $p=1$ and $m=m'$, as to be shown in Eq.\ \eqref{eq:msw-parametric-rate}.

\paragraph{Assumptions.} We assume that $\mu,\nu\in\mathcal{P}_{p'}(\R^d)$ where $p' = \max\{p,2\}$. The sample sizes $m,m'\in\mathbb{N}_+$ satisfy $\min\{m,m'\}> \max\{2(p+2)\slash \delta,\log(32 d\slash \overline{\delta})\slash (2\delta^2)\}$, where $\delta\in (0,1/2)$ is the trimming parameter and $\overline{\delta}\in(0,1)$ is the confidence level. We define effective radii $M_{\mu,p},M_{\nu,p}\in (0,\infty)$ such that $\mathbb{E}_{Z\sim\mu}[\|Z\|^p]<M_{\mu,p}$ and $\mathbb{E}_{Z\sim\nu}[\|Z\|^p]<M_{\nu,p}$; the existence of these radii follows from the fact that $\mu,\nu\in\mathcal{P}_p(\R^d)$.

\paragraph{Notation.} For simplicity, we denote one-dimensional projections along coordinate axes as $\mu_j:=(e_j)_\# \mu$ and $\nu_j:=(e_j)_\#\nu$ for all $j\in[d]$. Similarly, for general projections, we write $\mu_\varphi:=\varphi_\#\mu$ and $\nu_\varphi:=\varphi_\#\nu$ for all $\varphi\in\mathbb{S}^{d-1}$.

For any one-dimensional probability measure $\gamma$ and any $t\geq 0$, we define 
\begin{equation}\label{phideltat}
    \psi_{\delta,t}(\gamma):= \min\left\{F_{\gamma}\left(F_{\gamma}^{-1}(\delta)+t\right)-\delta,\; \delta-F_{\gamma}\left(F_{\gamma}^{-1}(\delta)-t\right)\right\},
\end{equation}
\begin{equation}
    \psi_{1-\delta,t}(\gamma):= \min\left\{F_{\gamma}\left(F_{\gamma}^{-1}(1-\delta)+t\right)-(1-\delta),\; (1-\delta)-F_{\gamma}\left(F_{\gamma}^{-1}(1-\delta)-t\right)\right\},
\end{equation}
where $F_{\gamma}$ is the CDF of $\gamma$. Note that from our assumption $\min\{m,m'\}>\log(32 d\slash \overline{\delta})\slash (2\delta^2)$, we have $2\exp(-2\min\{m,m'\}\delta^2)<\overline{\delta}\slash (16 d)$; as $\lim_{t\rightarrow\infty} \psi_{\delta,t}(\gamma)=\delta$, there exist $\eps_{m,d,\delta, \overline{\delta}}(\gamma),\eps_{m',d,\delta, \overline{\delta}}(\gamma)\in (0,\infty)$ such that
\begin{equation}\label{eps_def}
    2\exp\left(-2m\, \psi_{\delta,\eps_{m,d, \delta, \overline{\delta}}(\gamma)}(\gamma)^2\right) \leq \frac{\overline{\delta}}{16d}, \quad 2\exp\left(-2m'\, \psi_{\delta,\eps_{m',d, \delta, \overline{\delta}}(\gamma)}(\gamma)^2\right) \leq \frac{\overline{\delta}}{16d}.
\end{equation}
Similarly, let $\eps_{m,d,1-\delta, \overline{\delta}}(\gamma),\eps_{m',d,1-\delta, \overline{\delta}}(\gamma)\in (0,\infty)$ be such that
\begin{equation}\label{def_epsnew}
    2\exp\left(-2m\, \psi_{1-\delta,\eps_{m,d,1-\delta, \overline{\delta}}(\gamma)}(\gamma)^2\right) \leq \frac{\overline{\delta}}{16d}, \quad 2\exp\left(-2m'\, \psi_{1-\delta,\eps_{m',d, 1-\delta, \overline{\delta}}(\gamma)}(\gamma)^2\right) \leq \frac{\overline{\delta}}{16d}.
\end{equation}
For every $j\in [d]$, we define 
\begin{equation*}
   R_{\mu_j,\delta}  :=  2 \left(\left(\frac{M_{\mu, p}}{\delta}\right)^{1/p} + \eps_{m,d,\delta, \overline{\delta}}(\mu_j) \vee \eps_{m,d,1-\delta, \overline{\delta}}(\mu_j)\right),
\end{equation*}
\begin{equation*}
    R_{\nu_j,\delta}  :=  2 \left(\left(\frac{M_{\nu, p}}{\delta}\right)^{1/p} + \eps_{m',d,\delta, \overline{\delta}}(\nu_j) \vee \eps_{m',d,1-\delta, \overline{\delta}}(\nu_j)\right).
\end{equation*}
We further define
\begin{equation}\label{Rdef}
    R_{\max} := \max_{j \in [d]} \{R_{\mu_j, \delta}\} \vee \max_{j \in [d]} \{R_{\nu_j, \delta}\}.
\end{equation}

\medskip 
\noindent The next theorem quantifies the convergence rate of the empirical trimmed \msw distance.

\begin{thm}[Convergence Rate of \msw Distance]\label{theory:msw-conv-rate}
Suppose that the assumptions given above hold. For any $\bar{\delta}\in (0,1)$, with probability at least $1-\bar{\delta}$, we have
\begin{equation*}
    \left|\msw_{p, \delta}(\hat{\mu}_m, \hat{\nu}_{m'})-\msw_{p, \delta}(\mu, \nu)\right|\leq t_{\msw},
\end{equation*}
where
\begin{align*}
     t_{\msw}:=& \left\{2\lambda  R_{\max} \left(\frac{p}{1-2\delta} \sqrt{\log(  16 d / \overline{\delta})} \right)^{1/p} + \frac{2(1-\lambda)C_p\big(\mathbb{E}_{Z\sim \mu}[\|Z\|^2]^{1\slash 2}\vee \mathbb{E}_{Z\sim \nu}[\|Z\|^2]^{1\slash 2}\big)}{\sqrt{\delta}(1-2\delta)^{1\slash p}\overline{\delta}}\right\}\nonumber\\
     &\quad\cdot\big(m^{-1/(2p)} + m'^{-1/(2p)}\big),
\end{align*}
where $R_{\max}$ is as defined in Eq.\ \eqref{Rdef} and $C_p>0$ is a constant that depends only on $p$.  
\end{thm}

\begin{rmk}\label{rmk:msw-conv-rate}
Theorem \ref{theory:msw-conv-rate} states that the empirical trimmed \(\msw_{p,\delta}\) distance between two samples converges to the true population trimmed $\msw_{p, \delta}$ distance at the rate
\begin{align}\label{eq:msw-parametric-rate}
    t_{\msw}
=O\bigl(m^{-1/(2p)}+m'^{-1/(2p)}\bigr).
\end{align}
In particular, when $p=1$ and $m=m'$, this recovers the familiar $O(m^{-1/2})$ parametric rate. 
\end{rmk}

\subsection{Theoretical properties of the \abi posterior}\label{theory:abi-property}

In this section, we investigate theoretical properties of the \abi posterior. First, we prove that the oracle \abi posterior converges to the true posterior distribution as $\epsilon\downarrow 0$. 
We additionally establish that the \msw distance is continuous with respect to ABC posteriors in the sense that this distance vanishes in the limit as $\epsilon\downarrow 0$ through a novel martingale-based technique. 

\begin{thm}[Convergence of the \abi Posterior]\label{thm:abi-conv}
Let $(\Theta,X)$ be defined on a probability space $(\Xi,\F,\p)$ with 
$\Theta\in \Omega$ where $\Omega\subseteq\R^{d}$ is a Polish parameter space with Borel $\sigma$-algebra $\mathcal{B}$, and $X\in\X^{n}$ where $\X^{n}\subseteq\R^{nd_X}$ is a Polish observation space with Borel $\sigma$-algebra $\mathcal{A}$. Let $p \in [1, \infty)$ and assume
\begin{align*}
    \Theta\sim\pi, \quad X\mid\Theta\sim P_\Theta^{(n)}. 
\end{align*}
Assume that the joint distribution of $(\Theta,X)$ (denoted by $P_{\Theta,X}$) admits the density $f_{\Theta,X}$, the marginal distribution of $X$ (denoted by $P_X$) admits the density $f_X$, and $P_\Theta^{(n)}$ admits the density $f_{X|\Theta}$, all with respect to Lebesgue measure. Let $\xobs\in\X^{n}$ be such that $f_X(\xobs)>0$ and $f_X$ is continuous at $\xobs$. Suppose
\begin{align*}
M := \sup_{x\in\X^{n}}\int_{\Omega}\|\theta\|_p^p\,
               f_{\Theta, X}(\theta, x)\,d\theta<\infty.
\end{align*}
    Then as $\epsilon \downarrow 0$, the oracle \abi posterior, with density
    \begin{align*}
        \pi^{(\epsilon)}_{\abi}(\theta \mid \xobs) &:= \pi_{\Theta \mid X}(\theta \mid \msw{p}(\pi_{\Theta \mid x},\pi_{\Theta \mid x^{*}})\leq\epsilon)\nonumber\\
        &= \frac{\pi(\theta)\int_{\mathcal{X}} f_{X \mid \Theta}(x \mid \theta)\ind\{\msw{p}(\pi_{\Theta  \mid  x}, \pi_{\Theta  \mid  \xobs}) \leq \epsilon\}\, dx}{\int_{\Theta\times\mathcal{X}}\pi(\theta)f_{X \mid \Theta}(x \mid \theta)\ind\{\msw{p}(\pi_{\Theta  \mid  x}, \pi_{\Theta  \mid  \xobs}) \leq \epsilon\}\, d\theta \, dx}
    \end{align*}
    converges weakly in $\mathcal{P}_p(\Omega)$ to the true posterior distribution $\pi_{\Theta \mid X}(\theta \mid \xobs)$.
\end{thm}

\begin{thm}[Continuity of the ABC Posterior under the $\msw_p$ Distance]\label{thm:cont-msw-dist}
Let $\Theta, X, x^{*}$ be defined similarly as and satisfy all the assumptions in Theorem \ref{thm:abi-conv}.
    For any $\epsilon>0$, define $B_{\epsilon}(\xobs)=\{x:\|x-\xobs\|\leq \epsilon\}$. Then for any decreasing sequence $\epsilon_t \downarrow 0$, 
\begin{align*}
    \pi_{\Theta\mid X}\bigl(d\theta\mid X\in B_{\epsilon_t}(\xobs)\bigr) \Rightarrow
\pi_{\Theta\mid X}(d\theta\mid X=\xobs)
\quad \text{as }t\to\infty,
\end{align*}
where we use $\Rightarrow$ to denote weak convergence in $\mathcal P_{p}(\Omega)$. Consequently,
\begin{align*}
    \lim_{t\to\infty}
\msw_{p}\!\Bigl(
  \pi_{\Theta\mid X}(d\theta\mid X\in B_{\epsilon_t}(\xobs)),
  \pi_{\Theta\mid X}(d\theta\mid X=\xobs)
\Bigr)=0 .
\end{align*}
\end{thm}
\begin{rmk}
    Contrary to the standard convergence proofs that rely on the Lebesgue differentiation theorem (\citealp{barber2015rate, biau2015new, prangle2017adapting}), we establish the convergence of the ABC posterior by leveraging martingale techniques. 
    To the best of the authors' knowledge, this represents the first convergence proof for ABC that employs a martingale-based method (specifically, leveraging L\'evy's 0--1 law).
\end{rmk}

\section{Empirical Evaluation}\label{sec:sim}
In this section, we present extensive empirical evaluations demonstrating the efficacy of \abi across a broad range of simulation scenarios. 
We benchmark the performance of \abi against four widely used alternative methods: ABC with the 2-Wasserstein distance (WABC; see \citealt{bernton2019approximate}\footnote{We use the implementation available at \url{https://github.com/pierrejacob/winference}}), 
ABC with automated neural summary statistic (ABC-SS; see \citealt{jiang2017learning}),
Sequential Neural Likelihood Approximation (SNLE; see \citealt{papamakarios2019sequential}), and Sequential Neural Posterior Approximation (SNPE; see \citealt{greenberg2019automatic}). For SNLE and SNPE, we employ the implementations provided by the Python \texttt{SBI} package \citep{tejero-cantero2020sbi}. In the Multimodal Gaussian example (Section \ref{simulation:complex-gaussian-toy-model}), we additionally compare \abi against the Wasserstein generative adversarial network with gradient penalty (WGAN-GP; see \citealt{gulrajani2017improved}). 
We summarize the key characteristics of each method below:
\vspace{.8em}
\begin{table}[!htp]
\centering
\renewcommand{\arraystretch}{1.2}
\begin{small}
\begin{tabularx}{0.8\textwidth}{l|ccccc}
\toprule
\hspace{.5em} Compatibility & \hspace{.6em} \textbf{\abi} \hspace{.1em} & WABC & ABC-SS & SNLE & SNPE \\
\midrule
\hspace{.8em}Intractable Prior      & \hspace{.3em}\textbf{Yes}  & No   & Yes    & No   & No \\
\hspace{.8em}Intractable Likelihood & \hspace{.3em}\textbf{Yes}  & Yes  & Yes    & Yes  & Yes \\
\hspace{.8em}ABC-based              & \hspace{.3em}\textbf{Yes}  & Yes  & Yes    & No   & No \\
\bottomrule
\end{tabularx}
\end{small}
\caption{\small Compatibility comparison of different inference methods. \abi provides full compatibility with both intractable priors and intractable likelihood, offering enhanced modeling flexibility.}
\label{tab:method_comparison}
\end{table}

\vspace{-.5em}

\subsection{Multimodal Gaussian Model with Complex Posterior}\label{simulation:complex-gaussian-toy-model}
For the initial example, we consider a model commonly employed in likelihood‐free inference (see \citealp{papamakarios2019sequential, wang2022adversarial}), which exposes the intrinsic fragility of traditional ABC methods even in a seemingly simple scenario.
In this setup, $\theta$ is a 5-dimensional vector drawn according to $\mathrm{Unif}(-3, 3)$; for each $\theta \in \R^5$, we observe four i.i.d.\ sets of bivariate Gaussian samples $X$, where the mean and covariance of these samples are determined by $\theta$. For simplicity, we will subsequently treat $X$ as a flattened 8-dimensional vector. The forward sampling model is defined as follows:

\begin{align*}
    \theta_k &\sim \mathrm{Unif}(-3, 3), \quad k = 1, \dots, 5,\\
    \mu_\theta &= (\theta_1, \theta_2)^\top,\\
    s_1 = \theta^2_3, ~~ s_2 &= \theta_4^2, ~~\rho = \tanh(\theta_5),\\
    \Sigma_\theta &= \begin{pmatrix}
        s_1^2 & \rho s_1 s_2\\
        \rho s_1 s_2 & s_2^2
    \end{pmatrix},\\
    X_j \mid \theta &\sim \N(\mu_\theta, \Sigma_\theta), \quad j = 1, \dots, 4.
\end{align*}
Despite its structural simplicity, this model yields a complex posterior distribution characterized by truncated support and four distinct modes that arise from the inherent unidentifiability of the signs of $\theta_3$ and $\theta_4$.

We implemented \abi with two sequential iterations using adaptively selected thresholds.
The \msw distance was evaluated using 10 quantiles and five SW slices. 
For SNLE and SNPE, we similarly conducted two-round sequential inference. 
To ensure fair comparison, we calibrated the training budget for ABC-SS\footnote{Since both \abi and ABC-SS are rejection-ABC-based, we applied the same adaptive rejection quantile thresholds for \abi (iteration 1) and ABC-SS.}, WABC, and WGAN to match the total number of samples utilized for training across both \abi iterations.

\begin{figure}[!htp]
\vspace{1em}
\begin{center}
    \includegraphics[width=1.05\columnwidth]{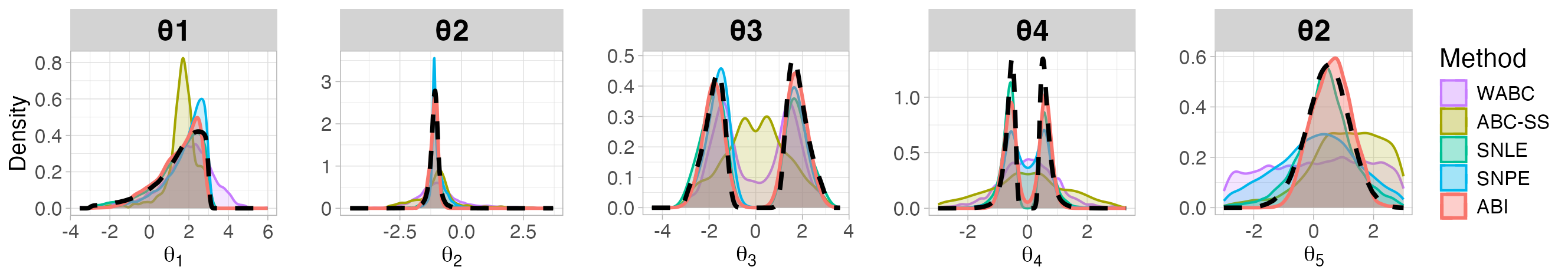}
\caption{Comparison of approximate posterior densities obtained from \abi and alternative benchmark methods under the Multimodal Gaussian model. The true posterior is shown in the black dashed line. \abi produced posteriors that accurately align with the true posterior distribution.}
\label{fig:toy-model-full}
\end{center}
\end{figure}

Figures \ref{fig:toy-model-full} and \ref{fig:toy-model-wgan-abc} present comparative analyses of posterior marginal distributions generated by \abi (in red) and alternative inference methods, with the true posterior distribution (shown in black) obtained via the No-U-Turn Sampler (NUTS) implemented in \texttt{rstan} using 10 MCMC chains.
We illustrate the evolution of the \abi posterior over iterations in Figure \ref{fig:toy-model-sample-path}.
Notably, ABC-SS produces a predominantly unimodal posterior distribution centered around the posterior mean, illustrating its fundamental limitation of yielding only first-order sufficient statistics (i.e., mean-matching) in the asymptotic regime with vanishing tolerance.
WGAN partially captures the bimodality of $\theta_3$ and $\theta_4$, yet produces posterior samples that significantly deviate from the true distribution.
The parameter $\theta_5$ poses the greatest challenge for accurate estimation across methods.
Overall, \abi generates samples that closely approximate the true posterior distribution across all parameters.

\begin{figure}[!htp]
\vspace{1em}
\begin{center}
    \includegraphics[width=1.05\columnwidth]{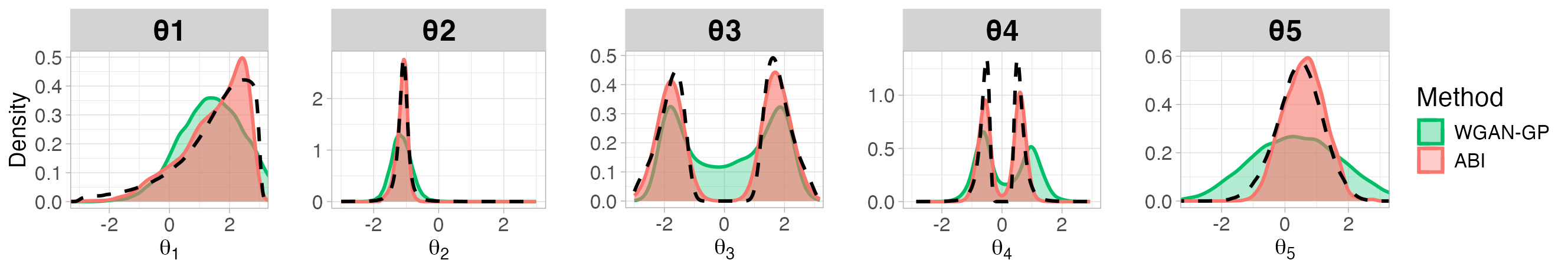}
\caption{Comparison of marginal posteriors generated by \abi and WGAN-GP. The true posterior is shown in the black dashed line.}
\label{fig:toy-model-wgan-abc}
\end{center}
\end{figure}

\begin{figure}[!htp]
\vspace{1em}
\begin{center}
    \includegraphics[width=1.00\columnwidth]{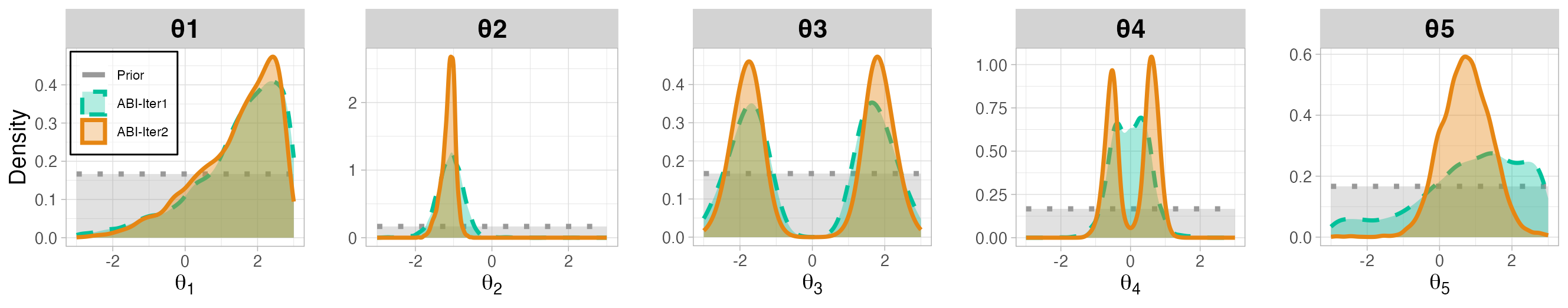}
\caption{Evolution of the sample path over successive iterations of \abi.}
\label{fig:toy-model-sample-path}
\end{center}
\end{figure}

Table \ref{table:toy-comparison} presents a quantitative comparison using multiple metrics: maximum mean discrepancy with Gaussian kernel, empirical 1-Wasserstein distance\footnote{The $W_1$ distance is computed using the Python Optimal Transport (\texttt{POT}) package.}, bias in posterior mean (measured as absolute difference between posterior distributions), and bias in posterior correlation (calculated as summed absolute deviation between empirical correlation matrices). \abi consistently demonstrates superior performance across the majority of evaluation criteria (with the exception of parameters $\theta_2$ and $\theta_5$),

\vspace{1.5em}
\begin{table}[!htp]
\centering
\renewcommand{\arraystretch}{1.2}
\begin{small}
\begin{tabularx}{\textwidth}{l l|cccccc}
\toprule
\multicolumn{1}{l}{\hspace{0em}\textbf{Evaluation Metric}} 
  & \multicolumn{1}{c|}{(Parameter)} 
    & \abi & WABC & ABC-SS & SNLE & SNPE & WGAN \\
\midrule
\hspace{0.8em}MMD       &      & \textbf{0.466} & 0.592 & 0.573 & 0.511 & 0.536 & 0.514 \\
\hspace{0.8em}$1$-Wasserstein       &     & \textbf{0.609} & 2.738 & 1.663 & 0.912 & 1.126 & 1.079 \\
\midrule
\multirow{5}{*}{\makecell[l]{\hspace{0.8em}Bias\\\hspace{0.8em}(Posterior\\\hspace{0.8em}Mean)}} 
          &\hspace{1.7em} $\theta_1$ & \textbf{0.001}  & 0.383  & 0.18  & 0.033  & 0.363  & 0.039  \\
          &\hspace{1.7em} $\theta_2$ & 0.030  & 0.193  & 0.103  & \textbf{0.001}  & 0.112  & 0.049  \\
          &\hspace{1.7em} $\theta_3$ & \textbf{0.016}  & 0.058  & 0.076  & 0.24  & 0.192  & 0.030  \\
          &\hspace{1.7em} $\theta_4$ & \textbf{0.006}  & 0.022  & 0.007  & 0.084  & 0.012  & 0.143  \\
          &\hspace{1.7em} $\theta_5$ & 0.137  & 0.345  & 0.642  & \textbf{0.078}  & 
 0.361  & 0.087  \\
\midrule
\hspace{0.8em}Bias (Posterior Corr.) & & \textbf{0.881} & 1.776 & 1.382 & 1.146 & 1.094 & 2.340 \\
\bottomrule
\end{tabularx}
\end{small}
\caption{\small Comparative performance evaluation of inference methodologies. Lower values indicate superior accuracy; best results are highlighted in bold.}
\label{table:toy-comparison}
\end{table}

\subsection{M/G/1 Queuing Model}
We now turn to the M/G/1 queuing model \citep{fearnhead2012constructing, bernton2019approximate}. 
This system illustrates a setting where, despite dependencies among observations, the model parameters remain identifiable from the marginal distribution of the data.

In this model, customers arrive at a single server with interarrival times $W_i\sim\mathrm{Exp}(\theta_3)$ (where $\theta_3$ represents the arrival rate) and model the service times as $U_i\sim\mathrm{Unif}(\theta_1,\theta_2)$. Rather than observing $W_i$ and $U_i$ directly, we record only the interdeparture times, defined through the following relationships:
\begin{align*}
V_i &= \sum_{j=1}^i W_j \qquad \text{(arrival time of $i$-th customer)},\\
X_i &= \sum_{j=1}^i Y_j, \,\, X_0 \equiv 0 \qquad \text{(departure time of $i$-th customer)},\\
Y_i &= U_i + \max\Bigl\{0,\sum_{j=1}^i W_j - \sum_{j=1}^{i-1}Y_j\Bigr\} = U_i + \max(0, V_i - X_{i-1}).
\end{align*}
We assume that the queue is initially empty before the first customer arrives.
We assign the following truncated prior distributions:
\begin{align*}
\theta_1,\;\theta_2 \sim \mathrm{Unif}(0,10),
\qquad
\theta_3 \sim \mathrm{Unif}\bigl(0,\tfrac{1}{3}\bigr).
\end{align*}
For our analysis, we use the dataset from \cite{shestopaloff2014bayesian}\footnote{This corresponds to the \texttt{Intermediate} dataset, with true posterior means provided in Table 4 of \cite{shestopaloff2014bayesian}.}, which was simulated with true parameter values $(\theta_1,\;\theta_2-\theta_1,\;\theta_3)=(4,3,0.15)$ and consists of $n=50$ observations.
Sequential versions of all algorithms were executed using four iterations with 10,000 training samples per iteration. As in the first example, to ensure fair comparison, we allocated an equal total number of training samples to WABC and ABC-SS as provided to \abi, SNLE, and SNPE. Figure \ref{fig:mg1-marginal-posteriors} presents the posterior distributions for parameters $\theta_1$, $\theta_2 - \theta_1$, and $\theta_3$, with the true posterior mean indicated by a black dashed line. 
The results demonstrate that the approximate posteriors produced by \abi (in red) not only align most accurately with the true posterior means, but also concentrate tightly around them.

\begin{figure}[!htp]
    \vspace{1em}
    \centering
    \includegraphics[width=1.0\linewidth]{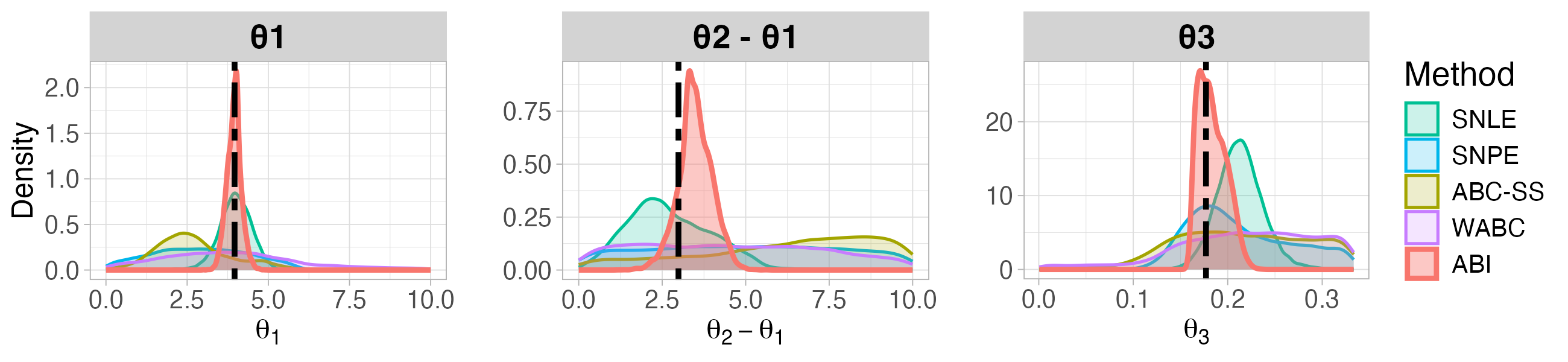}
    \caption{Comparison of approximate posterior densities under the M/G/1 queuing example. The dashed black line indicates the true posterior mean at $(3.96, \,2.99, \,0.177)$. \abi outperforms alternative approaches and exhibits superior alignment with the true posterior mean. }
    \label{fig:mg1-marginal-posteriors}
\end{figure}

\begin{figure}[!htp]
\begin{center}
    \includegraphics[width=0.8\columnwidth]{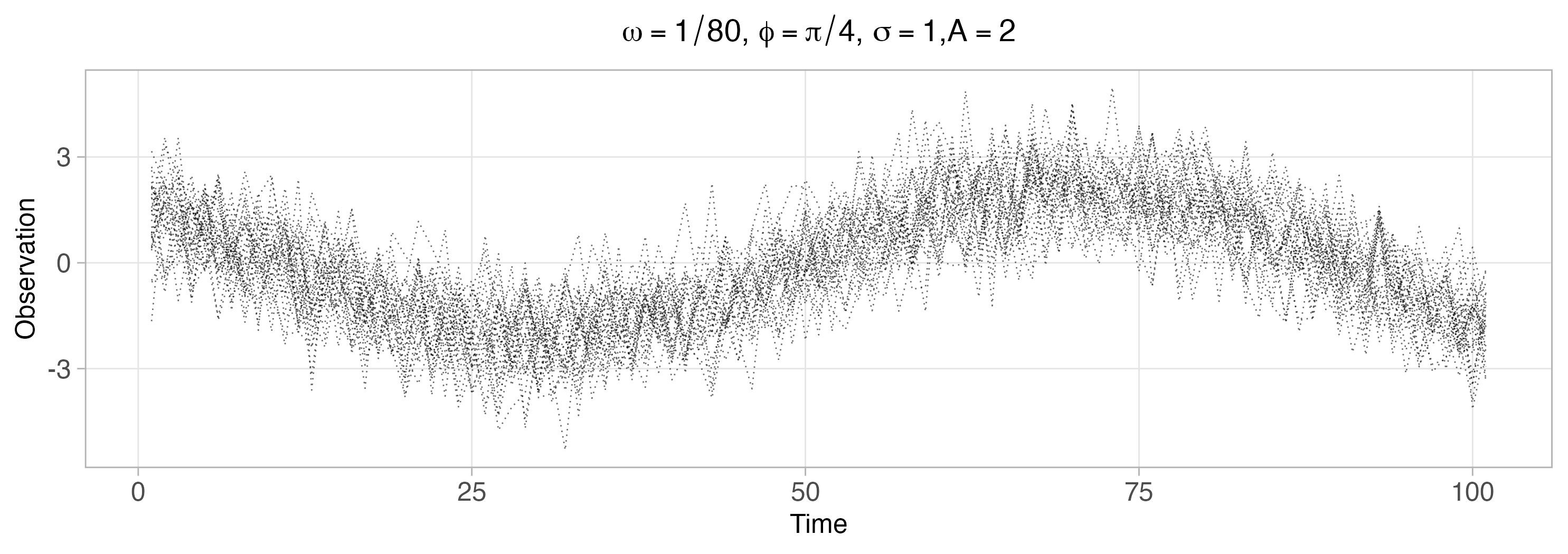}
\caption{30 trajectories simulated from the cosine model with parameter values $\omega^* = 1/80$, $\phi^* = \pi/4$, $\log(\sigma^*) = 0$, and $\log(A^*) = \log(2)$.}
\label{fig:cosine-model-obs-data}
\end{center}
\end{figure}

\subsection{Cosine Model}
In the third demonstration, we examine the cosine model \citep{bernton2019approximate}, defined as:
\begin{align*}
    Y_t &= A \cos (2 \pi \omega t + \phi) + \sigma \epsilon_t, \quad \epsilon_t \sim \N(0, 1) \quad \text{for } t\geq 1,
\end{align*}
with prior distributions:
\begin{align*}
    \omega \sim \mathrm{Unif}[0, 1/10], \quad \phi \sim \mathrm{Unif}[0, 2\pi],\quad \log(\sigma) \sim \N(0, 1), \quad \log(A) \sim \N(0, 1).
\end{align*}

Posterior inference for these parameters is challenging because information about $\omega$ and $\phi$ is substantially obscured in the marginal empirical distribution of observations $(Y_1, \dots, Y_n)$.
The observed data was generated with $t=100$ time steps using parameter values $\omega^* = 1/80$, $\phi^* = \pi/4$, $\log(\sigma^*) = 0$, and $\log(A^*) = \log(2)$; 30 example trajectories are displayed in Figure \ref{fig:cosine-model-obs-data}.
The exact posterior distribution was obtained using the NUTS algorithm.
Sequential algorithms (\abi, SNLE, SNPE) were executed with three iterations, each utilizing 5,000 training samples, while WABC\footnote{Wasserstein distance was computed by treating each dataset as a flattened vector containing 100 independent one-dimensional observations.} and ABC-SS were trained with a total of 15,000 simulations.

\begin{figure}[!htp]
\vspace{1em}
\begin{center}
    \includegraphics[width=1.03\columnwidth]{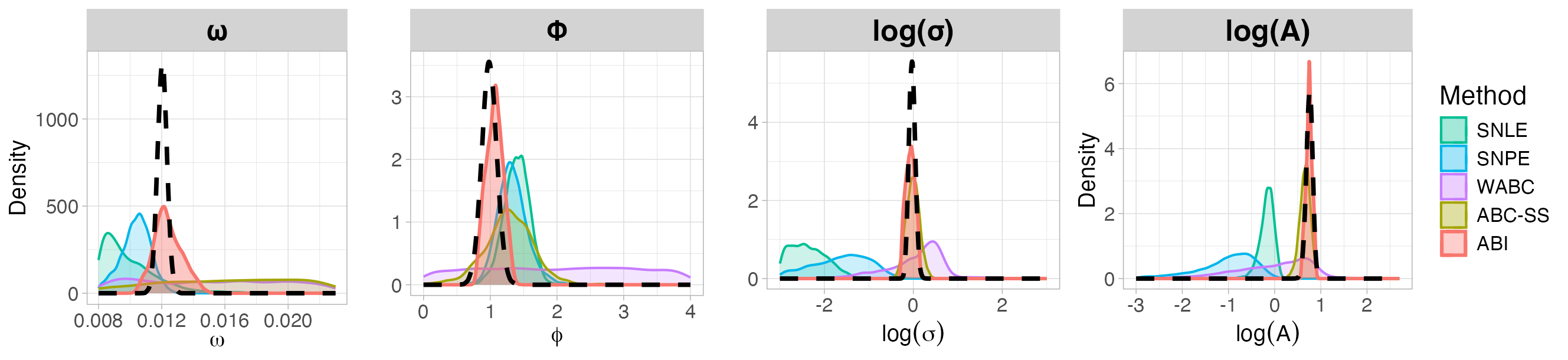}
\caption{Comparison of approximate posterior densities under the cosine model. The true posterior density is displayed as a dashed black line. Among all methods, \abi achieves the most accurate representation of the true posterior.}
\label{fig:cosine-model}
\end{center}
\end{figure}

\begin{figure}[!htp]
\vspace{1em}
\begin{center}
    \includegraphics[width=1.03\columnwidth]{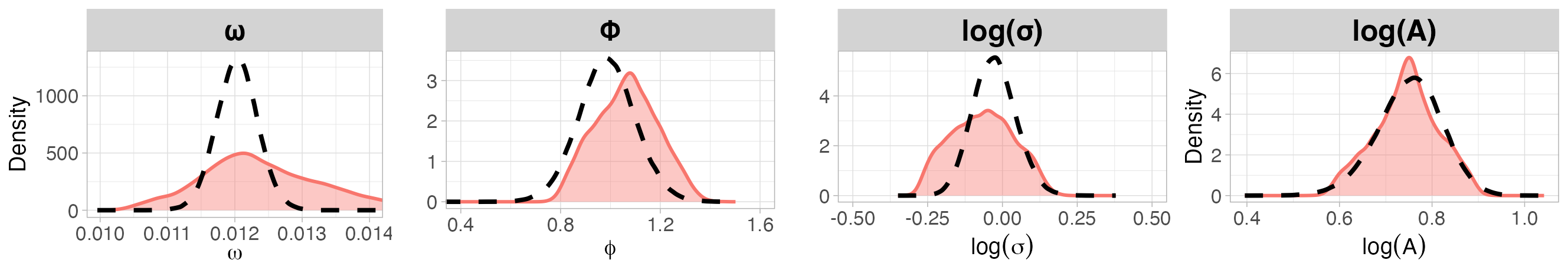}
\caption{Approximate posterior densities generated by \abi under the cosine model.}
\label{fig:cosine-model-abi}
\end{center}
\end{figure}

Figure \ref{fig:cosine-model} compares the approximate posterior distributions obtained from \abi and alternative methods, with the true posterior shown in black.
Among these parameters, $\omega$ indeed proves to be the most difficult.
We observe that \abi again yields the most satisfactory approximation across all four parameters. For clarity, we additionally provide a direct comparison between the \abi-generated posterior and the true posterior in Figure \ref{fig:cosine-model-abi}.

\subsection{Lotka-Volterra Model}

The final simulation example investigates the Lotka-Volterra (LV) model \citep{din2013dynamics}, which involves a pair of first-order nonlinear differential equations.
This system is frequently used to describe the dynamics of biological systems involving interactions between two species: a predator population (represented by $y_t$) and a prey population (represented by $x_t$). The populations evolve deterministically according to the following set of equations:
\begin{align*}
    \frac{dx}{dt} &= \alpha x - \beta x y, \quad \frac{dy}{dt} = -\gamma y + \delta x y.
\end{align*}
The changes in population states are governed by four parameters $(\alpha, \beta, \gamma, \delta)$ controlling the ecological processes: prey growth rate $(\alpha)$, predation rate $(\beta)$, predator growth rate $(\delta)$, and predator death rate $(\gamma)$.

\vspace{0.5em}
\begin{figure}[!htp]
\vspace{1em}
\begin{center}
    \includegraphics[width=0.97\columnwidth]{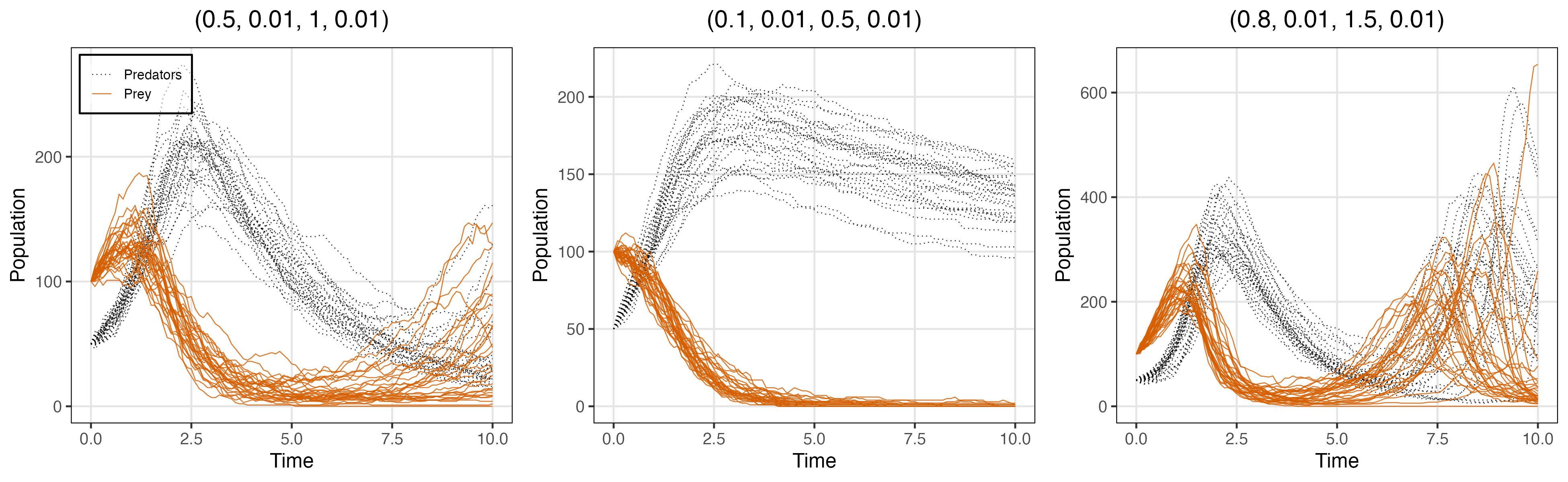}
\caption{30 trajectories sampled from the Lotka-Volterra model, each corresponding to one of three distinct parameter configurations $(\alpha^*, \beta^*, \gamma^*, \delta^*)$.}
\label{fig:lv-model-obs-data}
\end{center}
\end{figure}

Due to the absence of a closed-form transition density, inference in LV models using traditional methods presents significant challenges; however, these systems are particularly well-suited for ABC approaches since they allow for efficient generation of simulated datasets. 
The dynamics can be simulated using a discrete-time Markov jump process according to the Gillespie algorithm \citep{gillespie1976general}. 
At time $t$, we evaluate the following rates,
\begin{gather*}
    r_\alpha(t) = \alpha x_t, \quad r_\beta(t) = \beta x_t y_t, \quad r_\gamma(t) = \gamma y_t, \quad r_\delta(t) = \delta x_t y_t,\\
    r_\Delta(t) \coloneqq  r_\alpha(t) + r_\beta(t) + r_\gamma(t) + r_\delta(t)
\end{gather*}
The algorithm first samples the waiting time until the next event from an exponential distribution with parameter $r_\Delta(t)$, then selects one of the four possible events (prey death/birth, predator death/birth) with probability proportional to its own rate $r_{\boldsymbol{\cdot}}(t)$. 
We choose the truncated uniform prior over the restricted domain $[0, 1] \times [0.0.1]\times [0,2] \times [0,0.1]$.
The true parameter values are $(\alpha^*, \beta^*, \gamma^*, \delta^*) = (0.5, 0.01, 1.0, 0.01)$ \citep{papamakarios2019sequential}.
Posterior estimation in this case is particularly difficult because the likelihood surface contains concentrated probability mass in isolated, narrow regions throughout the parameter domain.
We initialized the populations at time $t=0$ with $(x_0, y_0) = (50, 100)$ and recorded system states at $0.1$ time unit intervals over a duration of $10$ time units, yielding a total of 101 observations as illustrated in Figure \ref{fig:lv-model-obs-data}.

We carried out sequential algorithms across two rounds, with each round utilizing 5,000 training samples. For fair comparison, WABC and ABC-SS were provided with a total of 10,000 training samples.
The resulting posterior approximations are presented in Figure \ref{fig:lv-model}, with black dashed lines marking the true parameter values, as the true posterior distributions are not available for this model.
The results again demonstrate that \abi consistently provides accurate approximations across all four parameters of the model. 

\begin{figure}[!htp]
\vspace{1em}
\begin{center}
    \includegraphics[width=1.03\columnwidth]{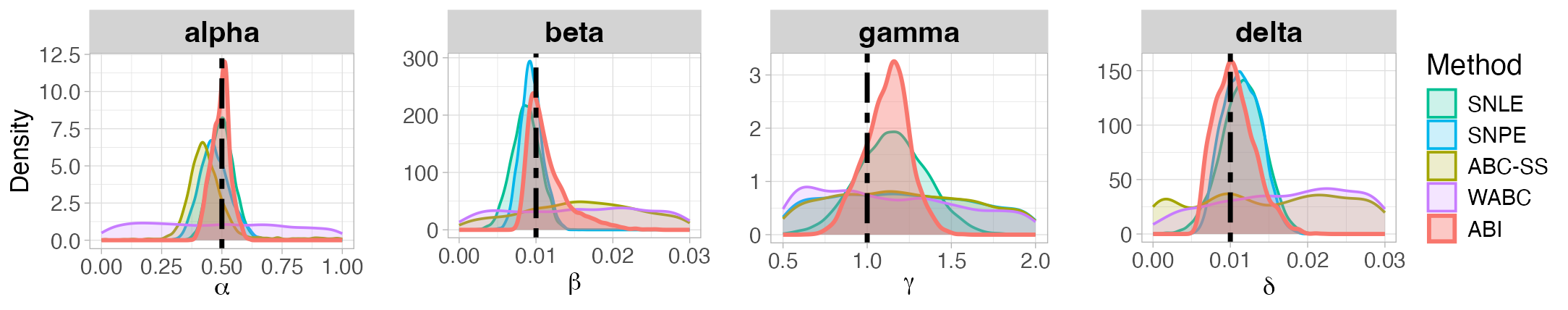}
\caption{Comparison of approximate posterior distributions for the Lotka-Volterra model. True parameter values are indicated by the dashed black lines.}
\label{fig:lv-model}
\end{center}
\end{figure}


\section{Discussion}\label{sec:discussion}

In this work, we introduce the Adaptive Bayesian Inference (\abi) framework, which shifts the focus of approximate Bayesian computation from data‐space discrepancies to direct comparisons in posterior space.
Our approach leverages a novel Marginally-augmented Sliced Wasserstein (\msw) distance---an integral probability metric defined on posterior measures that combines coordinate-wise marginals with random one-dimensional projections. 
We then establish a quantile-representation of \msw that reduces complex posterior comparisons to a tractable distributional regression task.
Moreover, we propose an adaptive rejection‑sampling scheme in which each iteration's proposal is updated via generative modeling of the accepted parameters in the preceding iteration.
The generative modeling–based proposal updates allow \abi to perform posterior approximation without explicit prior density evaluation, overcoming a key limitation of sequential Monte Carlo, population Monte Carlo, and neural density estimation approaches. 
Our theoretical analysis shows that \msw retains the topological guarantees of the classical Wasserstein metric, including the ability to metrize weak convergence, while achieving parametric convergence rates in the trimmed setting when $p=1$. 
The martingale‐based proof of sequential convergence offers an alternative to existing Lebesgue differentiation arguments and may find broader application in the analysis of other adaptive algorithms.

Empirically, \abi delivers substantially more accurate posterior approximations than Wasserstein ABC and summary-based ABC, as well as state-of-the-art likelihood‑free simulators and Wasserstein GAN.
Through a variety of simulation experiments, we have shown that the posterior \msw distance remains robust under small observed sample sizes, intricate dependency structures, and non‐identifiability of parameters. 
Furthermore, our conditional quantile‐regression implementation exhibits stability to network initialization and requires minimal tuning.

Several avenues for future research arise from this work.
One promising direction is to integrate our kernel‑statistic approach into Sequential Monte Carlo algorithms to improve efficiency in likelihood‑free settings.
Future work could also apply \abi to large‑scale scientific simulators, such as those used in systems biology, climate modeling, and cosmology, to spur domain‑specific adaptations of the posterior‑matching paradigm.

In summary, \abi offers a new perspective on approximate Bayesian computation as well as likelihood‐free inference by treating the posterior distribution itself as the primary object of comparison.  
By combining a novel posterior space metric, quantile-regression–based estimation, and generative-model–driven sequential refinement, \abi significantly outperforms alternative ABC and likelihood-free methods. 
Moreover, this posterior‑matching viewpoint may catalyze further advances in approximate Bayesian computation and open new avenues for inference in complex, simulator‐based models.


\section*{Acknowledgements}
The authors would like to thank Naoki Awaya, X.Y. Han, Iain Johnstone, Tengyuan Liang, Art Owen, Robert Tibshirani, John Cherian, Michael Howes, Tim Sudijono, Julie Zhang, and Chenyang Zhong for their valuable discussions and insightful comments.
The authors would like to especially acknowledge Michael Howes and Chenyang Zhong for their proofreading of the technical results. 
W.S.L. gratefully acknowledges support from the Stanford Data Science Scholarship and the Two Sigma Graduate Fellowship Fund during this research. 5366
W.H.W.'s research was partially supported by NSF grant 2310788.

\bibliographystyle{plainnat}
\bibliography{refs}

\begin{thebibliography}{2}
\providecommand{\natexlab}[1]{#1}
\providecommand{\url}[1]{\texttt{#1}}
\expandafter\ifx\csname urlstyle\endcsname\relax
  \providecommand{\doi}[1]{doi: #1}\else
  \providecommand{\doi}{doi: \begingroup \urlstyle{rm}\Url}\fi

\bibitem[Gillespie(1976)]{gillespie1976general}
Daniel~T Gillespie.
\newblock A general method for numerically simulating the stochastic time evolution of coupled chemical reactions.
\newblock \emph{Journal of computational physics}, 22\penalty0 (4):\penalty0 403--434, 1976.

\bibitem[Shestopaloff and Neal(2014)]{shestopaloff2014bayesian}
Alexander~Y Shestopaloff and Radford~M Neal.
\newblock On bayesian inference for the m/g/1 queue with efficient mcmc sampling.
\newblock \emph{arXiv preprint arXiv:1401.5548}, 2014.

\end{thebibliography}


\begin{thebibliography}{53}
\providecommand{\natexlab}[1]{#1}
\providecommand{\url}[1]{\texttt{#1}}
\expandafter\ifx\csname urlstyle\endcsname\relax
  \providecommand{\doi}[1]{doi: #1}\else
  \providecommand{\doi}{doi: \begingroup \urlstyle{rm}\Url}\fi

\bibitem[Alsing et~al.(2018)Alsing, Wandelt, and Feeney]{alsing2018massive}
Justin Alsing, Benjamin Wandelt, and Stephen Feeney.
\newblock Massive optimal data compression and density estimation for scalable, likelihood-free inference in cosmology.
\newblock \emph{Monthly Notices of the Royal Astronomical Society}, 477\penalty0 (3):\penalty0 2874--2885, 2018.

\bibitem[Alvarez-Esteban et~al.(2008)Alvarez-Esteban, Del~Barrio, Cuesta-Albertos, and Matran]{alvarez2008trimmed}
Pedro~C{\'e}sar Alvarez-Esteban, Eustasio Del~Barrio, Juan~Antonio Cuesta-Albertos, and Carlos Matran.
\newblock Trimmed comparison of distributions.
\newblock \emph{Journal of the American Statistical Association}, 103\penalty0 (482):\penalty0 697--704, 2008.

\bibitem[Barber et~al.(2015)Barber, Voss, and Webster]{barber2015rate}
Stuart Barber, Jochen Voss, and Mark Webster.
\newblock The rate of convergence for approximate bayesian computation.
\newblock \emph{Electronic Journal of Statistics.}, 2015.

\bibitem[Beaumont et~al.(2009)Beaumont, Cornuet, Marin, and Robert]{beaumont2009adaptive}
Mark~A Beaumont, Jean-Marie Cornuet, Jean-Michel Marin, and Christian~P Robert.
\newblock Adaptive approximate bayesian computation.
\newblock \emph{Biometrika}, 96\penalty0 (4):\penalty0 983--990, 2009.

\bibitem[Berger and Wolpert(1988)]{berger1988likelihood}
James~O Berger and Robert~L Wolpert.
\newblock The likelihood principle.
\newblock IMS, 1988.

\bibitem[Bernton et~al.(2019)Bernton, Jacob, Gerber, and Robert]{bernton2019approximate}
Espen Bernton, Pierre~E Jacob, Mathieu Gerber, and Christian~P Robert.
\newblock Approximate bayesian computation with the wasserstein distance.
\newblock \emph{Journal of the Royal Statistical Society Series B: Statistical Methodology}, 81\penalty0 (2):\penalty0 235--269, 2019.

\bibitem[Biau et~al.(2015)Biau, C{\'e}rou, and Guyader]{biau2015new}
G{\'e}rard Biau, Fr{\'e}d{\'e}ric C{\'e}rou, and Arnaud Guyader.
\newblock New insights into approximate bayesian computation.
\newblock In \emph{Annales de l'IHP Probabilit{\'e}s et statistiques}, volume~51, pages 376--403, 2015.

\bibitem[Bonassi and West(2015)]{bonassi2015sequential}
Fernando~V Bonassi and Mike West.
\newblock Sequential monte carlo with adaptive weights for approximate bayesian computation.
\newblock \emph{Bayesian Analysis}, 2015.

\bibitem[Bonnotte(2013)]{bonnotte2013unidimensional}
Nicolas Bonnotte.
\newblock \emph{Unidimensional and evolution methods for optimal transportation}.
\newblock PhD thesis, Universit{\'e} Paris Sud-Paris XI; Scuola normale superiore (Pise, Italie), 2013.

\bibitem[Cameron and Pettitt(2012)]{cameron2012approximate}
Ewan Cameron and AN~Pettitt.
\newblock Approximate bayesian computation for astronomical model analysis: a case study in galaxy demographics and morphological transformation at high redshift.
\newblock \emph{Monthly Notices of the Royal Astronomical Society}, 425\penalty0 (1):\penalty0 44--65, 2012.

\bibitem[Chatterjee et~al.(2021)Chatterjee, Sharma, Swisher, and Chatterjee]{chatterjee2021approximate}
Neel Chatterjee, Somya Sharma, Sarah Swisher, and Snigdhansu Chatterjee.
\newblock Approximate bayesian computation for physical inverse modeling.
\newblock \emph{arXiv preprint arXiv:2111.13296}, 2021.

\bibitem[Chatterjee et~al.(2025)Chatterjee, Hastie, and Tibshirani]{chatterjee2025univariate}
Sourav Chatterjee, Trevor Hastie, and Robert Tibshirani.
\newblock Univariate-guided sparse regression.
\newblock \emph{arXiv preprint arXiv:2501.18360}, 2025.

\bibitem[Chiachío-Ruano et~al.(2021)Chiachío-Ruano, Chiachío-Ruano, and Jalón]{chiachio2021solving}
Manuel Chiachío-Ruano, Juan Chiachío-Ruano, and María~L. Jalón.
\newblock Solving inverse problems by approximate bayesian computation.
\newblock In \emph{Bayesian Inverse Problems}. 2021.

\bibitem[Dabney et~al.(2018)Dabney, Ostrovski, Silver, and Munos]{dabney2018implicit}
Will Dabney, Georg Ostrovski, David Silver, and R{\'e}mi Munos.
\newblock Implicit quantile networks for distributional reinforcement learning.
\newblock In \emph{International conference on machine learning}, pages 1096--1105. PMLR, 2018.

\bibitem[Del~Moral et~al.(2012)Del~Moral, Doucet, and Jasra]{del2012adaptive}
Pierre Del~Moral, Arnaud Doucet, and Ajay Jasra.
\newblock An adaptive sequential monte carlo method for approximate bayesian computation.
\newblock \emph{Statistics and computing}, 22:\penalty0 1009--1020, 2012.

\bibitem[Din(2013)]{din2013dynamics}
Qamar Din.
\newblock Dynamics of a discrete lotka-volterra model.
\newblock \emph{Advances in Difference Equations}, 2013:\penalty0 1--13, 2013.

\bibitem[Drovandi et~al.(2024)Drovandi, Nott, and Frazier]{drovandi2024improving}
Christopher Drovandi, David~J Nott, and David~T Frazier.
\newblock Improving the accuracy of marginal approximations in likelihood-free inference via localization.
\newblock \emph{Journal of Computational and Graphical Statistics}, 33\penalty0 (1):\penalty0 101--111, 2024.

\bibitem[Fearnhead and Prangle(2011)]{fearnhead2011constructing}
Paul Fearnhead and Dennis Prangle.
\newblock Constructing abc summary statistics: semi-automatic abc.
\newblock \emph{Nature Precedings}, pages 1--1, 2011.

\bibitem[Fearnhead and Prangle(2012)]{fearnhead2012constructing}
Paul Fearnhead and Dennis Prangle.
\newblock Constructing summary statistics for approximate bayesian computation: semi-automatic approximate bayesian computation.
\newblock \emph{Journal of the Royal Statistical Society Series B: Statistical Methodology}, 74\penalty0 (3):\penalty0 419--474, 2012.

\bibitem[Frazier(2020)]{frazier2020robust}
David~T Frazier.
\newblock Robust and efficient approximate bayesian computation: A minimum distance approach.
\newblock \emph{arXiv preprint arXiv:2006.14126}, 2020.

\bibitem[Ghosh(2021)]{ghosh2021exponential}
Malay Ghosh.
\newblock Exponential tail bounds for chisquared random variables.
\newblock \emph{Journal of Statistical Theory and Practice}, 15\penalty0 (2):\penalty0 35, 2021.

\bibitem[Gillespie(1976)]{gillespie1976general}
Daniel~T Gillespie.
\newblock A general method for numerically simulating the stochastic time evolution of coupled chemical reactions.
\newblock \emph{Journal of computational physics}, 22\penalty0 (4):\penalty0 403--434, 1976.

\bibitem[Godambe(1968)]{godambe1968bayesian}
VP~Godambe.
\newblock Bayesian sufficiency in survey-sampling.
\newblock \emph{Annals of the Institute of Statistical Mathematics}, 20\penalty0 (1):\penalty0 363--373, 1968.

\bibitem[Greenberg et~al.(2019)Greenberg, Nonnenmacher, and Macke]{greenberg2019automatic}
David Greenberg, Marcel Nonnenmacher, and Jakob Macke.
\newblock Automatic posterior transformation for likelihood-free inference.
\newblock In \emph{International conference on machine learning}, pages 2404--2414. PMLR, 2019.

\bibitem[Gulrajani et~al.(2017)Gulrajani, Ahmed, Arjovsky, Dumoulin, and Courville]{gulrajani2017improved}
Ishaan Gulrajani, Faruk Ahmed, Martin Arjovsky, Vincent Dumoulin, and Aaron~C Courville.
\newblock Improved training of wasserstein gans.
\newblock \emph{Advances in neural information processing systems}, 30, 2017.

\bibitem[Huber(1964)]{huber1964robust}
Peter~J Huber.
\newblock Robust estimation of a location parameter.
\newblock \emph{The Annals of Mathematical Statistics}, 35\penalty0 (1):\penalty0 73--101, 1964.

\bibitem[Jiang(2018)]{jiang2018approximate}
Bai Jiang.
\newblock Approximate bayesian computation with kullback-leibler divergence as data discrepancy.
\newblock In \emph{International conference on artificial intelligence and statistics}, pages 1711--1721. PMLR, 2018.

\bibitem[Jiang et~al.(2017)Jiang, Wu, Zheng, and Wong]{jiang2017learning}
Bai Jiang, Tung-yu Wu, Charles Zheng, and Wing~H Wong.
\newblock Learning summary statistic for approximate bayesian computation via deep neural network.
\newblock \emph{Statistica Sinica}, pages 1595--1618, 2017.

\bibitem[Kallenberg and Kallenberg(1997)]{kallenberg1997foundations}
Olav Kallenberg and Olav Kallenberg.
\newblock \emph{Foundations of modern probability}, volume~2.
\newblock Springer, 1997.

\bibitem[Kim et~al.(2024)Kim, Zhai, and Ro{\v{c}}kov{\'a}]{kim2024deep}
Jungeum Kim, Percy~S Zhai, and Veronika Ro{\v{c}}kov{\'a}.
\newblock Deep generative quantile bayes.
\newblock \emph{arXiv preprint arXiv:2410.08378}, 2024.

\bibitem[Koenker and Bassett~Jr(1978)]{koenker1978regression}
Roger Koenker and Gilbert Bassett~Jr.
\newblock Regression quantiles.
\newblock \emph{Econometrica: journal of the Econometric Society}, pages 33--50, 1978.

\bibitem[Legramanti et~al.(2022)Legramanti, Durante, and Alquier]{legramanti2022concentration}
Sirio Legramanti, Daniele Durante, and Pierre Alquier.
\newblock Concentration of discrepancy-based abc via rademacher complexity.
\newblock \emph{arXiv preprint arXiv:2206.06991}, 2022.

\bibitem[Lu et~al.(2025)Lu, Zhong, and Wong]{lu2025generative}
Wenhui~Sophia Lu, Chenyang Zhong, and Wing~Hung Wong.
\newblock Efficient generative modeling via penalized optimal transport network.
\newblock \emph{arXiv preprint arXiv:2402.10456v2}, 2025.

\bibitem[Manole et~al.(2022)Manole, Balakrishnan, and Wasserman]{manole2022minimax}
Tudor Manole, Sivaraman Balakrishnan, and Larry Wasserman.
\newblock Minimax confidence intervals for the sliced wasserstein distance.
\newblock \emph{Electronic Journal of Statistics}, 16\penalty0 (1):\penalty0 2252--2345, 2022.

\bibitem[Marin et~al.(2012)Marin, Pudlo, Robert, and Ryder]{marin2012approximate}
Jean-Michel Marin, Pierre Pudlo, Christian~P Robert, and Robin~J Ryder.
\newblock Approximate bayesian computational methods.
\newblock \emph{Statistics and computing}, 22\penalty0 (6):\penalty0 1167--1180, 2012.

\bibitem[Markram et~al.(2015)Markram, Muller, Ramaswamy, Reimann, Abdellah, Sanchez, Ailamaki, Alonso-Nanclares, Antille, Arsever, et~al.]{markram2015reconstruction}
Henry Markram, Eilif Muller, Srikanth Ramaswamy, Michael~W Reimann, Marwan Abdellah, Carlos~Aguado Sanchez, Anastasia Ailamaki, Lidia Alonso-Nanclares, Nicolas Antille, Selim Arsever, et~al.
\newblock Reconstruction and simulation of neocortical microcircuitry.
\newblock \emph{Cell}, 163\penalty0 (2):\penalty0 456--492, 2015.

\bibitem[Moala and O’Hagan(2010)]{moala2010elicitation}
Fernando~A Moala and Anthony O’Hagan.
\newblock Elicitation of multivariate prior distributions: A nonparametric bayesian approach.
\newblock \emph{Journal of Statistical Planning and Inference}, 140\penalty0 (7):\penalty0 1635--1655, 2010.

\bibitem[Nadjahi et~al.(2019)Nadjahi, Durmus, Simsekli, and Badeau]{nadjahi2019asymptotic}
Kimia Nadjahi, Alain Durmus, Umut Simsekli, and Roland Badeau.
\newblock Asymptotic guarantees for learning generative models with the sliced-wasserstein distance.
\newblock \emph{Advances in Neural Information Processing Systems}, 32, 2019.

\bibitem[Padilla et~al.(2022)Padilla, Tansey, and Chen]{padilla2022quantile}
Oscar Hernan~Madrid Padilla, Wesley Tansey, and Yanzhen Chen.
\newblock Quantile regression with relu networks: Estimators and minimax rates.
\newblock \emph{Journal of Machine Learning Research}, 23\penalty0 (247):\penalty0 1--42, 2022.

\bibitem[Papamakarios and Murray(2016)]{papamakarios2016fast}
George Papamakarios and Iain Murray.
\newblock Fast $\varepsilon$-free inference of simulation models with bayesian conditional density estimation.
\newblock \emph{Advances in neural information processing systems}, 29, 2016.

\bibitem[Papamakarios et~al.(2019)Papamakarios, Sterratt, and Murray]{papamakarios2019sequential}
George Papamakarios, David Sterratt, and Iain Murray.
\newblock Sequential neural likelihood: Fast likelihood-free inference with autoregressive flows.
\newblock In \emph{The 22nd international conference on artificial intelligence and statistics}, pages 837--848. PMLR, 2019.

\bibitem[Polson and Sokolov(2023)]{polson2023generative}
Nicholas~G Polson and Vadim Sokolov.
\newblock Generative ai for bayesian computation.
\newblock \emph{arXiv preprint arXiv:2305.14972}, 2023.

\bibitem[Prangle(2017)]{prangle2017adapting}
Dennis Prangle.
\newblock Adapting the abc distance function.
\newblock \emph{Bayesian Analysis}, 2017.

\bibitem[Rabin et~al.(2012)Rabin, Peyr{\'e}, Delon, and Bernot]{rabin2012wasserstein}
Julien Rabin, Gabriel Peyr{\'e}, Julie Delon, and Marc Bernot.
\newblock Wasserstein barycenter and its application to texture mixing.
\newblock In \emph{Scale Space and Variational Methods in Computer Vision: Third International Conference, SSVM 2011, Ein-Gedi, Israel, May 29--June 2, 2011, Revised Selected Papers 3}, pages 435--446. Springer, 2012.

\bibitem[Shestopaloff and Neal(2014)]{shestopaloff2014bayesian}
Alexander~Y Shestopaloff and Radford~M Neal.
\newblock On bayesian inference for the m/g/1 queue with efficient mcmc sampling.
\newblock \emph{arXiv preprint arXiv:1401.5548}, 2014.

\bibitem[Talagrand(1994)]{talagrand1994transportation}
Michel Talagrand.
\newblock The transportation cost from the uniform measure to the empirical measure in dimension {$\geq 3$}.
\newblock \emph{The Annals of Probability}, pages 919--959, 1994.

\bibitem[Tavar{\'e}(2018)]{tavare2018history}
Simon Tavar{\'e}.
\newblock On the history of abc.
\newblock In \emph{Handbook of Approximate Bayesian Computation}, pages 55--69. Chapman and Hall/CRC, 2018.

\bibitem[Tejero-Cantero et~al.(2020)Tejero-Cantero, Boelts, Deistler, Lueckmann, Durkan, Gonçalves, Greenberg, and Macke]{tejero-cantero2020sbi}
Alvaro Tejero-Cantero, Jan Boelts, Michael Deistler, Jan-Matthis Lueckmann, Conor Durkan, Pedro~J. Gonçalves, David~S. Greenberg, and Jakob~H. Macke.
\newblock sbi: A toolkit for simulation-based inference.
\newblock \emph{Journal of Open Source Software}, 5\penalty0 (52):\penalty0 2505, 2020.
\newblock \doi{10.21105/joss.02505}.
\newblock URL \url{https://doi.org/10.21105/joss.02505}.

\bibitem[Villani et~al.(2009)]{villani2009optimal}
C{\'e}dric Villani et~al.
\newblock \emph{Optimal transport: old and new}, volume 338.
\newblock Springer, 2009.

\bibitem[Wang and Ro{\v{c}}kov{\'a}(2022)]{wang2022adversarial}
Yuexi Wang and Veronika Ro{\v{c}}kov{\'a}.
\newblock Adversarial bayesian simulation.
\newblock \emph{arXiv preprint arXiv:2208.12113}, 2022.

\bibitem[Wood(2010)]{wood2010statistical}
Simon~N Wood.
\newblock Statistical inference for noisy nonlinear ecological dynamic systems.
\newblock \emph{Nature}, 466\penalty0 (7310):\penalty0 1102--1104, 2010.

\bibitem[Zeng et~al.(2019)Zeng, Wang, Zhang, Cai, and Li]{zeng2019novel}
Yang Zeng, Hu~Wang, Shuai Zhang, Yong Cai, and Enying Li.
\newblock A novel adaptive approximate bayesian computation method for inverse heat conduction problem.
\newblock \emph{International Journal of Heat and Mass Transfer}, 134:\penalty0 185--197, 2019.

\bibitem[Zhou et~al.(2023)Zhou, Jiao, Liu, and Huang]{zhou2023deep}
Xingyu Zhou, Yuling Jiao, Jin Liu, and Jian Huang.
\newblock A deep generative approach to conditional sampling.
\newblock \emph{Journal of the American Statistical Association}, 118\penalty0 (543):\penalty0 1837--1848, 2023.

\end{thebibliography}


\newpage
\begin{appendices}

{\LARGE\bfseries\section*{Appendix}}
\noindent The Appendix section contains detailed proofs of theoretical results, supplementary results and descriptions on the simulation setup, 
as well as further remarks.

\section{Marginally-augmented Sliced Wasserstein (\msw) Distance}

In this section, we further discuss the Marginally-augmented Sliced Wasserstein (\msw) distance. We begin by presenting the formal definition of the untrimmed \msw distance.

\begin{defn}\label{def:msw-untrimmed}
Let $p \geq 1$ and $\mu, \nu \in \mathcal{P}_p(\R^d)$ with $d\geq 1$. The Marginally-augmented Sliced Wasserstein (\msw) distance between $\mu$ and $\nu$ is defined as
\begin{align}
\msw_{p}(\mu,\nu) = \lambda \, \underbrace{\frac{1}{d} \sum_{j=1}^{d} \W_{p} \Bigl(({e_j})_\# \mu, ({e_j})_\# \nu \Bigr)}_{\text{marginal augmentation}} + (1 - \lambda) \,  \underbrace{ \left (\E_{\varphi \sim \sigma} \Bigl[\W_{p}^p \left(\varphi_\# \mu, \varphi_\#\nu \right) \Bigr] \right)^{1/p} }_{\text{Sliced Wasserstein distance}},
\end{align}
where $\lambda \in (0, 1)$ is a mixing parameter, $\sigma(\cdot)$ denotes the uniform probability measure on the unit sphere $\smb^{d-1}$, and $({e_j})_\#$ represents the pushforward of a measure under projection onto the $j$-th coordinate axis.
\end{defn}

\section{Curse of Dimensionality in Rejection-\abc{ABC}}
In practice, observed data often involve numerous covariates, nonexchangeable samples, or dependencies between samples, resulting in a high-dimensional sample space.
Examples of such scenarios include the Lotka-Volterra model for modeling predator-prey interactions, the M/G/1 queuing model, and astronomical models of high-redshift galaxy morphology (\citealp{cameron2012approximate, bernton2019approximate}).
However, it is well known that distance metrics such as the Euclidean distance become highly unreliable in high dimensions, as observations concentrate near the hypersphere---a phenomenon called the \textit{curse of dimensionality}. 
In the following example, we demonstrate that even when the likelihood has fast decay, the number of samples needed to obtain an observation within an $\epsilon$-ball around $\xobs$ grows exponentially with the dimensionality of the observation space.

\begin{example}[High-dimensional Gaussian]
    For an illustrative example, consider the case when the prior distribution is uninformative, i.e.,
\begin{align*}
     \theta \in \R^d &\sim \mathrm{Unif}[-1, 1]^d ,\\
     X | \theta \in \R^n&\sim \N(\mathbf{H}\theta, \sigma^2 \mathbf{I}_n) ,
\end{align*}
where $\mathbf{H} \in \R^{n \times d}$ is a matrix that maps from the low-dimensional parameter space to the high-dimensional observation space. 
\end{example}

\begin{lemma}[Curse of dimensionality]\label{lemma:curse-of-dim-gaussian}
    The expected number of samples needed to produce a draw within the $\epsilon$-ball centered at $\xobs$ grows as  $\Omega(\exp(n))$ for $0 < \epsilon < 1/2(n + \|\mathbf{H}\theta - \xobs\|_2^2 / \sigma^2)^{1/2}$.
\end{lemma}

\begin{proof}
    Since $\bX |\theta \sim \N(\mathbf{H}\theta, \sigma^2 \mathbf{I}_n)$, we have $\bX - \xobs \mid \theta \sim \N(\mathbf{H} \theta - \xobs, \sigma^2 \mathbf{I}_n)$. 
Define $\Delta := \|\mathbf{H} \theta - \xobs\|_2$, so $\|\bX - \xobs\|_2^2 \mid \theta \sim \sigma^2 \chi_n^2(\zeta)$ with noncentrality parameter $\zeta = \Delta^2 / \sigma^2$. 
Therefore, the expected distance is given by:
\begin{align*}
    \E\left[\sigma^{-2}\|\bX - \xobs\|^2_2\right] = n + \sigma^{-2}(\|\mathbf{H}\theta\|_2^2 + \|\xobs\|_2^2  - 2(\xobs)^\top \mathbf{H}\theta) = n + \Delta^2 / \sigma^2,
\end{align*}
which scales linearly in $n$ as the dimension $n \to \infty$. 
Let $Y \sim \chi_n^2(\zeta)$, then for $0 < c < n + \zeta = n + \Delta^2 / \sigma^2$, by Theorem 4 of \cite{ghosh2021exponential} we have,
\begin{align*}
    \p(Y < n + \zeta -c ) \leq \exp \left(\frac{n}{2} \left( \frac{c}{n + 2\zeta} + \log \left(1 - \frac{c}{n + 2\zeta}\right)\right)\right) \leq \exp\left(- \frac{nc^2}{4(n + 2\zeta)^2}\right).
\end{align*}
For $0< \epsilon < 1/2\sqrt{n + \zeta}$, where $\zeta = \|\mathbf{H}\theta - \xobs\|_2^2 / \sigma^2$, we set $c := n + \zeta - \epsilon^2 > 0$. Applying the inequality above yields
\begin{align*}
    \p\left(\sigma^{-2}\|\bX - \xobs\|_2^2 < \epsilon^2 | \theta \right) \leq \exp\left(-\frac{n (n + \zeta - \sigma^2\epsilon^2)^2}{4 (n + 2\zeta)^2}\right) = O\left(\exp\left(-n\right)\right).
\end{align*}
Therefore, the expected number of samples needed to produce a draw within the $\epsilon$-ball centered at $\xobs$ is $\Omega(\exp(n))$; specifically, the number of required samples grows exponentially with the dimension $n$.
\end{proof}

\begin{lemma}[High-dimensional Bounded Density]\label{lemma:curse-of-dim-bounded}
    Let $\bX = (X_1, \dots, X_n) \in \R^n$ be a random vector with a joint density function $f(x_1, \dots, x_n)$ that is bounded by $K > 0$. 
    Then, for any $\epsilon > 0$, there exists a constant $C > 0$, independent of $n$, such that
    \begin{align*}
        \p(\|\bX\| \leq \epsilon) \leq C \epsilon^n.
    \end{align*}
    Consequently, if we use an ABC‐acceptance region of radius $\eps$ around the observed $\xobs$, the expected number of simulations needed to obtain a single accepted draw is at least $\frac{1}{C\,\epsilon^n}$,
    which grows exponentially in the dimension $n$.
\end{lemma}

\begin{proof}
    We begin by bounding the probability using the volume of an $n$-dimensional ball:
    \begin{align*}
        \p(\|\bX\| \leq \epsilon) = \int_{\|\bx\|\leq \epsilon} f(\bx) \, d \bx &\leq \int_{\|\bx\|\leq \epsilon} K\, d \bx = K \mathrm{Vol}(B^n_{\epsilon})\\
        &= \frac{K \pi^{n/2} \epsilon^n}{\Gamma\left(\frac{n}{2} + 1\right)},
    \end{align*}
    where $B^n_{r}$ denotes a ball of radius $r$ in $\R^n$.
    To analyze the behavior of this bound as $n$ increases, we apply Stirling's approximation:
    \begin{align*}
        \frac{\pi^{n/2}}{\Gamma\left(\frac{n}{2} + 1\right)} \sim \frac{\pi^{n/2}}{\sqrt{\pi n} \left(\frac{n}{2e}\right)^{n/2}} = \frac{1}{\sqrt{\pi n}} \left(\frac{2 \pi e}{n}\right)^{n/2} \overset{n \to \infty}{\to} 0
    \end{align*}
    Therefore, we can define a constant $C$, independent of $n$, as
    \begin{align*}
        C := \sup_{n\ge1} \frac{K\pi^{n/2}}{\Gamma\left(\frac{n}{2}+1\right)} < \infty.
    \end{align*}
    The inequality then follows as claimed for all $n \geq 1$.
\end{proof}

\section{Adaptive Rejection Sampling}

\subsection{Warmup: Adaptive Bayesian Inference for the Univariate Gaussian Model}\label{example:gaussian-gaussian}

As a simple illustration of our sequential procedure, consider the conjugate Gaussian–Gaussian model:
\begin{gather*}
    \theta \sim \mathcal{N}(0,20), \qquad X \mid \theta\sim \mathcal{N}(\theta,1).
\end{gather*}
The prior on $\theta$ is diffuse, while $X\mid\theta$ concentrates around the true parameter.  

We run \abi for $\mathsf{T}$ iterations with a decreasing tolerance sequence $\epsilon_1> \epsilon_2> \cdots> \epsilon_\mathsf{T}$.  In this univariate setting, we use the Euclidean distance, and write $B_r(x)= \{y:|y-x|\leq r\}$.
Initialize the proposal as $\pi^{(0)}(\theta):=\pi(\theta)$.  
At iteration $t$, we draw $\mathsf{N}$ synthetic parameter–data pairs from the current proposal $\pi^{(t)}(\theta) := \pi(\theta \mid E_{t-1})$:
\begin{enumerate}
    \item Draw $\theta \sim \pi^{(t)}(\theta)$;
    \item Draw $X \mid \theta \sim P_\theta^{(n)}(\cdot \mid E_{t-1})$ by rejection sampling using Algorithm \ref{step:rej-samp-while}.
\end{enumerate}
This yields the selected set, whose distribution is 
$\pi(\theta,X\mid E_{t-1})$, which we denote by $\pi^{(t)}(\theta,X)$:
\begin{align*}
  \mathsf{S}_0^{(t)}=\{(\theta^{(i)},X^{(i)})\}_{i=1}^\mathsf{N}\,.
\end{align*}
We then retain those $\theta^{(i)}$ whose associated $X^{(i)}$ falls within $B_{\epsilon_t}(\xobs)$:
\begin{align*}
\mathsf{S}_{\theta,*}^{(t)}
&=\bigl\{\theta^{(i)} : (\theta^{(i)},X^{(i)})\in\mathsf{S}_0^{(t)}
\ \text{and}\ X^{(i)}\in B_{\epsilon_t}(\xobs)\bigr\}.
\end{align*}
Let $\bar{\theta}_*^{(t)}$ and $(\hat{\sigma}^{(t)}_*)^{2}$ denote the sample mean and variance of the set of retained samples $\mathsf{S}_{\theta,*}^{(t)}$. Since a Gaussian distribution is fully determined by its mean and variance, proposal distribution for the next iteration is updated to be 
\begin{align*}
\theta \sim \N\bigl(\bar\theta_*^{(t)}, \,(\hat{\sigma}^{(t)}_*)^{2}\bigr).
\end{align*}

We generated a single observed value $\xobs = 6.24$ from the model; by conjugacy the exact posterior is $\pi(\theta\mid \xobs)=\N(5.94,\;0.95)$.
At each iteration $t$, we drew $\mathsf N=10^4$ synthetic pairs $(\theta^{(i)},X^{(i)})$ from the current Gaussian proposal and retained those $\theta^{(i)}$ for which $X^{(i)}\in B_{\epsilon_t}(\xobs)$. We used the decreasing tolerance schedule
\begin{align*}
\epsilon=(2,\;0.7,\;0.3,\;0.01,\;0.005,\;0.003,\;0.001,\;0.001,\;0.001).
\end{align*}
The retained $\theta$-values were used to fit the next Gaussian proposal by moment matching.  We repeated this procedure for $\mathsf{T}=9$ iterations. 
Let the empirical retention rate be defined as
\begin{align*}
    \frac{1}{\mathsf N}\,\bigl|\{\,i:|X^{(i)}-\xobs|\le0.001\}\bigr|.
\end{align*}
Figure \ref{fig:simple-gaussian-example-fig} displays the smoothed posterior densities across iterations, together with the corresponding retention rates.  From our simulations, we make the following observations:  
\begin{enumerate}[label=(\arabic*).]
  \item The sequential rejection sampling refinement is \emph{self-consistent}: repeated application does not cause the sampler to drift away from the true posterior.  
  \item In this univariate example, the sampler converges rapidly, requiring just just three iterations in this univariate example.  
\end{enumerate}

\begin{figure}[H]
\centering
\includegraphics[width=0.96\columnwidth]{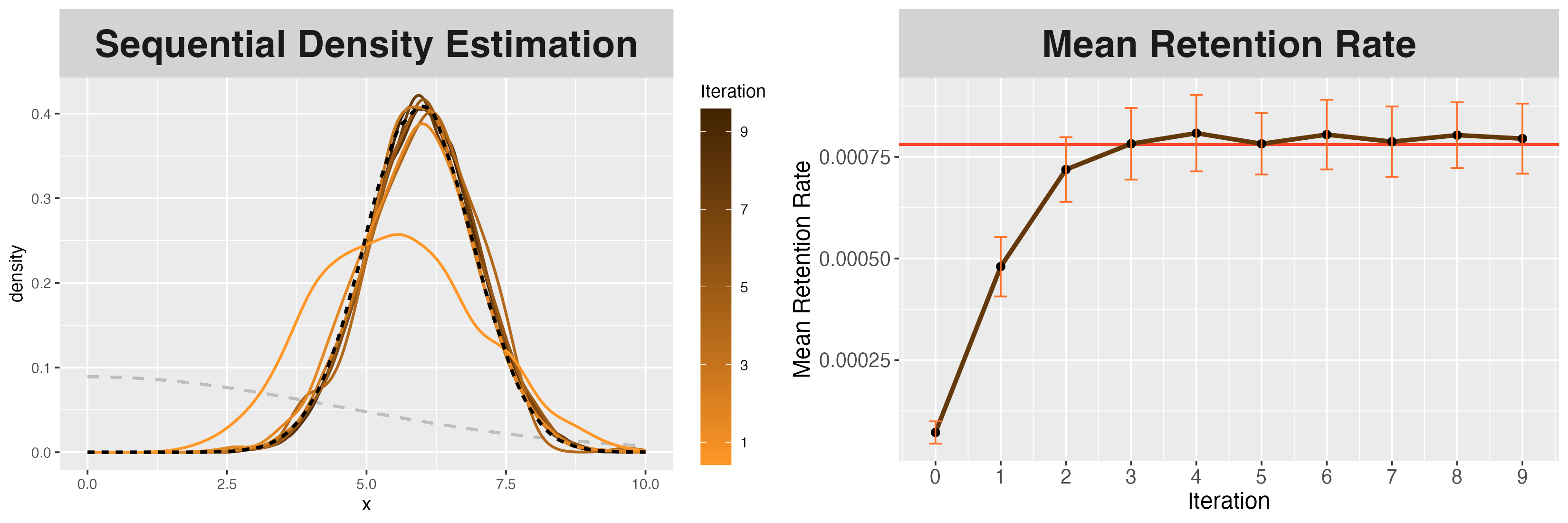}
\caption{Adaptive inference on the univariate Gaussian model.  
\textbf{Left:} Smoothed proposal and posterior densities over iterations; prior shown in gray, true posterior in black. 
\textbf{Right:} Empirical retention rate per iteration. The red line indicates the rate obtained using the true posterior, while the error bands represent $\pm 1$ standard deviation of the retention rate, calculated from 100 independent repetitions.}
\label{fig:simple-gaussian-example-fig}
\end{figure}

\section{Posterior Sufficiency}
First, for completeness and measure-theoretic rigor, we formally define the probability space. Let the parameter and data $(\Theta, X)$ be jointly defined on a nice probability space $(\Omega \times \X^n, \B \otimes \mathcal{A}, P)$, where $(\Omega, \mathcal{B}) \subseteq (\R^d, \mathcal{B}(\R^d))$ is the measurable Polish parameter space and $(\X^n, \mathcal{A}) \subseteq (\R^{d_X}, \mathcal{B}(\R^{d_X}))$ is the measurable observation space, with $d, d_X \in \mathbb{N}_+$. The prior probability measure $\Pi$ on $(\Omega, \mathcal{B})$ is assumed to be absolutely continuous with respect to the Lebesgue measure, with corresponding density function $\pi(d\theta)$ for $\theta \in \Omega$. For each $\theta \in \Omega$, let $P^{(n)}_\theta$ denote the conditional probability measure on $(\X^n, \mathcal{A})$ given $\Theta = \theta$.

We provide the definition of Bayes sufficiency presented in \cite{godambe1968bayesian} below.

\begin{defn}[Bayes Sufficiency]
Let $(\Omega, \mathcal{B})$ be a measurable parameter space, $(\mathcal{X}, \mathcal{A})$ be a sample space, and $C_\pi$ be a countable class of prior distributions on $(\Omega, \mathcal{B})$. A statistic $T: \mathcal{X} \to \mathcal{T}$ is said to be \textit{Bayes sufficient} with respect to $C_\pi$ if, for any $\pi \in C_\pi$, the posterior density $\pi(\theta \mid X)$ depends on $X$ only through $T(X)$ for almost every $x$, i.e.,
there exists a function $g: \Omega \times \mathcal{T} \to \mathbb{R}^+$ such that
\begin{align}
    \pi(\theta \mid x) = g(\theta, T(x)) \quad \text{for } \pi\text{-almost every } \theta \in \Omega \text{ and } P_\pi\text{-almost every } x \in \mathcal{X},
\end{align}
where $P_\pi$ is the marginal distribution of $X$ under the prior $\pi$ for all $\pi \in C_\pi$.
Equivalently, for any $\pi \in C_\pi$ and $B \in \mathcal{B}$,
\begin{align}
\E[\ind_B(\Theta) \mid X] = \E[\ind_B(\Theta) \mid T] \quad P_\pi\text{-almost surely}.
\end{align}
In other words, conditional on $T(X)$, the observation $X$ provides no additional information about $\Theta$ under the specified prior $\pi$.
\end{defn}

\begin{thm}[Minimal Bayes Sufficiency of Posterior Distribution]
    Let $(\Omega, \mathcal{B})$ be a measurable parameter space and $(\mathcal{X}, \mathcal{A})$ be a sample space. For $X \in \mathcal{X}$, define $T: \mathcal{X} \to \mathcal{M}(\Omega)$ as the function that maps each observation $X$ to its corresponding posterior distribution under the prior class $C_\pi$, i.e., $T(X) := \pi(d\theta \mid X)$, where $\mathcal{M}(\Omega)$ denotes the space of probability measures on $(\Omega, \mathcal{B})$.
    Note that this is a random probability measure on $\Omega$, and that $\mathcal{M}(\Omega)$ is a Polish space (under the weak topology).
    Then:
    \begin{enumerate}
        \item $T$ is Bayes sufficient for $\Theta$ with respect to $\pi$. 
        \item If $T'$ is any other Bayes sufficient statistic with respect to $\pi$, then $T'(x) = T'(x')$ implies that $T(x) = T(x')$, $P_\pi$-almost every $x, x'$.
    \end{enumerate}
    Consequently, the posterior map $X \mapsto \pi(d \theta \mid X)$ is minimally Bayes sufficient for the specified prior $\pi$.
\end{thm}

\begin{proof}
We will show that $T$ is Bayes sufficient and minimally so in two parts.

\paragraph{$T$ is Bayes sufficient}
For any measurable $B \subseteq \Omega$, by the definition of conditional expectation, for $P_\pi$-almost every $x \in \mathcal{X}$ we have
\begin{align*}
    \E\left[\ind_B(\Theta) \mid X=x\right] &= \int_B \pi\left(d\theta \mid x\right).
\end{align*}
Thus, for $P_\pi$-almost every $x$ and for every measurable set $B$ we obtain
\begin{align*}
    \E\left[\ind_B(\Theta) \mid X=x\right]  &= \int_B \pi\left(d\theta \mid x\right)\\
    &= \int_B \pi\left(d\theta \mid T(x)\right) \\
    &= \E\left[\ind_B(\Theta) \mid T(X)=T(x)\right].
\end{align*}
Since this equality holds for every measurable set $B$, we conclude that the conditional distributions $\pi\left(d\theta \mid X=x\right)$ and $\pi\left(d\theta \mid T(x)\right)$ are equal for $P_\pi$-almost every $x$. By the definition of Bayes sufficiency, this shows that $T$ is Bayes sufficient.

\paragraph{Minimality of $T$}
Let $T'$ be any other Bayes‑sufficient statistic.
If $T'(x)=T'(x')$, Bayes sufficiency of $T'$ gives
\begin{align*}
  \pi(d\theta \mid x)=\pi(d\theta \mid T'(x))
  =\pi(d\theta \mid T'(x'))
  =\pi(d\theta \mid x')
  \qquad (P_\pi\text{-a.e.\ }x,x')
\end{align*}
for all $\pi \in C_\pi$.
But $T(x)=\pi(d\theta \mid x)$, hence $T(x) = T(x')$ almost surely whenever $T'(x)=T'(x')$. 
Thus $T$ is minimal.
\end{proof}

\section{Proofs and Technical Results}

\subsection{Proof of Theorem \ref{thm:ars-sample-complexity}}
\begin{assumption}[Local Positivity]\label{assump:tv-dist-overlapping-supp}
    There exists global constants $c > 0$ and $\gamma>0$ such that, for all $\theta \in \supp(\pi^{(t)})$, $P_\theta^{(n)}(A_t)$ is uniformly bounded away from zero, i.e., $\inf \limits_{\theta \in \supp(\pi^{(t)})} P_\theta^{(n)}(A_t) \geq  c \epsilon_t^{\gamma}$, where
    \begin{align*}
        \bX &= (\bX_1, \dots, \bX_n) \in \R^{nd_X},\\
        A_t &= \{x : \widehat{\msw}_{p, \delta}(\bX, \bxobs) \leq \epsilon_t\},\\
        E_t &= \{\omega : \widehat{\msw}_{p, \delta}(\bX\left(\omega), \bxobs\right) \leq \epsilon_t \} = \{\omega: \bX(\omega) \in A_t\}.
    \end{align*}
\end{assumption}

\begin{thm}[Error Rate of \abc{ARS}]\label{thm:ars-tv-error}
Suppose Assumption \ref{assump:tv-dist-overlapping-supp} holds.
Then, the total variation distance between the marginal distributions over $\theta$, corresponding to the exact rejection sampling procedure in Algorithm \ref{step:rej-samp-while} and the \abc{ARS} procedure in Algorithm \ref{alg:approx-rej-samp}, converges to zero at a rate of $O(\exp(-\mathsf{R}\epsilon_t^{\gamma}))$ as $\mathsf{R} \to \infty$.
\end{thm}

\begin{proof}
    We begin by establishing key notations for our analysis. Let $\pi^{(t)} = \pi(\theta \mid E_t)$ denote the marginal target distribution of the exact rejection sampling procedure.
    Define $\pi^{(t)}_{\mathsf{ARS}}$ as the marginal distribution of $\theta$ under the approximate sampling procedure with $\mathsf{R}$ attempts. We denote the corresponding measures by $\pi^{(t)}$ and $\pi_\mathsf{ARS}^{(t)}$, respectively.

    For $r = 1, 2, \dots, \mathsf{R}$, let $E_{t,r} = \{\widehat{\msw}_{p, \delta}(\bX_r, \xobs) \leq \epsilon_t\}$ be the event where the $r$th replicate $\bX_r$ satisfies the threshold condition. Define $E_t = E_{t,1}$. A particular $\theta$ is retained if at least one of the corresponding data replicate passes the threshold,
    \begin{align*}
        \bigcup_{r=1}^\mathsf{R} E_{t,r} = \left \{\omega: \left \{ \bX_1(\omega)\in A_t \right \} \; \bigcup\; \cdots \;\bigcup \; \left \{ \bX_R(\omega)\in A_t \right \} \right \}.
    \end{align*}

    Taking expectation over $E_{t,1}, \dots, E_{t,\mathsf{R}}$, we obtain:
    \begin{align*}
        \widetilde{\pi}^{(t)}_{\mathsf{ARS}}(\theta) &\propto \int \left(  \widetilde{\pi}^{(t)}(\theta) \bigcup_{r=1}^\mathsf{R} E_{t,r} \right) \, \prod_{r=1}^\mathsf{R} P_\theta^{(n)}(d\bX)\\
        &= \int \widetilde{\pi}^{(t)}(\theta) \left( 1- \bigcap_{r=1}^\mathsf{R} \ind_{E_{t,r}^c}(\omega)\right) \prod_{r=1}^\mathsf{R} P_\theta^{(n)}(d\bX)\\
        &= \widetilde{\pi}^{(t)}(\theta)\left( 1- P_{\theta}^{(n)}(A_t^c)^\mathsf{R} \right),
    \end{align*}
    where the last equality follows from the i.i.d. nature of $\bX_1, \dots, \bX_R$ conditional on $\theta$. Thus, under the \abc{ARS} procedure (Algorithm \ref{alg:approx-rej-samp}), the marginal distribution of $\Theta$ is proportional to:
    \begin{align}\label{eq:ars-unnorm-density}
        \widetilde{\pi}^{(t)}_{\mathsf{ARS}}(\theta) \propto \widetilde{\pi}^{(t)}(\theta) \left(1 - P_\theta^{(n)}(A_t^c)^\mathsf{R}\right).
    \end{align}

    Let $Z_t = \E_{\widetilde{\pi}^{(t)}}\left[1 - P_\theta^{(n)}(A_t^c)^\mathsf{R}\right] = \int_\Omega \widetilde{\pi}^{(t)} (\theta) \left(1 - P_\theta^{(n)} (A_t^c)^\mathsf{R} \right) \, d\theta$ be the normalizing constant for Eq.\ \eqref{eq:ars-unnorm-density}. We can then bound the total variation distance:
    
    \begin{align}
        \mathcal{D}_{\mathsf{TV}} \left(\pi^{(t)}, \widetilde{\pi}_{\mathsf{ARS}}^{(t)} \right) &= \frac{1}{2} \int_\Omega \left |\pi^{(t)}(\theta) - \widetilde{\pi}^{(t)}_{\mathsf{ARS}} \right|\, d\theta \nonumber \\
        &= \frac{1}{2} \int_\Omega \left| \pi^{(t)}(\theta) - Z_{t}^{-1} \widetilde{\pi}^{(t)}(\theta) \left(1 - P_\theta^{(n)}(A_t^c)^\mathsf{R}\right) \right|\, d\theta \nonumber \\
        &= \frac{1}{2 Z_t} \int_\Omega \left | Z_t \widetilde{\pi}^{(t)}(\theta) - \widetilde{\pi}^{(t)} (\theta)\left(1 - P_\theta^{(n)} (A_t^c)^\mathsf{R} \right) \right| \, d\theta \label{eq:tv-subst-dist}\\
        &= \frac{1}{2 Z_t} \int_\Omega \left | \widetilde{\pi}^{(t)}(\theta) \left(Z_t - \left(1 - P_\theta^{(n)} (A_t^c)^\mathsf{R} \right) \right) \right|\, d\theta \nonumber \\
        &\leq \frac{1}{2 Z_t} \left(\int_\Omega   \widetilde{\pi}^{(t)}(\theta) \, d\theta  \right)^{1/2} \left(\int_\Omega \left |Z_t - 1 + P_\theta^{(n)} (A_t^c)^\mathsf{R} \right|^2 \, \widetilde{\pi}^{(t)}(\theta) d\theta \right)^{1/2} \label{eq:tv-holder}\\
        &= \frac{1}{2 Z_t} \left (\int_\Omega \left| \E_{\widetilde{\pi}^{(t)}} \left[P_\theta^{(n)} (A_t^c)^\mathsf{R} \right] - P_\theta^{(n)} (A_t^c)^\mathsf{R} \right|^2\, d\theta  \right)^{1/2} \nonumber \\
        &= \frac{1}{2 Z_t}  \Var_{\widetilde{\pi}^{(t)}}\left (P_\theta^{(n)}(A_t^c)^\mathsf{R}\right )^{1/2} \nonumber \\
        &\leq \frac{1}{2 Z_t} \E_{\widetilde{\pi}^{(t)}}\left[P_\theta^{(n)}(A_t^c)^{2\mathsf{R}}\right]^{1/2} \label{eq:tv-last-line}
    \end{align}
    where Eq.\ \eqref{eq:tv-subst-dist} uses the fact that $\pi^{(t)}(\theta) \overset{d}{=}_{\theta} \widetilde{\pi}^{(t)}(\theta)$, Eq.\ \eqref{eq:tv-holder} follows from Hölder's inequality, and we also used the fact that $\int_\Omega \widetilde{\pi}(\theta)\, d\theta =1$.

    By Assumption \ref{assump:tv-dist-overlapping-supp}, there exists a constant $c$ such that $\inf \limits_{\theta \in \supp(\pi^{(t)})} P_\theta^{(n)}(A_t) \geq c\epsilon_t^{\gamma}$. Consequently,
    \begin{align*}
        \E_{\widetilde{\pi}^{(t)}}\left[P_\theta^{(n)}(A_t^c)^{2\mathsf{R}}\right]^{1/2} &\leq \left( \left(1 - c\epsilon_t^{\gamma}\right)_{+}^{2\mathsf{R}}\right)^{1/2} \leq \exp(-\mathsf{R}c\epsilon_t^{\gamma}),
    \end{align*}
    where we use the inequality $(1-u)_{+}^\mathsf{R} \leq \exp(-\mathsf{R} u)$ for $u \in\mathbb{\mathsf{R}}$ and integer $\mathsf{R}$. Similarly, for $Z_t$, using $\sup\limits_{\theta \in \supp(\pi^{(t)})} P_\theta^{(n)}(A_t^c)\leq 1 - c\epsilon_t^{\gamma}$:
    \begin{align*}
        Z_t &= \E_{\widetilde{\pi}^{(t)}} \left[1 - P_\theta^{(n)} (A_t^c)^\mathsf{R} \right] \geq 1 - (1-c\epsilon_t^{\gamma})_{+}^\mathsf{R} \geq 1 - \exp(-\mathsf{R} c\epsilon_t^{\gamma}).
    \end{align*}

    Substituting these bounds into Eq.\ \eqref{eq:tv-last-line} yields:
    \begin{align*}
        \frac{1}{2 Z_t} \E_{\widetilde{\pi}^{(t)}}\left[P_\theta^{(n)}(A_t^c)^{2\mathsf{R}}\right]^{1/2} \leq  \frac{1}{2}\cdot \frac{\exp(-\mathsf{R} c\epsilon_t^{\gamma})}{1-\exp(-\mathsf{R} c\epsilon_t^{\gamma})} \lesssim \exp\left(-\mathsf{R} c\epsilon_t^{\gamma}\right),
    \end{align*}
    which establishes the upper bound and the desired convergence rate.
\end{proof}

\begin{thm}[Sample Complexity for \abc{ARS}]\label{thm:ars-sample-complexity-appendix}
Suppose Assumption \ref{assump:tv-dist-overlapping-supp} holds.
For any $\bar{\delta} \in (0,1)$ and $\epsilon_t > 0$, if the number of samples $\mathsf{R}$ in the \abc{ARS} algorithm satisfies $\mathsf{R} = O\left(\frac{\log(1/\bar{\delta})}{\epsilon_t^{\gamma}}\right)$, 
then the total variation distance between the exact and approximate posterior distributions is bounded by $\bar{\delta}$, i.e.,
$\mathcal{D}_{\mathsf{TV}} \left(\pi^{(t)}, \widetilde{\pi}_{\mathsf{ARS}}^{(t)} \right) \leq \bar{\delta}$.
\end{thm}

\begin{proof}
    This follows directly from Theorem \ref{thm:ars-tv-error}.
\end{proof}

\begin{cor}[1-Wasserstein Error Incurred by \abc{ARS}]
Under the conditions of Theorem \ref{thm:ars-tv-error}, if the parameter space $\Omega$ has a finite diameter $d_\Omega = \text{diam}(\Omega)$, then the 1-Wasserstein distance between the exact and approximate posterior distributions is bounded by 
\begin{align}
\W_1\left(\pi^{(t)}, \widetilde{\pi}_{\mathsf{ARS}}^{(t)}\right) \lesssim d_\Omega \exp\left(-\mathsf{R} c\epsilon_t^{\gamma}\right).
\end{align}
Therefore, to achieve $\W_1\left(\pi^{(t)}, \widetilde{\pi}_{\mathsf{ARS}}^{(t)}\right) \leq \bar{\delta}$, we need $\mathsf{R} = O\left(\frac{\log(d_\Omega/\bar{\delta})}{\epsilon_t^{\gamma}}\right)$.
\end{cor}

\begin{proof}
By Theorem 6.15 of \cite{villani2009optimal}, we have
\begin{align*}
\W_1\left(\pi^{(t)}, \widetilde{\pi}_{\mathsf{ARS}}^{(t)}\right) &\leq d_\Omega \cdot \mathcal{D}_{\mathsf{TV}}\left(\pi^{(t)}, \widetilde{\pi}_{\mathsf{ARS}}^{(t)}\right).
\end{align*}
The rest follows from Theorem \ref{thm:ars-tv-error} and Theorem \ref{thm:ars-sample-complexity}.
\end{proof}

\subsection{Proofs of the Main Theorems from Section \ref{sec:theory:msw}}

\subsubsection{Proof of Proposition \ref{prop:msw-metric}}

Nonnegativity, symmetry, and triangle inequality for the \msw distance follow directly from the corresponding properties of $\W_p$ and $\SW_p$ distances, coupled with the additivity property of metrics.
    Note that if $\mu = \nu$, then 
    \begin{align*}
        \SW_p(\mu, \nu) &= 0,\\
        \W_p((e_j)_\#\mu, (e_j)_\#\nu) &= 0, \quad j = 1, \dots, d.
    \end{align*}
    For the converse direction, suppose $\msw_p(\mu, \nu) = 0$. As $0 < \lambda < 1$, both terms in the \msw distance must vanish:
    \begin{align*}
        \sum_{j=1}^d \W_p((e_j)_\# \mu, (e_j)_\# \nu) &= 0,\\
        \int_{\mathbb{S}^{d-1}} \W_p^p (\varphi_\# \mu, \varphi_\# \nu) \,d\sigma(\varphi) &= 0,
    \end{align*}
    which implies that the marginal distributions of $\mu$ and $\nu$ are identical for all coordinates, and $\W_p(\varphi_\# \mu, \varphi_\# \nu) = 0$ for $\sigma$-almost all $\varphi \in \mathbb{S}^{d-1}$.
    By the Cram\'er-Wold theorem, this is sufficient to conclude that $\mu = \nu$.

\subsubsection{Proof of Theorem \ref{thm:msw-ipm}}

\begin{lemma}\label{lem:sw1-ipm}
The 1-Sliced Wasserstein distance is an Integral Probability Metric on $\mathcal{P}_1(\R^d)$.
\end{lemma}

\begin{proof}
Recall the definition of the $SW_1$ distance,
\begin{align*}
    SW_1 (\mu, \nu) = \int_{\smb^{d-1}} W_1(\varphi_\# \mu, \varphi_\# \nu) d\sigma(\varphi),
\end{align*}
where $\smb^{d-1}$ is the unit sphere in $\R^d$ and $\sigma$ is the uniform measure on $\smb^{d-1}$. Define the critic function class,
\begin{align*}
    \F := \left\{f(x) = \int_{\smb^{d-1}} g_\varphi(\varphi^\top x ) \;d\sigma(\varphi)  \,\Big|\, g_\varphi \in \Lip_1(\R), \, \sup_{\varphi \in \smb^{d-1}} |g_\varphi(0)| < \infty\right\}
\end{align*}
where for each $\varphi \in \mathbb{S}^{d-1}$, $g_\varphi: \mathbb{R} \to \mathbb{R}$ is a 1-Lipschitz function, such that the mapping $(\varphi, t) \mapsto g_\varphi(t)$ is jointly measurable with respect to the product of the Borel $\sigma$-algebras on $\mathbb{S}^{d-1}$ and $\mathbb{R}$. Note that $\F$ is nonempty, as it includes constant functions. For instance, if we take $f_c: x \mapsto \int_{\smb^{d-1}} k \, d\sigma(\varphi) = k < \infty$, then $f_c \in \F$.

First, we show that $\sup_{f \in \F}\left|\int fd(\mu - \nu)\right|\leq SW_1(\mu, \nu)$ holds.
Fix $f \in \F$. By definition, $f(x) = \int_{\smb^{d-1}} g_\varphi (\varphi^\top x) \, d\sigma(\varphi)$. For any $\varphi\in\smb^{d-1}$, 
\begin{equation*}
    |g_\varphi(\varphi^\top x)| \leq |g_\varphi(0)| + |\varphi^\top x|\leq \sup_{\varphi\in\smb^{d-1}}|g_\varphi(0)|+\|x\|.
\end{equation*}
As $\sup_{\varphi\in\smb^{d-1}} |g_\varphi(0)|< \infty$ and $\int\|x\| \, d\mu < \infty$, we have $\int \int_{\smb^{d-1}} |g_\varphi (\varphi^\top x)| \, d\sigma(\varphi) d\mu<\infty$; similarly, $\int \int_{\smb^{d-1}} |g_\varphi (\varphi^\top x)| \, d\sigma(\varphi) d\nu<\infty$. By Fubini--Tonelli,
\begin{align*}
    \int f d(\mu - \nu) = \int_{\smb^{d-1}} \int g_\varphi(\varphi^\top x) d(\mu - \nu) (x) \, d\sigma(\varphi).
\end{align*}
Thus
\begin{align*}
    \left |\int f d(\mu - \nu) \right| &\leq \int_{\smb^{d-1}} \left|\int g_\varphi(\varphi^\top x) d(\mu-\nu) (x) \right|\; d\sigma(\varphi) \\
    &\leq \int_{\smb^{d-1}} W_1(\varphi_\# \mu, \varphi_\# \nu) \, d\sigma(\varphi) = SW_1(\mu, \nu),
\end{align*}
where we appeal to the dual representation of $W_1$ for each $\varphi$. Taking the supremum over $\F$ establishes the claimed inequality.

We now show that the reverse inequality $SW_1(\mu, \nu) \leq \sup_{f \in \F} \left|\int fd(\mu- \nu)\right|$ holds.
Fix $\eps > 0$.
For each $\varphi \in \smb^{d-1}$, choose a 1-Lipschitz function $g_\varphi^\eps$ with $g_\varphi^\eps(0)=0$ and 
\begin{align*}
     \int g_\varphi^\eps(\varphi^\top x) d(\mu - \nu)(x) \geq W_1(\varphi_\# \mu, \varphi_\# \nu) - \eps.
\end{align*}
By the Kuratowski–Ryll–Nardzewski measurable selection theorem, we can select $\varphi \mapsto g_\varphi^\eps$ such that the mapping $(\varphi,t)\mapsto g_\varphi^\eps(t)$ is jointly measurable with respect to the product of the Borel $\sigma$-algebras on $\mathbb{S}^{d-1}$ and $\mathbb{R}$.  

Define 
\begin{align*}
    f_\eps(x) = \int_{\smb^{d-1}} g_\varphi^\eps (\varphi^\top x)\, d\sigma(\varphi) \in \F.
\end{align*}
Note that $f_\eps \in \F$ because each slice is 1-Lipschitz and $\sup_{\varphi\in \mathbb{S}^{d-1}} |g_\varphi^\eps(0)|<\infty$, by construction.
Then, exactly as before,
\begin{align*}
    \left|\int f_{\eps}\,d(\mu-\nu)\right|
    &\geq \int_{\smb^{d-1}}
      \int g_{\varphi}^{\eps}(\varphi^{\top}x)\,d(\mu-\nu)(x)
      \,d\sigma(\varphi)\\
    &\geq \int_{\smb^{d-1}}
      \bigl(W_{1}(\varphi_\#\mu,\varphi_\#\nu)-\eps\bigr)
      \,d\sigma(\varphi)
      = SW_{1}(\mu,\nu)-\eps.
\end{align*}
Taking the supremum over $\F$ and letting $\eps \downarrow 0$ yields the reverse inequality.

Combining both steps gives
\begin{align*}
    SW_1(\mu,\nu)
    = \sup_{f\in\F}\left|\int f\,d(\mu-\nu)\right|,
\end{align*}
so $SW_1$ is an IPM with critic class $\F$.
\end{proof}

\begin{lemma}[Closure of IPMs under finite linear combinations]
\label{lem:ipms-combination}
Let $\D_1, \D_2, \dots, \D_K$ be IPMs on the same space of probability measures, and let $\lambda_1,\lambda_2,\dots,\lambda_K \geq 0$ be nonnegative constants. Define
\begin{align}
    \D(\mu, \nu) := \sum_{k=1}^K \lambda_k \, \D_k(\mu,\nu).
\end{align}
Then $\D(\cdot,\cdot)$ is itself an IPM.
\end{lemma}
\begin{proof}
Let $\D_1, \dots, \D_K$ be as given. For each $k = 1, \dots, K$, let $\F_k$ be the critic function class corresponding to $\D_k$, i.e.,
\begin{align*}
    \D_k(\mu, \nu) = \sup_{f \in \F_k} \left| \int f d(\mu - \nu)\right|.
\end{align*}
Let $\F_k^\dagger := \F_k \cup (-\F_k)$; note that $\D_k$ is unchanged when $\F_k$ is replaced by $\F_k^\dagger$. For any $\lambda_1,\lambda_2,\cdots,\lambda_K\geq 0$, define the new critic function class 
\begin{align*}
    \F := \bigg\{f = \sum_{k=1}^K \lambda_k f_k \ \Big|\ f_k \in \F_k^\dagger, ~k=1, \dots, K\bigg\}.
\end{align*}
We prove the desired equality in two steps. 

We first show that $\sup_{f\in \F} \left|f d(\mu-\nu) \right| \leq \D(\mu, \nu)$ holds.
For any choice of $f_k \in \F_k^\dagger$ for every $k\in [K]$, we have
\begin{align*}
    \left | \int \sum_{k=1}^K \lambda_k f_k d(\mu - \nu)\right| \leq \sum_{k=1}^K \lambda_k \left|\int f_k d(\mu - \nu) \right| \leq \sum_{k=1}^K \lambda_k \D_k(\mu, \nu).
\end{align*}
Taking the supremum over all $f = \sum_{k=1}^K \lambda_k f_k \in \F$ yields
\begin{align}\label{proof:ipm-closure:ub}
    \sup_{f \in \F} \left | \int f d(\mu - \nu) \right| \leq \D(\mu, \nu).
\end{align}

Now we show that $\sup_{f\in \F} \left|f d(\mu-\nu) \right| \geq \D(\mu, \nu)$ holds.
Fix $\eps > 0$. For each $k$ choose $f_k^{\eps} \in \F_k^\dagger$ such that
\begin{align*}
    \int f_k^{\eps} d(\mu - \nu) \geq \D_k(\mu, \nu) - \frac{\eps}{K (\lambda_k \vee 1)}.
\end{align*}
Define $f_\eps := \sum_{k=1}^K \lambda_k f_k^\eps \in \F$. 
By linearity,
\begin{align*}
   \int f_\eps d(\mu - \nu)  = \sum_{k=1}^K \lambda_k \int f_k^\eps d(\mu - \nu) \geq \sum_{k=1}^K \lambda_k \left(\D_k(\mu, \nu) - \frac{\eps}{K (\lambda_k \vee 1)}\right) \geq \D(\mu, \nu) - \eps.
\end{align*}
Since $f_\eps \in \F$, it follows that 
\begin{align*}
    \sup_{f \in \F} \left| f d(\mu - \nu) \right|\geq \D(\mu, \nu) - \eps. 
\end{align*}
Since the bound holds for arbitrary $\eps > 0$, letting $\eps \downarrow 0$ yields 
\begin{align}\label{proof:ipm-closure:lb}
    \sup_{f \in \F} \left | \int f d(\mu - \nu) \right| \geq \D(\mu, \nu).
\end{align}

Combining \eqref{proof:ipm-closure:ub} and \eqref{proof:ipm-closure:lb} gives
\begin{align*}
    \D(\mu, \nu) = \sup_{f \in \F} \left| \int f d(\mu - \nu) \right|,
\end{align*}
which confirms that $\D$ is an IPM, as desired.
\end{proof}

Below we provide the proof of Theorem \ref{thm:msw-ipm}. 

\begin{proof}
For each $e_j$, $j = 1, \dots, d$, $\W_1(({e_j})_\# \mu, ({e_j})_\# \nu)$ is an IPM. By Lemmas \ref{lem:sw1-ipm} and \ref{lem:ipms-combination}, we readily obtain the claimed result. 
\end{proof}

\subsubsection{Proof of Theorem \ref{thm:msw-topo-equiv}}
To prove Theorem \ref{thm:msw-topo-equiv}, we need to first establish some useful lemmas and propositions. 

\begin{lemma}\label{lem:monotonicity-msw}
For any $p \geq 1$ and any $\mu,\nu\in\mathcal{P}_p(\R^d)$, we have $\msw_1(\mu, \nu) \leq \msw_p(\mu, \nu)$.
\end{lemma}
\begin{proof}
By H\"older's inequality, for each $j\in [d]$, we have 
\begin{align*}
    \W_1(({e_j})_\# \mu, ({e_j})_\# \nu) &\leq \W_p(({e_j})_\# \mu, ({e_j})_\# \nu).
\end{align*}
For the Sliced Wasserstein distance, fixing any projection $\varphi \in \smb^{d-1}$, the one-dimensional Wasserstein distance satisfies $W_1(\varphi_\#\mu, \varphi_\# \nu) \leq W_p (\varphi_\# \mu, \varphi_\# \nu)$.
Hence by Jensen's inequality,
\begin{align*}
    \int W_1 (\varphi_\# \mu, \varphi_\# \nu) \, d\sigma(\varphi)\leq \left( \int W_1 (\varphi_\# \mu, \varphi_\# \nu)^p \, d\sigma(\varphi) \right)^{1/p} \leq \left(\int W_p (\varphi_\# \mu, \varphi_\# \nu)^p \, d\sigma(\varphi) \right)^{1/p}.
\end{align*}
Thus $\SW_1(\mu, \nu) \leq \SW_p(\mu, \nu)$. Since $\msw$ is a convex combination of these components, the desired inequality holds for the \msw distance.
\end{proof}

\begin{prop}\label{prop:msw-leq-wass}
    For any $p\geq 1$ and $\mu, \nu \in \cp_p(\R^d)$, we have
    \begin{equation*}
        \msw_p(\mu, \nu) \leq C_{d, p, \lambda} \W_p(\mu, \nu),
    \end{equation*}
    where
    \begin{align*}
    C_{d, p, \lambda} = \lambda+(1-\lambda)c_{d, p}^{1/p},\qquad
    c_{d, p} = \frac{1}{d} \int_{\smb^{d-1}} \|\varphi\|_p^p \, d\sigma(\varphi) \leq 1.
\end{align*}
    Note that $C_{d, p, \lambda} \leq 1$ for every $d \in \mathbb{N}_+$, $p \geq 1$, and $\lambda \in (0, 1)$.
\end{prop}

\begin{proof}
     Let $\gamma^* \in \Gamma(\mu, \nu)$ be an optimal transport plan for the minimization problem (\ref{Wpminimization}) with $\Omega=\R^d$. Then for any $j\in [d]$, $({e_j} \otimes {e_j})_\# \gamma^*$ is a transport plan between $({e_j})_\# \mu$ and $({e_j})_\# \nu$. By the definition of the Wasserstein distance as an infimum over all couplings, we have:
    \begin{align*}
        \W_p(({e_j})_\# \mu, ({e_j})_\# \nu) \leq \left(\int |\langle{e_j}, x - y\rangle|^p \, d\gamma^*(x, y)\right)^{1/p}\leq \left(\int \|x - y\|^p \, d\gamma^*(x, y)\right)^{1/p}= \W_p(\mu, \nu).
    \end{align*}
    Consequently, 
    \begin{align*}
        \frac{1}{d} \sum_{j=1}^d W_p(({e_j})_\# \mu, ({e_j})_\# \nu)\leq \W_p(\mu, \nu).
    \end{align*}
Furthermore, from Proposition 5.1.3 of \cite{bonnotte2013unidimensional}, we have:
    \begin{align*}
        \SW_p^p(\mu, \nu) \leq c_{d,p} \W_p^p(\mu, \nu),
    \end{align*}
    where
    \begin{align*}
        c_{d,p} = \frac{1}{d} \int_{\mathbb{S}^{d-1}} \|\varphi\|_p^p d\sigma(\varphi) \leq 1.
    \end{align*}
    Combining these results yields the desired inequality.
\end{proof}

\begin{prop}\label{prop:wass-leq-msw}
For all $\mu, \nu$ supported in $B_R(0)$, $R > 0$, there exists a constant $C_{d, \lambda} > 0$ such that 
\begin{align}
    \W_1(\mu, \nu) \leq C_{d, \lambda} R^{d/(d+1)} \msw_1^{1/(d+1)}(\mu, \nu)
\end{align}
\end{prop}

\begin{proof}
    For brevity of notation, let $\mathcal{M}_1(\mu, \nu) := \frac{1}{d} \sum_{j=1}^d \W_1(({e_j})_\# \mu, ({e_j})_\# \nu)$.
    By the definition of the \msw distance,
    \begin{align*}
        \SW_1(\mu, \nu) &= \frac{\msw_1(\mu, \nu) - \lambda \mathcal{M}_1(\mu, \nu)}{1-\lambda} \\
        &\leq (1-\lambda)^{-1}\msw_1(\mu, \nu),
    \end{align*}
    as $\lambda \mathcal{M}_1(\mu, \nu) \geq 0$ for $\lambda \in (0,1)$. 
    
    By Lemma 5.1.4 of \cite{bonnotte2013unidimensional}, there exists a constant $\widetilde{C}_d > 0$ such that
    \begin{align*}
        \W_1(\mu, \nu) &\leq \widetilde{C}_d R^{d/(d+1)} \SW_1(\mu, \nu)^{1/(d+1)}\\
        &\leq \widetilde{C}_d R^{d/(d+1)} ((1-\lambda)^{-1}\msw_1(\mu, \nu))^{1/(d+1)}\\
        &= \widetilde{C}_d R^{d/(d+1)} (1-\lambda)^{-1/(d+1)} \msw_1(\mu, \nu)^{1/(d+1)}
    \end{align*}
    for all probability measures $\mu, \nu$ supported in $B_R(0)$.
    
    Setting $C_{d, \lambda} := \widetilde{C}_d (1-\lambda)^{-1/(d+1)}$, we obtain the desired result:
    \begin{align*}
         \W_1(\mu, \nu) \leq C_{d, \lambda} R^{d/(d+1)} \msw_1^{1/(d+1)}(\mu, \nu).
    \end{align*}
\end{proof}

Now we can introduce the proof of the Theorem \ref{thm:msw-topo-equiv}.
\begin{proof}
    The first inequality follows from Proposition \ref{prop:msw-leq-wass}. The second inequality follows from Proposition \ref{prop:wass-leq-msw} and Lemma \ref{lem:monotonicity-msw}, on noting that $W_p^p(\mu,\nu)\leq (2R)^{p-1}W_1(\mu,\nu)$.  
\end{proof}

\subsubsection{Proof of Theorem \ref{thm:msw-metrize-conv}}
\begin{proof}
Let $(\mu_\ell)_{\ell \in \mathbb{N}_+}, 
\mu$ be probability measures in $\cp_p(\R^d)$ such that $\lim_{\ell \to \infty}\msw_p(\mu_\ell, \mu) = 0$. Since $\lambda \in (0, 1)$, this implies $\lim_{\ell \to \infty}SW_p(\mu_\ell, \mu) = 0$. By Theorem 1 of \cite{nadjahi2019asymptotic}, we conclude that $\mu_\ell \Rightarrow \mu$, where we use $\Rightarrow$ to denote weak convergence in $\cp_p(\R^d)$.

Conversely, suppose $\mu_\ell \Rightarrow \mu$.  
Theorem 6.9 of \cite{villani2009optimal} gives $W_p(\mu_\ell, \mu) \to 0$. Because $SW_p(\mu_\ell, \mu) \leq W_p(\mu_\ell, \mu)$, we also obtain $SW_p(\mu_\ell, \mu) \to 0$. 
For the marginal component, 
by the Cram\'er-Wold theorem,
a sequence of probability measures on $\R^d$ converges weakly if and only if their one-dimensional projections converge weakly for all directions; therefore, for any $j\in\{1,\cdots,d\}$, 
\begin{align*}
    ({e_j})_\# \mu_\ell \overset{d}{\to} ({e_j})_\# \mu.
\end{align*}
Now fix $j \in \{1, \dots, d\}$. 
Recall that weak convergence in $\cp_p(\R^d)$ requires that \cite[Definition 6.8]{villani2009optimal}
\begin{align*}
    \lim_{R\to\infty}\limsup_{\ell\rightarrow\infty}\int_{\R^d} \|z\|^p\mathbbm{1}_{\|z\|\geq R}d\mu_{\ell}(z) = 0.
\end{align*}
Since $|z_j |^p \leq \|z\|^p$, we have 
\begin{align*}
    \lim_{R\to\infty}\limsup_{\ell\rightarrow\infty} \int_{\mathbb{R}} |t|^p\ind_{|t|\geq R} \, d({e_j})_\# \mu_\ell(t) =\lim_{R\to\infty}\limsup_{\ell\rightarrow\infty}\int_{\R^d} |z_j|^p \mathbbm{1}_{|z_j|\geq R} d\mu_{\ell}(z) = 0.
\end{align*}
Hence $({e_j})_\# \mu_\ell$ converges to $({e_j})_\# \mu$ in $\cp_p(\R)$ for every $j\in\{1,\cdots,d\}$. Appealing once again to Theorem 6.9 of \cite{villani2009optimal} yields
each term $W_p(({e_j})_\#\mu_\ell, ({e_j})_\#\mu) \to 0$.
Hence,
\begin{align*}
    \msw_p(\mu_\ell, \mu) = \frac{\lambda}{d} \sum_{j=1}^d W_p(({e_j})_\# \mu_\ell, \, ({e_j})_\# \mu) + (1-\lambda) SW_p(\mu_\ell, \mu) \to 0.
\end{align*}
Therefore $\msw_p$ indeed metrizes weak convergence on $\cp_p(\R^d)$ without any compactness assumptions, as desired.
\end{proof}

\subsubsection{Proof of Theorem \ref{theory:msw-conv-rate}}

We start with two lemmas.

\begin{lemma}\label{theory:lemma:trim-constant-finite}
Fix $\delta \in (0,1/2)$. For every $\varphi \in \smb^{d-1}$ and $\tau \in [\delta, 1-\delta]$, we have 
\begin{align*}
    \bigl|F^{-1}_{\mu_\varphi}(\tau)\bigr| \leq \left(\frac{M_{\mu,p}}{\delta}\right)^{1/p} < \infty.
\end{align*}
\end{lemma}
\begin{rmk}
A similar result holds for $\nu$.
\end{rmk}

\begin{proof}
Let $Z \sim \mu$. For any $\varphi \in \smb^{d-1}$, denote by $Z_\varphi := \varphi^\top Z$ the corresponding one-dimensional random variable with distribution $\mu_\varphi$. By the Cauchy-Schwarz inequality, using the fact that $\E[\|Z\|^p] < M_{\mu,p}$, for every $\varphi \in \smb^{d-1}$, we have
    \begin{align*}
        \E[|Z_\varphi|^p] = \E[|\langle\varphi, Z\rangle|^p] \leq \E[\|Z\|^p]<M_{\mu,p} < \infty.
    \end{align*}
    Hence by Markov's inequality, for any $t>0$,
    \begin{align*}
        \p(|Z_\varphi| \geq t) \leq \frac{\E[|Z_\varphi|^p]}{t^p} < 
        \frac{M_{\mu,p}}{t^p}.
    \end{align*}
Setting $t_\delta = \left(\frac{M_{\mu, p}}{\delta}\right)^{1/p} < \infty$, we get
\begin{align*}
        \p(|Z_\varphi| \geq t_{\delta})< \frac{M_{\mu, p}}{t_{\delta}^p} = \delta.
    \end{align*}
Together with the 
    fact that $F_{\mu_\varphi}(F^{-1}_{\mu_\varphi}(\delta))\geq \delta$, this implies that 
    \begin{align*}
        F^{-1}_{\mu_\varphi}(\delta) \geq -t_{\delta}, \quad F^{-1}_{\mu_\varphi}(1-\delta) \leq t_{\delta}.
    \end{align*}
    Therefore,
    \begin{align*}
        \big| F^{-1}_{\mu_\varphi}(\tau) \big|\leq t_\delta = \left(\frac{M_{\mu, p}}{\delta}\right)^{1/p} < \infty \quad \text{for all } \tau \in [\delta, 1-\delta], \varphi \in \smb^{d-1}.
    \end{align*}
\end{proof}

\begin{lemma}[Empirical Quantile Deviation]\label{theory:lemma:empirical-quantile-deviation}
Let $\delta\in(0,1/2)$, $\varphi \in \smb^{d-1}$, and suppose that $Z_{\varphi,1},\dots,Z_{\varphi,m}$ are i.i.d.\ samples from $\varphi_\#\mu$, so that the empirical measure is 
\begin{align*}
\varphi_\#\widehat{\mu}_m = \frac{1}{m}\sum_{i=1}^{m}\delta_{Z_{\varphi,i}},
\end{align*}
with empirical CDF $\widehat{F}_{m,\mu_\varphi}$ and true CDF $F_{\mu_\varphi}$. Then, for any $t \geq 0$, we have
\begin{align*}
\p\left\{ \left|\widehat{F}^{-1}_{m,\mu_\varphi}(\delta)-F^{-1}_{\mu_\varphi}(\delta)\right| > t \right\}
&\leq 2\exp\left(-2m\, \psi_{\delta,t}(\mu_\varphi)^2\right),
\end{align*}
where
$\psi_{\delta,t}(\cdot)$ is as defined in Eq.\ \eqref{phideltat}.
\end{lemma}
\begin{rmk}
A similar result holds for $\nu$.
\end{rmk}

\begin{proof}
For simplicity of notation, we omit the dependence on the fixed direction $\varphi$ and $\mu$ throughout this proof. Let $F_m(z) :=\hat{F}_{m,\mu_{\varphi}}(z) = \frac{1}{m}\sum_{i=1}^m \ind\left\{{Z_{\varphi, i}}\leq z\right\}$ and $F(z) := F_{\mu_\varphi}(z)$ for every $z\in\mathbb{R}$.

By Hoeffding's inequality, for any $z\in\mathbb{R}$ and $t\geq 0$, we have
\begin{align*}
    \p(F_m(z) - F(z) \geq t) \leq \exp(-2mt^2),\quad \p(F_m(z) - F(z) \leq -t) \leq \exp(-2mt^2).
\end{align*}
Observe that the event $F_{m}^{-1}(\delta) > F^{-1}(\delta) + t$ is equivalent to $F_m(F^{-1}(\delta) + t) < \delta$. Thus, we can write:
\begin{align*}
    \p \big (F_m^{-1}(\delta) - F^{-1}(\delta) > t) &= \p\big(F_m(F^{-1}(\delta) + t) < \delta\big) \\
    &= \p\big(F_m(F^{-1}(\delta) + t) - F(F^{-1}(\delta) + t) < \delta - F(F^{-1}(\delta) + t)\big) \\
    &\leq \exp\big\{-2m (F(F^{-1}(\delta) + t) -\delta)^2\big\},
\end{align*}
where we use the fact that $F(F^{-1}(\delta)+t)\geq F(F^{-1}(\delta))\geq\delta$ in the last line). An analogous argument applies for upper bounding $\p(F_m^{-1}(\delta) - F^{-1}(\delta)<-t)$. Applying the union bound to combine the two upper bounds yields the desired inequality.
\end{proof}

\begin{proof}[Proof of Theorem \ref{theory:msw-conv-rate}]
Fix any $\delta\in(0,1\slash 2)$ and $\bar{\delta}\in (0,1)$. We let
    \begin{align*}
        \underline{q}_{\mu, \delta}(\varphi) :&= F^{-1}_{\mu_{\varphi}}(\delta) \wedge \widehat{F}_{m, \mu_\varphi}^{-1}(\delta),\\
        \overline{q}_{\mu, 1-\delta}(\varphi) :&= F^{-1}_{\mu_{\varphi}}(1-\delta) \vee \widehat{F}^{-1}_{m, \mu_\varphi}(1-\delta),\\
        \underline{q}_{\nu, \delta}(\varphi) :&= F^{-1}_{\nu_{\varphi}}(\delta) \wedge \widehat{F}^{-1}_{m, \nu_\varphi}(\delta),\\
        \overline{q}_{\nu, 1-\delta}(\varphi) :&= F^{-1}_{\nu_{\varphi}}(1-\delta) \vee \widehat{F}^{-1}_{m, \nu_\varphi}(1-\delta).
    \end{align*}
    
    \paragraph{Analysis of Empirical Convergence Rate}
    We decompose the total error into two components and apply concentration techniques to control each term.
    
    \paragraph{Error from the Marginal Component}
    For each coordinate $j = 1, \dots, d$, the projected measures $\mu_j$ and $\nu_j$ are one-dimensional. 
    Let $q_{\mu, \delta}(e_j) = |\underline{q}_{\mu, \delta}(e_j)| \vee |\overline{q}_{\mu, 1-\delta}(e_j)|$. 
    Note that 
     \begin{align*}
      \W^p_{p, \delta}\big(({e_j})_\#\widehat{\mu}_m&, ({e_j})_\# \mu \big) = \frac{1}{1-2\delta} \int_\delta^{1-\delta} \left | \widehat{F}^{-1}_{m, \mu_j}(\tau) - F^{-1}_{\mu_j}(\tau) \right|^p \, d\tau\\
        &\leq \frac{1}{1-2\delta}\int_{0}^{2 q_{\mu, \delta}(e_j)} p t^{p-1} \int_\delta^{1-\delta}\ind\left\{\left| \widehat{F}^{-1}_{m, \mu_j}(\tau) - F^{-1}_{\mu_j}(\tau) \right| \geq t  \right\} d\tau \, dt\\
        &\leq    \frac{1}{1-2\delta} \sup_{t \in [0, 2q_{\mu, \delta}(e_j)]}\{pt^{p-1}\}\int_{0}^{2 q_{\mu, \delta}(e_j)} \int_\delta^{1-\delta}\ind\left\{\left| \widehat{F}^{-1}_{m, \mu_j}(\tau) - F^{-1}_{\mu_j}(\tau) \right| \geq t \right\}d\tau \, dt  \\
        &\leq \frac{p}{1-2\delta} \big(2 q_{\mu, \delta}(e_j)\big)^{p-1}\int_{\delta}^{1-\delta} \left|\widehat{F}_{m,\mu_j}^{-1}(\tau) - F_{\mu_j}^{-1}(\tau) \right|\, d\tau\\
         &\leq \frac{p}{1-2\delta} 
         \big(2 q_{\mu, \delta}(e_j)\big)^{p-1}
         \int_{\underline{q}_{\mu, \delta}(e_j)}^{\overline{q}_{\mu, 1-\delta}(e_j)} |\widehat{F}_{m, \mu_j}(t) - F_{\mu_j}(t)|\, dt\\
         &\leq \frac{p}{1-2\delta} 
         \big(2 q_{\mu, \delta}(e_j)\big)^{p-1}
         ( \overline{q}_{\mu, 1-\delta}(e_j) - \underline{q}_{\mu, \delta}(e_j)) \|\widehat{F}_{m, \mu_j} - F_{\mu_j}\|_\infty,
    \end{align*}
    where the second, fourth, and fifth lines are derived using the Fubini--Tonelli theorem.

Recall Eq. \eqref{phideltat}-\eqref{def_epsnew}. By Lemma \ref{theory:lemma:empirical-quantile-deviation}, for any $j\in \{1,\cdots,d\}$, we have 
\begin{align}\label{proof:eq:quantile-dev}
        \p\left\{ \left|\widehat{F}^{-1}_{m,\mu_j}(\delta)-F^{-1}_{\mu_j}(\delta)\right| > \eps_{m,d,\delta, \overline{\delta}}(\mu_j) \right\}
        &\leq 2\exp\left(-2m\, \psi_{\delta,\eps_{m,d, \delta, \overline{\delta}}(\mu_j)}(\mu_j)^2\right) \leq \frac{\overline{\delta}}{16d}.
    \end{align}
    Similarly, \begin{align}\label{proof:eq:quantile-dev2}
        \p\left\{ \left|\widehat{F}^{-1}_{m,\mu_j}(1-\delta)-F^{-1}_{\mu_j}(1-\delta)\right| > \eps_{m,d,1-\delta, \overline{\delta}}(\mu_j) \right\}
        &\leq 2\exp\left(-2m\, \psi_{1-\delta,\eps_{m,d, 1-\delta, \overline{\delta}}(\mu_j)}(\mu_j)^2\right) \leq \frac{\overline{\delta}}{16d}.
    \end{align}
    By Lemma \ref{theory:lemma:trim-constant-finite}, Eq.\ \eqref{proof:eq:quantile-dev}-\eqref{proof:eq:quantile-dev2}, and the union bound, the following inequalities hold with probability at least $1-\frac{\overline{\delta}}{8}$ for all $j\in\{1, \dots, d\}$:
    \begin{align*}
        q_{\mu, \delta}(e_j) &\leq
\left(\frac{M_{\mu, p}}{\delta}\right)^{1/p} + \eps_{m,d,\delta, \overline{\delta}}(\mu_j) \vee \eps_{m,d,1-\delta,\overline{\delta}}(\mu_j)=\frac{1}{2}R_{\mu_j, \delta}, \\
        \overline{q}_{\mu, 1-\delta}(e_j) - \underline{q}_{\mu, \delta}(e_j) &\leq 2 \left(\left(\frac{M_{\mu, p}}{\delta}\right)^{1/p} + \eps_{m,d,\delta, \overline{\delta}}(\mu_j) \vee \eps_{m,d,1-\delta, \overline{\delta}}(\mu_j)\right)= R_{\mu_j, \delta}.
    \end{align*}

    Therefore, with probability at least $1-\frac{\overline{\delta}}{8}$, we have for all $j \in \{1,\cdots,d\}$, 
    \begin{align}
  \W^p_{p, \delta}\big(({e_j})_\#\widehat{\mu}_m, ({e_j})_\# \mu \big) &\leq \frac{p}{1-2\delta} 
        R_{\mu_j, \delta}^p \big \|\widehat{F}_{m, \mu_j} - F_{\mu_j} \big\|_\infty.
    \end{align}
    \noindent By the Dvoretzky-Kiefer-Wolfowitz inequality and a union bound over $j = 1, \dots, d$, we have
    \begin{align*}
        \p\left(\max_{j \in [d]}\{\|\widehat{F}_{m, \mu_j} - F_{\mu_j}\|_\infty\}>t \right) \leq 2d e^{-2mt^2}.
    \end{align*}
    Combining these results, we obtain that
    \begin{align*}
        \p \left(\bigcup_{j \in [d]} \left\{ \W^p_{p, \delta}\big( ({e_j})_\# \widehat{\mu}_m,\, ({e_j})_\# \mu \big) > \frac{p}{1-2\delta}R_{\mu_j, \delta}^p \cdot  t\right\}\right) \leq \frac{\overline{\delta}}{8} + 2de^{-2mt^2}.
    \end{align*}
    Setting $\frac{\overline{\delta}}{8} = 2d e^{-2mt^2}$, we obtain $t = \sqrt{\frac{\log(16d/\overline{\delta})}{2m}}$.
    Therefore, with probability at least $1-\frac{\overline{\delta}}{4}$, for all $j\in\{1,\cdots,d\}$,
    \begin{align*}
           \W^p_{p, \delta}\big( ({e_j})_\# \widehat{\mu}_m,\, ({e_j})_\# \mu \big)   \leq \frac{p}{1-2\delta}\max_{j \in [d]}\{R_{\mu_j, \delta}^p\}\cdot \sqrt{\frac{\log(16d/\overline{\delta})}{2m}}.
    \end{align*}
    Applying an analogous argument to $\nu$ yields
    \begin{align*}
        \p\left( \max_{j\in[d]} \big\{\W^p_{p, \delta}\big( ({e_j})_\# \widehat{\nu}_{m'},\, ({e_j})_\# \nu\big)\big\}\leq \frac{p}{1-2\delta}\max_{j \in [d]}\{R_{\nu_j, \delta}^p\}\cdot \sqrt{\frac{\log(16d/\overline{\delta})}{2m'}}\right) \geq 1-\frac{\overline{\delta}}{4}.
    \end{align*}

    Recall that $R_{\max}= \max_{j \in [d]} \{R_{\mu_j, \delta}\} \vee \max_{j \in [d]} \{R_{\nu_j, \delta}\}$.
    By the union bound and the triangle inequality, with probability at least $1-\frac{\overline{\delta}}{2}$, we have for all $j \in [d]$,
    \begin{align*}
    &\bigl|\W_{p, \delta}\big(({e_j})_\# \widehat{\mu}_m, ({e_j})_\# \widehat{\nu}_{m'}\big) - \W_{p, \delta}\big(({e_j})_\# \mu, ({e_j})_\# \nu\big)\bigl| \\
    \leq &  \left|\W_{p, \delta}\big(({e_j})_\# \widehat{\mu}_m, ({e_j})_\# \mu\big)\right| + \left|\W_{p, \delta}\big(({e_j})_\# \widehat{\nu}_{m'}, ({e_j})_\# \nu\big)\right|.
\end{align*}
Define $\widehat{\mu}_{m, j} := ({e_j})_{\#} \widehat{\mu}_m$ and similarly $\widehat{\nu}_{m', j} := ({e_j})_{\#} \widehat{\nu}_{m'}$, we therefore have
\begin{equation*}
\begin{gathered}
     \p\Biggl(\max_{j\in[d]} \biggl|\W_{p, \delta} \big(\widehat{\mu}_{m, j}, \widehat{\nu}_{m', j}\big) 
     - \W_{p, \delta}\big(\mu_j, \nu_j\big)\bigg| \leq 2R_{\max}\Biggl( \frac{p}{1-2\delta} 
     \cdot \sqrt{\log(16d/\overline{\delta}})\Biggr)^{1/p}
     \big(m^{-\frac{1}{2p}}+m'^{-\frac{1}{2p}}\big)
    \Biggr)\\
    \centering
    \geq 1 - \frac{\overline{\delta}}{2}.
\end{gathered}
\end{equation*}
Thus
\begin{equation}
\begin{gathered}
    \p\Biggl(\left | \frac{1}{d}\sum_{j=1}^dW_{p, \delta}\big(\widehat{\mu}_{m, j}, \, \widehat{\nu}_{m', j}\big) - \frac{1}{d}\sum_{j=1}^d W_{p, \delta} (\mu_j, \nu_j) \right | \leq 2R_{\max}\Biggl( \frac{p}{1-2\delta} \cdot \sqrt{\log(16d/\overline{\delta}})\Biggr)^{1/p}(m^{-\frac{1}{2p}}+m'^{-\frac{1}{2p}}) \Biggr)\\
    \centering
    \geq 1 - \frac{\overline{\delta}}{2}.
\end{gathered}
\end{equation}

\paragraph{Error from the Sliced Wasserstein Component}

By Proposition 1(ii) of \cite{manole2022minimax}, there exists a constant $C_p > 0$ depending only on $p$ such that
\begin{align*}
   & \E\left[\big |SW_{p, \delta}(\widehat{\mu}_m , 
 \widehat{\nu}_{m'}) - SW_{p, \delta}(\mu, \nu)\big |\right] \nonumber\\
 \leq & \frac{C_p\big(\mathbb{E}_{Z\sim \mu}[\|Z\|^2]^{1\slash 2}\vee \mathbb{E}_{Z\sim \nu}[\|Z\|^2]^{1\slash 2}\big)}{\sqrt{\delta}(1-2\delta)^{1\slash p}}\big(m^{-1\slash (2p)} + m'^{-1\slash (2p)}\big).
\end{align*}
Combining the expression above with Markov's inequality, for any $t\geq 0$, we have
\begin{align*}
   & \p\left(\big |SW_{p, \delta}(\widehat{\mu}_m, \widehat{\nu}_{m'}) - SW_{p, \delta}(\mu, \nu)\big | \geq t \right) \nonumber\\
   &\leq  \frac{C_p\big(\mathbb{E}_{Z\sim \mu}[\|Z\|^2]^{1\slash 2}\vee \mathbb{E}_{Z\sim \nu}[\|Z\|^2]^{1\slash 2}\big)}{\sqrt{\delta}(1-2\delta)^{1\slash p}t}\big(m^{-1\slash (2p)} + m'^{-1\slash (2p)}\big).
\end{align*}
Hence taking 
\begin{align*}
    t_{\mathrm{SW}} := 
    \frac{2C_p\big(\mathbb{E}_{Z\sim \mu}[\|Z\|^2]^{1\slash 2}\vee \mathbb{E}_{Z\sim \nu}[\|Z\|^2]^{1\slash 2}\big)}{\sqrt{\delta}(1-2\delta)^{1\slash p}\overline{\delta}} 
    \big(m^{-1\slash (2p)} + m'^{-1\slash (2p)}\big),
\end{align*}
we get 
\begin{align}
 \p\left(\big |SW_{p, \delta}(\widehat{\mu}_m,\widehat{\nu}_{m'}) - SW_{p, \delta}(\mu, \nu)\big | \geq t_{\mathrm{SW}} \right) \leq \frac{\overline{\delta}}{2}.
\end{align}

Thus, combining both parts, we obtain the desired rate
\begin{equation*}
    \p\left(\big| \msw_{p, \delta}(\widehat{\mu}_m, \widehat{\nu}_{m'}) - \msw_{p, \delta}(\mu, \nu) \big| \geq t_{\msw}\right) \leq \overline{\delta},
\end{equation*}
where
\begin{align*}
     t_{\msw}:=& \left\{2\lambda  R_{\max} \left(\frac{p}{1-2\delta} \sqrt{\log(  16 d / \overline{\delta})} \right)^{1/p} + \frac{2(1-\lambda)C_p\big(\mathbb{E}_{Z\sim \mu}[\|Z\|^2]^{1\slash 2}\vee \mathbb{E}_{Z\sim \nu}[\|Z\|^2]^{1\slash 2}\big)}{\sqrt{\delta}(1-2\delta)^{1\slash p}\overline{\delta}}\right\}\nonumber\\
     &\quad\cdot\big(m^{-1/(2p)} + m'^{-1/(2p)}\big).
\end{align*}
    
\end{proof}

\subsection{Proofs of the Main Theorems from Section \ref{theory:abi-property}}

\subsubsection{Proof of Theorem \ref{thm:abi-conv}}

As $(\Omega,\B)$ and $(\X^n,\mathcal{A})$ are standard Borel, the disintegration theorem (Thoerem 3.4, \citealt{kallenberg1997foundations}) yields a measurable kernel
\begin{align*}
\pi_{\Theta\mid X}: \B \times \X^n \to [0,1],\quad
(B,x)\mapsto\pi_{\Theta\mid X}(B,x),
\end{align*}
satisfying
\begin{align*}
    P_{(\Theta, X)}(B \times A) = \int_A \pi_{\Theta \mid X}(B, x) P_X(dx) 
\end{align*}
for all $B \in \B$, $A \in \mathcal{A}$.
The random variable
\begin{align*}
\omega \mapsto
\pi_{\Theta\mid X}\bigl(B, X(\omega)\bigr) = \int_B \pi_{\Theta \mid X}\bigl(d\theta, X(\omega)\bigr)
\end{align*}
is a version of the conditional probability $\p (\Theta\in B \mid \sigma(X))$.

Recall that $f_{X}(\cdot)$ the marginal density of $X$, i.e., for any $x\in\mathcal{X}^n$,
\begin{equation*}
    f_{X}(x)=\int_{\Theta}\pi(\theta)f_{X \mid \Theta}(x \mid \theta).
\end{equation*}
For any $\epsilon>0$, denote $A_{\epsilon}:=\{x\in\mathcal{X}:\msw{p}(\pi_{\Theta \mid x},\pi_{\Theta \mid x^{*}})\leq\epsilon\}$. Note that
\begin{equation}
    \pi^{(\epsilon)}_{\abi}(\theta \mid x^{*})=\pi_{\Theta \mid X}(\theta \mid X\in A_{\epsilon})=\frac{\int_{A_{\epsilon}}\pi_{\Theta \mid X}(\theta \mid x)f_X(x)dx}{\int_{A_{\epsilon}}f_X(x)dx}.
\end{equation}
Hence
\begin{align}
&\msw{p}\left(\pi^{(\epsilon)}_{\abi}(\theta \mid x^{*}),\pi_{\Theta \mid X}(\theta \mid \xobs)\right)=\msw{p}\left(\frac{\int_{A_{\epsilon}}\pi_{\Theta \mid X}(\theta \mid x)f_X(x)dx}{\int_{A_{\epsilon}}f_X(x)dx},\pi_{\Theta \mid X}(\theta \mid \xobs)\right)\nonumber\\
&\hspace{5cm}\leq \frac{\int_{A_{\epsilon}}f_X(x) \msw{p}(\pi_{\Theta \mid X}(\theta \mid x),\pi_{\Theta \mid X}(\theta \mid \xobs))
dx}{\int_{A_{\epsilon}}f_X(x)dx}.
\end{align}
By the definition of $A_{\epsilon}$, for any $x\in A_{\epsilon}$, $\msw{p}(\pi_{\Theta \mid x},\pi_{\Theta \mid x^{*}})\leq\epsilon$. Hence
\begin{equation}
    \msw{p}\left(\pi^{(\epsilon)}_{\abi}(\theta \mid x^{*}),\pi_{\Theta \mid X}(\theta \mid \xobs)\right)\leq \epsilon.
\end{equation}
Therefore, as $\epsilon\rightarrow 0^{+}$, $\pi^{(\epsilon)}_{\abi}(\theta \mid x^{*})$ converges weakly in $\mathcal{P}_p(\Omega)$ to $\pi_{\Theta|X}(\theta|\xobs)$. 

\subsubsection{Proof of Theorem \ref{thm:cont-msw-dist}}

As $f_X(\xobs) > 0$ and $f_X$ is continuous at $\xobs$, the point $\xobs$ lies in the support of $P_X$, so all conditional probabilities below are well-defined.

Consider a deterministic decreasing sequence $\epsilon_t \downarrow 0$ as $t \to \infty$. Let
\begin{align*}
    D_t = \{X \in B_{\epsilon_t}(\xobs)\}.
\end{align*}
Define the random variable 
\begin{align*}
    Z_t := \ind{}_{D_t},
\end{align*}
which records whether $X$ falls in successively smaller neighborhoods of the observed data $\xobs$. 

Define the filtration 
\begin{align*}
    \F_t = \sigma(Z_1, \dots, Z_t) \subseteq \F,
\end{align*}
By convention, we let $\F_0 := \{\emptyset, \Xi\}$. 
By construction, the filtration is increasing, i.e., $\F_{t'} \subseteq \F_t$ for all $t' \leq t$.
Define $\F_\infty$ as the minimal $\sigma$-algebra generated by $(\F_t)_{t\in \mathbb{N}}$, i.e., 
\begin{align*}
    \F_\infty = \sigma(\cup_t \F_t).
\end{align*}
A key observation is that
\begin{align*}
    \bigcap_{t \geq 1} D_t = \{X = \xobs\}.
\end{align*}
Hence, the indicator random variable of the exact match satisfies
\begin{align}
    \ind\{X = \xobs\} = \ind\{X \in \cap_t D_t\} \in \F_\infty.
\end{align}
In other words, $\F_\infty$ reveals whether the random vector $X$ coincides with the realized observation $\xobs$. 

For any set $B \in \B$, denote $U_B(\omega) := \ind\{\Theta(\omega) \in B\}$, $\omega \in \Xi$. Then, by L\'evy's 0-1 law, we have
\begin{align*}
    \E\big[\ind_{D_t}U_B \mid \F_t\big] &\underset{L^1}{\overset{a.s.}{\to}} \E\big[\ind\{X=\xobs\}U_B \mid \F_\infty\big] ~\text{as } t \to \infty\\
    &= \ind\{X = \xobs\}\E[U_B \mid \F_\infty]\\
    &= \ind\{X=\xobs\} \p[U_B \mid X=\xobs].
\end{align*}
This convergence holds for all $B \in \B$, hence
\begin{align*}
    \pi_{\Theta \mid X}(\cdot \mid D_t) \overset{d}{\to} \pi_{\Theta \mid X}(\cdot \mid X=\xobs).
\end{align*}

Let $V(\theta) = \|\theta\|_p^p$. By assumption, we have
    $M = \sup_{x \in \X^n} \int_\Omega V(\theta) f_{\Theta, X}(\theta, x)\, d\theta < \infty$. Note that for any $t\geq 1$, $V\ind_{D_t}\leq V\ind_{D_1}$, and
    \begin{equation*}
        \mathbb{E}[V\ind_{D_1}]=\int_{\Omega}\int_{\mathcal{X}^n\cap B_{\epsilon_1}(\xobs)}\|\theta\|_p^pf_{\Theta,X}(\theta,x)d\theta dx\leq M \cdot \text{Leb}(\mathcal{X}^n\cap B_{\epsilon_1}(\xobs))<\infty,
    \end{equation*}
    where $\text{Leb}$ denotes the Lebesgue measure on $\R^{nd_X}$. Hence as $t\rightarrow\infty$, by L\'evy's 0-1 law, 
    \begin{align*}
        \E[V \ind_{D_t} \mid \F_t] \underset{L^1}{\overset{a.s.}{\to}} \ind_{X=\xobs} \E[V \mid X= \xobs].
    \end{align*}
    Therefore, we have 
    \begin{align*}
        \int_\Omega V(\theta) \pi_{\Theta \mid X}(d\theta \mid D_t) \underset{L^1}{\overset{a.s.}{\to}} \int_\Omega V(\theta) \pi_{\Theta \mid X}(d\theta \mid X=\xobs).
    \end{align*}
    Together with weak convergence, this gives
    \begin{align*}
        \pi_{\Theta \mid X}(\cdot \mid D_t) \Rightarrow \pi_{\Theta \mid X}(\cdot \mid X = \xobs).
    \end{align*}

    Since the $\msw_p$ distance is topologically equivalent to $W_p$ on $\cp_p(\Omega)$, we have
    \begin{align*}
    \lim_{t\to\infty}
    \msw_{p}\!\Bigl(
      \pi_{\Theta\mid X}(d\theta\mid X\in B_{\epsilon_t}(\xobs)),
      \pi_{\Theta\mid X}(d\theta\mid X=\xobs)
    \Bigr)=0,
\end{align*}
as desired.

\newpage
\section{Additional Details on Simulation Settings}
\subsection{Multimodal Gaussian Model}
Below, we present pairwise bivariate density plots of the posterior distribution under the multimodal Gaussian model. The plots demonstrate that \abi accurately recovers the joint posterior structure.
\begin{figure}[H]
    \centering
    \begin{subfigure}[b]{0.48\linewidth}
        \centering
        \includegraphics[width=\linewidth]{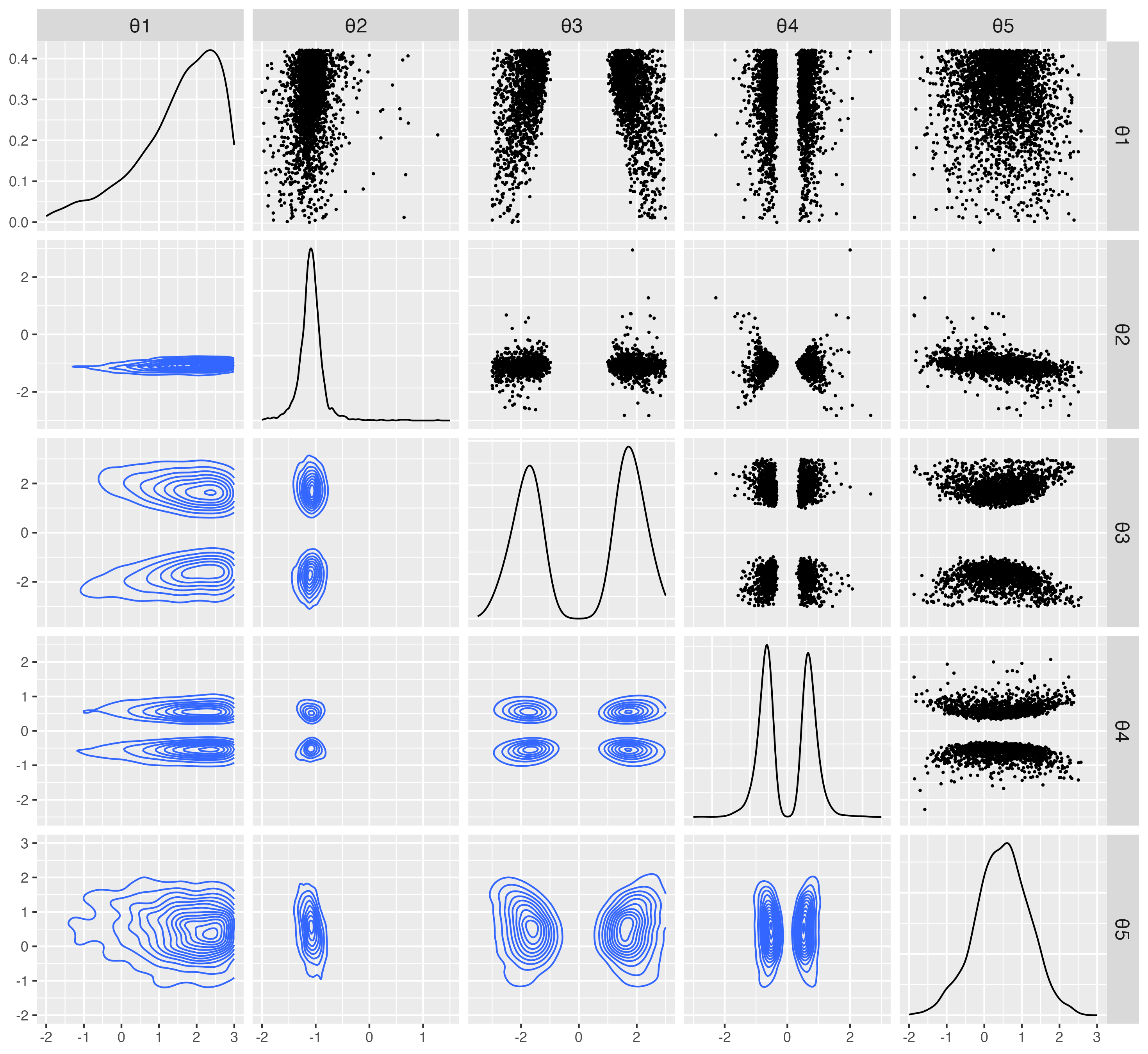}
        \caption{}
        \label{fig:toy-truth}
    \end{subfigure}
    \hfill
    \begin{subfigure}[b]{0.48\linewidth}
        \centering
        \includegraphics[width=\linewidth]{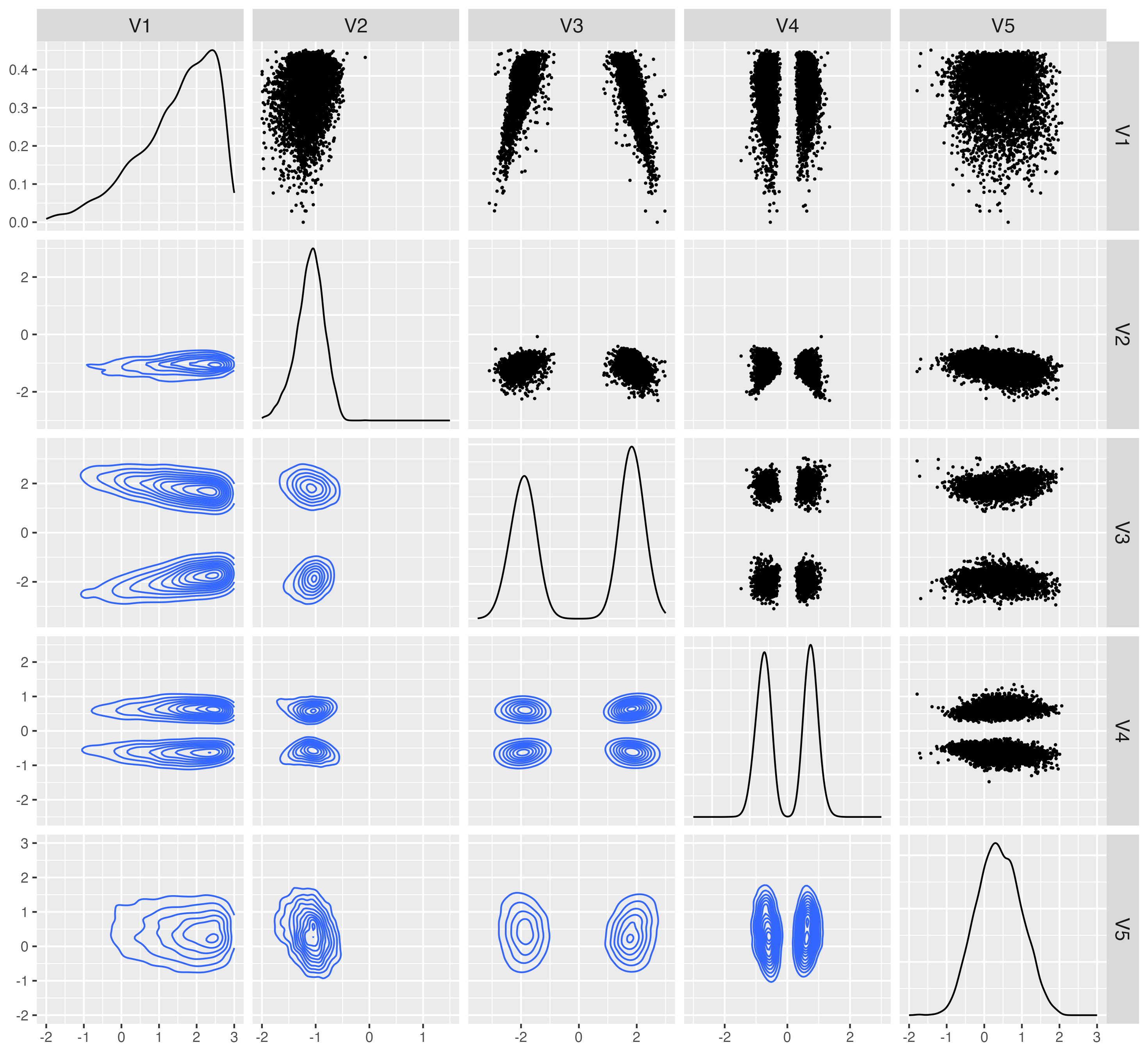}
        \caption{}
        \label{fig:toy-abi}
    \end{subfigure}
    \caption{Bivariate density plot of the posterior distribution.}
    \label{fig:side-by-side-plots}
\end{figure}

\subsection{Lotka-Volterra model}
In the Lotka-Volterra model \citealp{din2013dynamics}, the populations evolve over time based on the following set of equations:
\begin{align*}
    \frac{dx}{dt} &= \alpha x - \beta x y,\\
    \frac{dy}{dt} &= -\gamma y + \delta x y.
\end{align*}

We simulate the Lotka–Volterra process using the Gillespie algorithm \citeapp{gillespie1976general}.  This approach is a stochastic Euler scheme in which time steps are drawn from an exponential distribution.  Algorithm \ref{gillespie-algo} outlines the Gillespie procedure for the Lotka–Volterra model.

\begin{algorithm}[t]
\caption{Gillespie algorithm for Lotka-Volterra model}
\label{gillespie-algo}

\Input{$t = 0, 1, \dots, T$ \newline
$\alpha > 0, \beta > 0, \gamma > 0, \delta > 0$ \newline
$X_0 > 0, Y_0 > 0$}

\While{$(t < T \text{ and } X_t \neq 0 \text{ and } Y_t \neq 0)$}{
    Compute rates: \\
    \Indp
    $r_1 \gets \alpha X_t$\;
    $r_2 \gets \beta X_t Y_t$\;
    $r_3 \gets \gamma Y_t$\;
    $r_4 \gets \delta X_t Y_t$\;
    Draw: $\tau \sim \text{Exp}(R)$, where $R = r_1 + r_2 + r_3 + r_4$\;
    \Indm
    Choose index $i$ from list $\{1, 2, 3, 4\}$\;
    \Indp
    $i \sim \text{Discrete}([r_1/R, ~r_2/R, ~r_3/R, ~r_4/R])$\;
    \Indm
    \If{$i = 1$}{
        $X_{t+1} \gets X_t + 1$\tcp{Prey birth}
        $Y_{t+1} \gets Y_t$\;
    }
    \ElseIf{$i = 2$}{
        $X_{t+1} \gets X_t - 1$\tcp{Prey death}
        $Y_{t+1} \gets Y_t - 1$\;
    }
    \ElseIf{$i = 3$}{
        $X_{t+1} \gets X_t$\;
        $Y_{t+1} \gets Y_t - 1$\tcp{Predator death}
    }
    \ElseIf{$i = 4$}{
        $X_{t+1} \gets X_t$\;
        $Y_{t+1} \gets Y_t + 1$\tcp{Predator birth}
    }
    $t \gets t + \tau$\;
}
\end{algorithm}

\subsection{M/G/1 Queuing Model}

Here we provide additional details on the queuing model setup. Following \citeapp{shestopaloff2014bayesian}, we formulate the model using arrival times $V_i$ as latent variables that evolve as a Markov process:
\begin{align*}
    V_1 &\sim \mathrm{Exp}(\theta_3)\\
    V_i \mid V_{i-1} &\sim V_{i-1} + \mathrm{Exp}(\theta_3), \quad i=2,\dots, n\\
    Y_i \mid X_{i-1}, V_i &\sim \mathrm{Unif}(\theta_1 + \max(0, V_i - X_{i-1}), \theta_2 + \max(0, V_i - X_{i-1})), \quad i = 1, \dots, n.
\end{align*}

Setting $V = (V_1, \dots, V_n)$, the joint density of $V_i$ and $Y_i$ can be factorized as:
\begin{align*}
    \p(V, Y \mid \theta) = \p(V_1 \mid \theta) \prod_{i=2}^n \p(V_i \mid V_{i-1}, \theta) \prod_{i=1}^n \p(Y_i \mid V_i, X_{i-1}, \theta).
\end{align*}

\bibliographystyleapp{plainnat}
\bibliographyapp{refs}

\end{appendices}

\end{document}